\documentclass[preprint,11pt]{article}
\usepackage{amssymb,amsfonts,amsthm,amscd,dsfont,mathrsfs,verbatim,enumerate}
\usepackage[cmex10]{amsmath}
\usepackage[tight,footnotesize]{subfigure}
\usepackage{graphicx,wasysym,textcomp}

\DeclareMathAlphabet{\mathpzc}{OT1}{pzc}{m}{it}

\newtheorem{propo}{Proposition}[section]
\newtheorem{lemma}[propo]{Lemma}
\newtheorem{definition}[propo]{Definition}
\newtheorem{coro}[propo]{Corollary}
\newtheorem{thm}[propo]{Theorem}

\newtheorem{remark}[propo]{Remark}
\newtheorem{conj}[propo]{Conjecture}
\newtheorem{fact}[propo]{Fact}
\newcommand\T{\rule{0pt}{2.6ex}}
\newcommand\B{\rule[-1ex]{0pt}{0pt}}

\def\bb{\,\textup{\textbrokenbar}\, }
\def\Ev{{\sf E}}
\def\parent{P}

\def\cC{{\mathcal C}}

\def\I{\mathcal{I}}

\def\sX{\mathbb{X}}
\def\sY{\mathbb{Y}}
\def\sD{\mathbb{D}}
\def\hsD{\widehat{\mathbb{D}}}
\def\bsD{\breve{\mathbb{D}}}

\def\sZ{\mathbb{Z}}
\def\sU{\mathbb{U}}
\def\bsZ{\breve{\mathbb{Z}}}
\def\bZ{\breve{Z}}
\def\bz{\breve{z}}
\def\bK{\breve{K}}
\def\bY{\breve{Y}}
\def\bX{\breve{X}}

\def\prob{{\mathbb P}}

\def\tY{\widetilde{Y}}

\def\ind{\mathbb{I}}

\def\St{\mathcal{S}}
\def\E{\mathbb{E}}
\def\eps{{\epsilon}}
\def\ve{{\varepsilon}}
\def\naturals{{\mathbb N}}
\def\integers{{\mathbb Z}}
\def\tW{\widetilde{W}}

\def\hW{{\widehat{W}}}
\def\hY{{\widehat{Y}}}
\def\hX{{\widehat{X}}}
\def\tY{{\widetilde{Y}}}

\def\tmu{\widetilde{\mu}}

\def\hD{{\widehat{D}}}
\def\hK{{\widehat{K}}}

\def\const{\kappa}
\def\gt1{ >\hskip-3pt 1}
\def\prob{\mathbb{P}}
\def\t{\mathtt{t}}
\def\tT{\widetilde{T}}
\def\sS{{\mathbb S}}
\def\tc{\tilde{c}}

\newcommand{\xleftrightarrow}[1]{\xleftarrow{\hskip2pt#1}\hskip-7pt\rightarrow}

\newcommand{\reals}{\mathbb{R}}
\newcommand{\tL}{\widetilde{L}}
\newcommand{\tl}{\widetilde{l}}
\newcommand{\tLr}{\widetilde{L}^{\textup{\tiny rep}}}
\newcommand{\tLa}{\widetilde{L}^{\textup{\tiny alt}}}
\def\lr{l^{\textup{\tiny rep}}}

\newcommand{\cH}{\mathcal{H}}
\newcommand{\hcH}{\widehat{\mathcal{H}}}
\newcommand{\cI}{\mathcal{I}}
\newcommand{\asX}{\acute{\mathbb{X}}}
\newcommand{\asY}{\acute{\mathbb{Y}}}
\newcommand{\aX}{\acute{X}}
\newcommand{\aY}{\acute{Y}}
\newcommand{\aq}{\acute{q}}

\begin{document}

\title{Optimal coding for the deletion channel\\ with small deletion probability}

\author{Yashodhan Kanoria\thanks{Department of Electrical Engineering,
Stanford University}
 \;\;\; and\;\;\;  Andrea Montanari\thanks{Department of Electrical
Engineering and Department of Statistics,
Stanford University}}

\maketitle
\begin{abstract}
The deletion channel is the simplest point-to-point  communication
channel that models lack of synchronization. Input bits are deleted
independently with probability $d$, and when they are not deleted,
they are not affected by the channel. Despite
significant effort, little is known about the capacity of this channel,
and even less about optimal coding schemes. In this paper we develop
a new systematic approach to this problem, by demonstrating that
capacity can be computed in a series expansion for small deletion probability.
We compute three leading terms of this expansion, and find an input
distribution that achieves capacity up to this order. This constitutes
the first optimal coding result for the deletion channel.

The key idea employed is the following: We understand perfectly the deletion channel
with deletion probability $d=0$. It has capacity 1 and the optimal input distribution is i.i.d. Bernoulli$(1/2)$. It is natural to expect that the channel with small deletion probabilities
has a capacity that varies smoothly with $d$, and that the optimal input distribution is
obtained by smoothly perturbing the i.i.d. Bernoulli$(1/2)$ process. Our results show that
this is indeed the case. We think that this general strategy can be useful in a number of capacity calculations.
\end{abstract}

%
%

\section{Introduction}
\label{sec:intro}

The (binary) deletion channel accepts bits as inputs,
and deletes each transmitted bit independently with  probability $d$.
Computing or providing systematic approximations to its capacity is
one of the outstanding problems in information theory
\cite{MitzenmacherReview}. An important
motivation comes from the need to understand synchronization errors and
optimal ways to cope with them.

In this paper we suggest a new approach. We demonstrate that capacity
can be computed in a series expansion for small deletion probability,
by computing the first two orders of such an expansion. Our main result
is the following.
\begin{thm}\label{thm:main_theorem}
Let $C(d)$ be the capacity of the deletion channel with deletion
probability $d$. Then, for small $d$ and
any $\eps>0$,
\begin{align}
C(d)=1+ d\log d - A_1\, d+ A_2\, d^2 + O(d^{3-\epsilon})\, ,\label{eq:MainFormula}
\end{align}
where
\begin{align*}
A_1 &\equiv \log(2e)-\sum_{l=1}^\infty 2^{-l-1}l\log l\approx 1.15416377 \\
A_2 &= c_3 +c_4+ \frac{1}{4 \ln 2}  \left( 2+ \frac{3}{2} c_2^2 +
\sum_{l=1}^{\infty}
    2^{-l} \left (l \ln l \right)^2- c_2\sum_{l=1}^{\infty}
    2^{-l}l^2 \ln l \right) \approx 1.67814594\\
c_2 &\equiv \sum_{l=1}^\infty 2^{-l} l \ln l \approx 1.78628364\\
c_3 &\equiv \frac{1}{2} \left ( -1 + \sum_{l=3}^\infty 2^{-l} \left\{\binom{l}{2} \log \binom{l}{2} - l^2 \log l
+(l-1)(l-3) \log (l-1) + (l-2) \log (l-2) \right \} \,\right)\\
&\approx \, -0.88636960 \\
c_4 &\equiv \sum_{j=4}^\infty 2^{-(2+j)}  \,  (j-1 ) (j-3)\, h\!\left( \frac{1}{j-1}\right)\nonumber\\
&\phantom{\equiv \ }+\sum_{i=2}^\infty \sum_{j=4}^\infty  2^{-(i+j+1)}  \, (i+j-1 )(j-3)\, h\!\left( \frac{i+ 1}{i+j-1}\right) \approx 0.69001321
\end{align*}
Here $h(\cdot)$ is the binary entropy function, i.e., $h(p) \equiv -p \log p - (1-p) \log(1-p)$.

Further, the binary stationary source defined by the property that the
times at which it switches from $0$ to $1$ or viceversa form a renewal
process with holding time distribution $p_L(l) = 2^{-l}(1+ d( l \ln l -c_2 l/2))$, achieves
rate within $O(d^{3-\epsilon})$ of capacity.
 \end{thm}

Given a binary sequence, we will call `runs' its maximal blocks of
contiguous  $0$'s or $1$'s. We shall refer to binary sources such that the switch times form a renewal process
as \emph{sources (or processes) with i.i.d. runs}.

The `rate' of a given binary source is the maximum rate at which
information can be transmitted through the deletion channel using
input sequences distributed as the source.
A formal definition is provided below (see Definition \ref{def:rate}). Logarithms denoted by $\log $ here
(and in the rest of the paper) are understood to be
in base $2$.
While one might be skeptical about the concrete meaning
of asymptotic expansions of the type (\ref{eq:MainFormula}),
they often prove surprisingly accurate. For instance at $d=0.1$
($10\%$ of the input symbols are deleted),
the expression in Eq.~(\ref{eq:MainFormula}) (dropping the error term
$O(d^{3-\eps})$) is larger than the best lower bound
\cite{Drinea07lb} by about $0.007$ bits. The lower bound of
\cite{Drinea07lb} is derived using a Markov source and `jigsaw'
decoding. Our asymptotic analysis implies that the loss in rate due to
restricting to Markov sources and jigsaw decoding (cf. Theorem
\ref{thm:markov_loss} and Remark \ref{rem:lb_off_pt9dsq}), to leading
order, is $0.904d^2 \approx 0.009$.
Hence, we estimate that our asymptotic approach incurs an error of
about $0.002$ bits for computing the capacity at  $d=0.1$.

More importantly asymptotic expansions can provide useful design
insight. Theorem \ref{thm:main_theorem}
shows that the stationary process consisting of i.i.d. runs with
the specified run length distribution, achieves capacity to within
$O(d^{3-\epsilon})$. In comparison, the best performing approach
tried before this was to use a first order Markov source for coding \cite{Drinea07lb}.
We are able to show, in fact, that this approach incurs a loss that is
$\Omega(d^2)$, which is the same order as the loss incurred by the trivial
approach of using i.i.d. Bernoulli$(1/2)$!

\begin{remark}
\label{rem:constants_not_computed}
In this work, we prove rigorous upper and lower bounds on capacity
that match up to  quadratic order in $d$ (cf. Theorem
\ref{thm:main_theorem}), but without explicitly evaluating the
constants in the error terms.
It would be very interesting to obtain explicit expressions for these
constants.
\end{remark}

Before this work, there was no non-trivial optimal coding result known for
the deletion channel\footnote{The trivial exception
is the case $d=0$, for which the  i.i.d. Bernoulli$(1/2)$
process achieves capacity.}. Further terms in the capacity expansion can be expected
to supply even more detailed information about the optimal coding scheme and
allow us to achieve capacity to higher orders.
%

We think that the strategy adopted here might be useful in other
information theory problems. The underlying philosophy is
that whenever capacity is known for a specific value of
the channel parameter, and the corresponding
optimal input distribution is unique and well characterized,
it should be possible to compute an asymptotic expansion around that value.
In the present context the special channel is the perfect channel,
i.e. the deletion channel with deletion probability $d=0$. The
corresponding input distribution is the i.i.d. Bernoulli$(1/2)$ process.
%
%
\subsection{Related work}

Dobrushin \cite{Dobrushin} proved a coding theorem for the deletion channel,
and other channels with synchronization errors. He showed that the maximum
rate of reliable communication is given by the  maximal mutual information per
bit, and proved that this can be achieved through a random coding scheme.
This characterization has so far found limited use in proving concrete
estimates.
An important exception is provided by the work of Kirsch and Drinea
\cite{KirschDrinea} who use Dobrushin coding theorem to prove
lower bounds on the capacity of channels with deletions and
duplications.
We will also use Dobrushin theorem in a crucial way, although
most of our effort will be devoted to proving upper bounds on the capacity.

Several capacity bounds have been developed over the last few years,
following alternative approaches, and are surveyed in
\cite{MitzenmacherReview}.
In particular, it has been proved that $C(d)=\Theta(1-d)$
as $d\to 1$ \cite{Drinea06dnear1}.
The papers \cite{Fertonani09,Dalai} improve the upper bound in this
limit obtaining $\limsup_{d\to 1}C(d)/(1-d)\leq 0.413$. However, determining the asymptotic behavior in
this limit (i.e. finding a constant $B_1$ such that
$C(d) = B_1(1-d)+o(1-d)$) is an open problem.
When applied to the small $d$ regime, none of the known upper bounds
actually captures the correct behavior as stated in Eq.~(\ref{eq:MainFormula}).
A simple calculation shows that the first upper bound in \cite{DMP07upperbounds}
has asymptotics of $1+(3/4)d\log d$. Another work \cite{Fertonani09} shows that
$C \geq 1- 4.19 d$ as $d \rightarrow \infty$.
As we show in the present paper, this behavior can be
controlled exactly, up to the third leading term
of the expansion.

A short version of this paper  was presented at the 2010 International
Symposium on Information Theory (ISIT) \cite{KanoriaMontanari}.
At the same conference,  Kalai, Mitzenmacher and Sudan \cite{KMS}
presented a result analogous to Theorem \ref{thm:main_theorem}.
The proof is based on a counting argument, very different from the
the techniques employed here.
Also, the result of \cite{KMS} is not the same as in Theorem
\ref{thm:main_theorem}, since only the $d\log d$ term of the series
is established in \cite{KMS}. Theorem
\ref{thm:main_theorem} improves on our ISIT result \cite{KanoriaMontanari}, that contained only the first two terms in the
series expansion, but not the order $d^2$ term. Also, we obtain a non-trivial coding
scheme for the first time in this paper. The trivial i.i.d. Bernoulli$(1/2)$ coding
scheme is enough to achieve capacity up to linear order as shown in our
conference paper \cite{KanoriaMontanari}.

%

\subsection{Numerical illustration of results}

We can numerically evaluate the expression in Eq.~(\ref{eq:MainFormula}) (dropping the error term) to obtain estimates of capacity for small deletion probabilities.
\begin{align*}
C_{\textup{\tiny est}} = 1+ d\log d - A_1\, d+ A_2\, d^2\, .
\end{align*}

The values of $C_{\textup{\tiny est}}$ are presented in Table \ref{table:capacity_bounds_and_expansion} and Figure \ref{fig:capacity_bounds_and_expansion}. We compare with the best known numerical lower bounds \cite{Drinea07lb} and upper bounds \cite{Fertonani09,DMP07upperbounds}.

We stress here that $C_{\textup{\tiny est}}$ is \emph{neither an upper nor a lower bound on capacity}.
It is an estimate based on taking the leading terms of the asymptotic expansion of capacity for small $d$, and is expected to be accurate for small values of $d$. Indeed, we see that for $d$ larger than $0.4$, our estimate $C_{\textup{\tiny est}}$ \emph{exceeds} the upper bound. This simply indicates that we should not use $C_{\textup{\tiny est}}$ as an estimate for such large $d$. We believe that $C_{\textup{\tiny est}}$ provides an excellent estimate of capacity for $d \apprle 0.2$.

\small
\begin{table}
\begin{center}
\begin{tabular}{c|c|c|c}
$d$& Best lower bound & $C_{\textup{\tiny est}}$ & Best upper bound\\[5pt]
\hline
 0.05 &    0.7283  &  0.7304 &   0.8160\T\\
 0.10 &    0.5620  &  0.5692 &   0.6890\\
 0.15 &    0.4392  &  0.4541 &   0.5790\\
 0.20 &    0.3467  &  0.3719 &   0.4910\\
 0.25 &    0.2759  &  0.3163 &   0.4200\\
 0.30 &    0.2224  &  0.2837 &   0.3620\\
 0.35 &    0.1810  &  0.2715 &   0.3150\\
 0.40 &    0.1484  &  0.2781 &   0.2750\\
 0.45 &    0.1229  &  0.3020 &   0.2410\\
 0.50 &    0.1019  &  0.3425 &   0.2120\B
\end{tabular}
\end{center}
\caption{Table showing best known numerical bounds on capacity (from \cite{Drinea07lb,Fertonani09,DMP07upperbounds}) compared with our estimate based on the small $d$ expansion.}\label{table:capacity_bounds_and_expansion}
\end{table}

\begin{figure}[htp]
\centering
\includegraphics[scale=0.9]{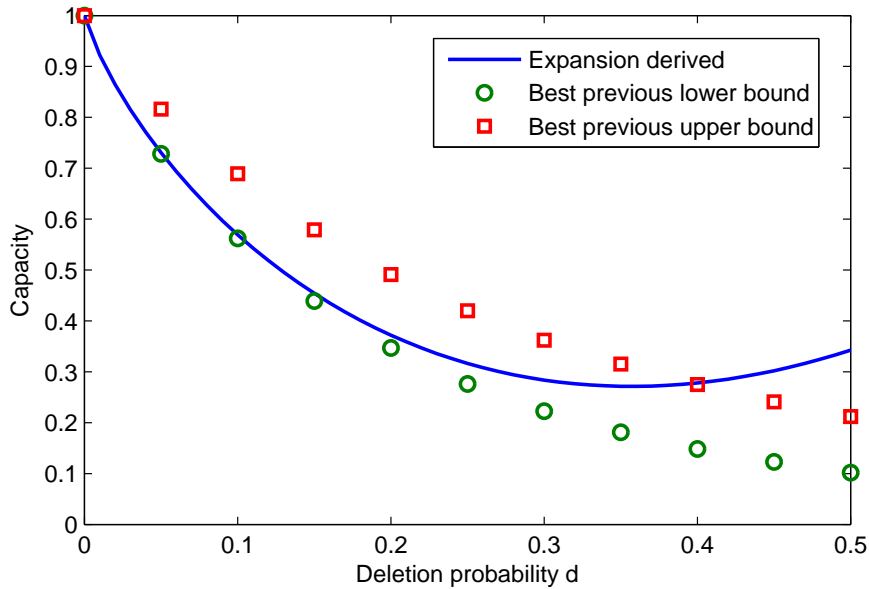}
\caption{Plot showing best known numerical bounds on capacity (from \cite{Drinea07lb,Fertonani09,DMP07upperbounds}) compared with our estimate based on the small $d$ expansion.}\label{fig:capacity_bounds_and_expansion}
\end{figure}

\subsection{Notation}
We borrow $O(\cdot)$, $\Omega(\cdot)$ and $\Theta(\cdot)$ notation from the computer science
literature. We define these as follows to fit our needs.
Let $f: [0,1] \rightarrow \mathbb{R}$ and $g: [0,1] \rightarrow \mathbb{R_+}$. We say:
\begin{itemize}
\item We say $f=O(g)$ if there is a constant $c < \infty$ such that $|f(x)| \leq c g(x)$ for all
$x \in [0,1]$.
\item We say $f=\Omega(g)$ if there is a constant $c >0$ such that $f(x) \geq c g(x)$ for all
$x \in [0,1]$.
\item We say $f=\Theta(g)$ if there are constants $c < \infty \,,\, c' >0$ such that $c g(x) \geq f(x) \geq c' g(x)$ for all
$x \in [0,1]$.
\end{itemize}
Throughout this paper, we adhere to the convention that the constants $c, c'$ above should
not depend on the processes $\sX, \sY, \ldots$ etc. under consideration, if there are such
processes.

\subsection{Outline of the paper}

Section \ref{sec:Preliminaries} contains the basic definitions and
results necessary for our approach to estimating the capacity of the
deletion channel. We show that it is sufficient to consider stationary ergodic input sources, and
define their corresponding rate (mutual information per bit). Capacity
is obtained by maximizing this quantity over stationary processes.
In Section \ref{sec:intuition}, we present an informal argument that
contains the basic intuition leading to our main result (Theorem
\ref{thm:main_theorem}), and allows us to correctly guess the optimal input
distribution.
Section \ref{sec:outline} states a small number of core lemmas, and
shows that they imply  Theorem
\ref{thm:main_theorem}.
Finally, Section \ref{sec:ProofLemma} states several technical results
(proved in appendices) and uses them to prove the core lemmas. We conclude with a short discussion, including a  list of open problems, in Section \ref{sec:discussion}.

\section{Preliminaries}
\label{sec:Preliminaries}

For the reader's convenience, we restate here some known results
that we will use extensively, along with some definitions
and auxiliary lemmas.

Consider a sequence of channels $\{W_n\}_{n\ge 1}$, where $W_n$ allows exactly
$n$ inputs bits, and deletes each bit independently with probability $d$.
The output of $W_n$ for input $X^n$ is a binary vector denoted by $Y(X^n)$.
The length of $Y(X^n)$ is a binomial random variable.
We want to find maximum rate at which we can send information
over this sequence of channels
with vanishingly small error probability.

The following characterization follows from \cite{Dobrushin}.
\begin{thm}\label{thm:cap_limit}
Let
\begin{align*}
C_n \equiv \frac{1}{n}\max_{p_{X^n}}\, I(X^n; Y(X^n))\, .
\end{align*}
Then, the following limit exists
\begin{align}
C \equiv \lim_{n\rightarrow \infty}C_n = \inf_{n\ge 1} C_n\, ,
\label{eq:cap_defined}
\end{align}
and is equal to the capacity of the deletion channel.
\end{thm}

A further useful remark is that, in computing capacity, we can
assume $(X_1,\dots,X_n)$ to be $n$ consecutive coordinates of a stationary
ergodic process. We denote by $\St$ the class of stationary and ergodic processes
that take binary values.

\begin{lemma}\label{lemma:stationary_suffices}
Let $\sX= \{X_i\}_{i\in\integers}$ be a stationary and ergodic
process, with $X_i$ taking values in $\{0,1\}$.  Then the limit
$I(\sX)=\lim_{n \rightarrow \infty} \frac{1}{n}I(X^n; Y(X^n))$ exists and
\begin{align*}
C=\max_{\sX \in \St} I(\sX)\,.
\end{align*}
\end{lemma}
We use the following natural definition of the \emph{rate} achieved by a stationary ergodic process.
\begin{definition}\label{def:rate}
For stationary and ergodic $\sX$, we call $I(\sX)=\lim_{n \rightarrow \infty} \frac{1}{n}I(X^n; Y(X^n))$ the \emph{rate} achieved by $\sX$.
\end{definition}
Proofs of Theorem \ref{thm:cap_limit} and Lemma \ref{lemma:stationary_suffices} are provided in Appendix \ref{app:prelim_results}.

Given a stationary process $\sX$, it is convenient to consider it 
%
from the point of view of a `uniformly
random' block/run. Intuitively, this corresponds to choosing a large
integer $n$ and selecting as reference point the beginning of a
uniformly random block in $X_1,\dots,X_n$. Notice that this approach
naturally discounts longer blocks for finite $n$. While such a procedure can
be made rigorous by taking the limit $n\rightarrow \infty$, it is more convenient to make use of the notion of \emph{Palm measure} from the theory of point processes \cite{PP1,PP2},
which is, in this case, particularly easy to define.
To a binary source $\sX$, we can associate in a bijective way 
a subset of times $\sS\subseteq \sZ$, by letting $t\in \sS$ if and
only if $X_t$ is the first bit of a run.
The Palm measure $\prob_{1}$ is then the distribution of $\sX$
conditional on the event $1\in \sS$.

We denote by $L$ the length of the block starting at $1$
under the Palm measure, and denote by $p_{L}$ its distribution.
As an example, if $\sX$  is the i.i.d. Bernoulli$(1/2)$ process,
we have $p_{L}= p_L^*$ where $p_L^*(l) \equiv 2^{-l}$.
We  will also call   $p_{L}$ the \emph{block-perspective run
length distribution} or simply the \emph{run length distribution}, and let
\begin{eqnarray*}
\mu(\sX) \equiv \E \sum_{l=1}^{\infty}p_{L}(l) \, l\, ,
\end{eqnarray*}
be its average. 
Let $L_0$ be the length of the block containing bit $X_0$ in
the stationary process $\sX$. A standard calculation\cite{PP1,PP2}  yields
$
\prob(L_0= l)=lp_{L}(l)/\mu(\sX) \, 
$.
Since $L_0$ is a well defined and almost surely finite (by
ergodicity), we necessarily have $\mu(\sX) < \infty$.

In our main result, Theorem \ref{thm:main_theorem}, a special role is
played by processes $\sX$ such that the associated switch times form a
stationary renewal process. We will refer to such an $\sX$ as 
\emph{a process with i.i.d. runs.}
%
%
\section{Intuition behind the main theorem}
\label{sec:intuition}
In this section, we provide a heuristic/non-rigorous explanation for our main result. The aim is build intuition and motivate our approach, without getting bogged down with the numerous technical difficulties that arise. In fact, we focus here on heuristically deriving the optimal input process $\sX^\dagger$, and do not actually obtain the quadratic term of the capacity expansion. We find $\sX^\dagger$ by computing various quantities to leading order and using the following observation (cf. Remark \ref{rem:h_within_1minuseps}).

\vspace{0.1cm}

{\bf Key Observation:}
\emph{The process that achieves capacity for small $d$ should be
  `close' to the Bernoulli$(1/2)$ process, since $H(\sX)$ must be close to $1$.}

\vspace{0.1cm}

We have
\begin{align}
I(X^n ; Y(X^n)) = H(Y) -H(Y|X^n) \, .
\label{eq:ixy_intuitive}
\end{align}
Let $D^n$ be a binary vector containing a one at position $i$ if and
only if $X_i$ is deleted from the input vector.
We can write
\begin{align*}
H(Y|X^n) =  H(Y, D^n | X^n) - H(D^n| X^n, Y)\, .
\end{align*}
But $Y$ is a function of $(X^n, D^n)$, leading to $H(Y, D^n | X^n) = H(D^n | X^n) = H(D^n) =n h(d)$, where we used the fact that $D^n$ is i.i.d. Bernoulli($d$), independent of $X^n$.
It follows that
\begin{align}
H(Y|X^n) =  n h(d) - H(D^n| X^n, Y)\, .
\label{eq:hygivenx_intuitive}
\end{align}

The term $H(D^n| X^n, Y)$ represents ambiguity in the location of deletions, given the input and output strings. Now, since $d$ is small, we expect that most deletions occur in
`isolation', i.e., far away from other deletions. Make the (incorrect) assumption that all deletions occur such that no three consecutive runs have more than one deletion in total. In this case, we can unambiguously associate runs in $\sY$ with runs in $\sX$. Ambiguity in the location of a deletion occurs if and only if a deletion occurs in a run of length $l>1$. In this case, each of $l$ locations is equally likely for the deletion,  leading to a contribution of
$\log l$ to $H(D^n| X^n, Y)$. Now, a run of length $l$ should suffer a deletion with probability $\approx ld$. Thus, we expect
\begin{align*}
\frac{1}{n} \, H(D^n| X^n, Y) \approx \frac{d}{\mu(\sX)}
\sum_{l=1}^\infty p_L(l) l \log l\, .
\end{align*}
We know that $H(\sX)$ is close to $1$, implying $\mu(\sX)$ is close to $2$ and
$p_L$ is close to $ p_L^*(l) \equiv 2^{-l}$. This leads to
\begin{align}
\frac{1}{n} \, H(D^n| X^n, Y) &\approx \frac{d}{2} \sum_{l=1}^\infty p_L(l) l \log l
+ \frac{d(\mu(\sX)-2)}{4} \sum_{l=1}^\infty p_L^*(l) l \log l \nonumber \\
&= \frac{d}{2} \left [ -\frac{c_2}{ \ln 2} + \sum_{l=1}^\infty p_L(l) l \left ( \log l - \frac{c_2}{2 \ln 2} \right) \right]\,.
\label{eq:hd_givenxy_intuitive_bd}
\end{align}


Consider $H(Y)$. Now, if the input $X^n$ is drawn from a stationary process $\sX$, we expect the output $Y(X^n)$ to also be a segment of some stationary process $\sY$. (It turns out that this is the case.) Moreover, we expect that the channel output has $n(1-d)+o(n)$ bits, leading to $H(Y) \approx n(1-d) H(\sY)$. Denote the run length distribution in $\sY$ by $q_L(\cdot)$. Define $\mu(\sY) \equiv \sum_{l=1}^\infty q_L(l) l$. Let $L_{\sY}$ denote the length of a random run drawn according to $q_L(\cdot)$. It is not hard to see that
\begin{align*}
H(\sY) \leq H(L_{\sY})/\mu(\sY)\, ,
\end{align*}
with equality iff $\sY$ consists of i.i.d. runs, which occurs iff $\sX$ consists of i.i.d. runs.
Define $q_L^*(l) \equiv 2^{-l}$. An explicit calculation yields $H(L_{\sY}) = 1 - D(q_L|| q_L^*)/\mu(\sY)$. We know that $H(\sY)$ is close to $1$, implying $\mu(\sY)$ is close to $2$ and $D(q_L|| q_L^*)$ is small.
Thus,
\begin{align*}
\lim_{n \rightarrow \infty} \frac{1}{n} \, H(Y)= (1-d)H(\sY) \leq (1-d) (1 - D(q_L|| q_L^*)/\mu(\sY))
\approx 1 - d - D(q_L|| q_L^*)/2 \, .
\end{align*}

Notice  that an i.i.d. Bernoulli$(1/2)$ input results in an i.i.d. Bernoulli$(1/2)$ output from the deletion channel.
The following is made precise in Lemma \ref{lemma:xy_similar_runs}: Let $\Delta$ be the `distance' between $p_L$ and $p_L^*$. Then a short calculation tells us that the distance between $p_L$ and $q_L$ should be $O(d^{1-\eps}\Delta) $. In other words $p_L$ and $q_L$ are very nearly equal to each other.

So we obtain, to leading order,
\begin{align}
\lim_{n \rightarrow \infty} \frac{1}{n} \, H(Y)
\apprle 1 - d - D(p_L|| p_L^*)/2 \, ,
\label{eq:hy_intuitive_bd}
\end{align}
with (approximate) equality iff $\sX$ consists of i.i.d. runs.

Putting Eqs.~\eqref{eq:ixy_intuitive}, \eqref{eq:hygivenx_intuitive}, \eqref{eq:hd_givenxy_intuitive_bd} and \eqref{eq:hy_intuitive_bd} together, we have
\begin{align*}
I(\sX) &= \lim_{n \rightarrow \infty} \frac{1}{n} I(X^n; Y) \nonumber \\
&\apprle 1-d - D(p_L|| p_L^*)/2  - h(d) + \frac{d}{2} \left [ -\frac{c_2}{ \ln 2} + \sum_{l=1}^\infty p_L(l) l \left ( \log l - \frac{c_2}{2 \ln 2} \right) \right ] \nonumber \\
& \approx 1 - d\log(1/d) - A_1 d - \frac{1}{2}  D(p_L|| p_L^*)
+ \frac{d}{2} \left [\sum_{l=1}^\infty p_L(l) l \left ( \log l - \frac{c_2}{2 \ln 2} \right) \right ]\, .
\end{align*}

Since this (approximate) upper bound on $I(\sX)$ depends on input $\sX$ only through $p_L$, we choose
$\sX$
consisting of i.i.d. runs so that (approximate) equality holds.

We expect $p_L$ to be close to $p_L^*(l)$. A Taylor expansion gives
\begin{align*}
D(p_L|| p_L^*) &= \sum_{l=1}^\infty p_L(l) (l+\log p_L(l)) \\
& \approx  \frac{1}{\ln2}\sum_{l=1}^{\infty} \bigg ( \left(p_L(l) - 2^{-l}\right ) + 2^{l-1} \left (p_L(l) - 2^{-l}\right )^2 \bigg )\\
& =  \frac{1}{\ln2}\sum_{l=1}^{\infty} 2^{l-1} \left (p_L(l) - 2^{-l}\right )^2\, .
\end{align*}

Thus, we want to maximize
\begin{align*}
\frac{1}{2 \ln2}\sum_{l=1}^{\infty} 2^{l-1} \left (p_L(l) - 2^{-l}\right )^2
+ \frac{d}{2} \left [\sum_{l=1}^\infty p_L(l) l \left ( \log l - \frac{c_2}{2 \ln 2} \right) \right ]\, ,
\end{align*}
subject to $\sum_{l=1}^\infty p_L(l) =1$, in order to achieve the largest possible $I(\sX)$. A simple calculation tells us that the maximizing distribution is $p_L^\dagger(l) = 2^{-l}(1+ d( l \ln l -c_2 l/2))$.

\section{Proof of the main theorem: Outline}
\label{sec:outline}

In this section we provide the proof of Theorem \ref{thm:main_theorem}
after stating the key lemmas involved.
We defer the proof of the lemmas to the next section.
Sections \ref{subsec:run_charac}-\ref{subsec:self_improving_bd_hy} develop
the technical machinery we use, and the proofs of the lemmas are in
Section \ref{subsec:lemma_proofs}.

Given a (possibly infinite) binary sequence,
a \emph{run} of $0$'s (of $1$'s) is a maximal subsequence of consecutive $0$'s
($1$'s), i.e. an subsequence of $0$'s bordered by $1$'s
(respectively, of $1$'s bordered by $0$'s). The first step consists in proving achievability by
estimating $I(\sX)$ for a process having i.i.d. runs with appropriately chosen distribution.
\begin{lemma}\label{lemma:achievability}
Let $\sX^\dagger$ be the process consisting of i.i.d.  runs with
distribution $p_L^\dagger(l) = 2^{-l}(1+ d( l \log l -c_2l/2))$.
Then for any $\eps >0$, we have
\begin{align*}
I(\sX^\dagger) = 1+ d\log d - A_1\, d + A_2\, d^2 + O(d^{3-\eps})\, .
\end{align*}
\end{lemma}
Lemma \ref{lemma:achievability} is proved in Section \ref{subsec:lemma_proofs}.

Lemma \ref{lemma:stationary_suffices} allows us to restrict our
attention to stationary ergodic processes in proving the converse.
For a process
$\sX$, we denote by $H(\sX)$ its \emph{entropy rate}. Define
\begin{align}
\label{eq:hyx_defined}
H(Y_{\sX}) \equiv \lim_{n \rightarrow \infty} \frac{H(Y(X^n))}{n(1-d)}\, .
\end{align}
A simple argument shows that this limit exists and is bounded above by $1$ for any
stationary process $\sX$ and any $d$, with $H(Y_{\sX})=1$ iff $\sX$ is the i.i.d. Bernoulli($1/2$)
process.

In light of Lemma \ref{lemma:achievability}, we can restrict
consideration to processes $\sX$ satisfying $I(\sX) > 1-d^{1-\epsilon}$
whence $H(\sX) > 1-d^{1-\epsilon}, H(Y_{\sX}) > 1-d^{1-\epsilon}$:
\begin{remark}\label{rem:h_within_1minuseps}
There exists $d_0(\eps)>0$ such that for all $d < d_0(\eps)$, if $I(\sX) > C - d$, we
have $I(\sX) > 1-d^{1-\epsilon}$ and hence also
$H(\sX) > 1-d^{1-\epsilon}\, ,\ H(Y_{\sX}) > 1-d^{1-\epsilon}$\,.
\end{remark}

We define a `super-run' next.
\begin{definition}
\label{def:super-run}
A \emph{super-run} consists of a maximal contiguous sequence of runs such that all runs in the sequence after the first one (on the left) have length one.
We divide a realization of $\sX$ into \emph{super-runs} $\ldots, S_{-1}, S_0, S_1, \ldots\,$.
Here $S_1$ is the super-run including the bit at position 1.
\end{definition}

See Table \ref{table:super-run_example} for an example showing division into super-runs.

\begin{table}[t]
\begin{center}
\begin{tabular}{cc|ccc|ccc|cccc|ccc}
$\ldots$ & $b_{-4}$ & $b_{-3}$ & $b_{-2}$ & $b_{-1}$ & $b_{0}$ & $b_{1}$ & $b_{2}$ & $b_{3}$ & $b_{4}$ & $b_{5}$ & $b_{6}$ & $b_{7}$ & $b_{8}$ & $\ldots$\\
$\ldots$ & 1 & 0 & 0 & 1 & 0 & 0 & 0 & 1 & 1 & 0 & 1 & 0 & 0 & $\ldots$
\end{tabular}
\end{center}
\caption{An example showing how $\sX$ is divided into super-runs}\label{table:super-run_example}
\end{table}

Denote by $\St$ the set of all stationary ergodic processes and
by $\St_{L^*}$ the set of stationary ergodic processes such
that, with probability one, no super-run has length larger than $L^*$.

Our next lemma tightens
the constraint given by Remark \ref{rem:h_within_1minuseps} further for processes in $\St_{\lfloor 1/d \rfloor}$.
\begin{lemma}\label{lemma:hy_is_large}
Consider  any $\epsilon > 0$ and constant $\const$.
There exists $d_0(\epsilon, \const)>0$ such that the following
happens for any $\sX \in \St_{\lfloor 1/d \rfloor}$.
For any $d< d_0$, if
\begin{align*}
I(\sX) &\geq C - \const d^{2-(\epsilon/2)} \, ,
\end{align*}
then
\begin{align*}
H(Y_{\sX}) \geq 1 - d^{2-\epsilon} \, .
\end{align*}
\end{lemma}
We show an upper bound for the restricted class of processes $\St_{L^*}$.
\begin{lemma}\label{lemma:converse_for_restricted_runs}
For any $\eps>0$ there exists  $d_0=d_0(\eps)>0$ and $\const < \infty$
such that the following happens.
If $d < d_0(\epsilon)$, for any $\sX \in \St_{\lfloor 1/d \rfloor}$,
\begin{align*}
I(\sX) \leq 1 + d\log d - A_1d + A_2 d^2 + \const d^{3-\epsilon}\, .
\end{align*}
\end{lemma}
Finally, we show a suitable reduction from the class $\St$ to the class
$\St_{L^*}$.
\begin{lemma}\label{lemma:small_loss_by_restricting_runs}
For any $\eps>0$ there exists  $d_0=d_0(\eps)>0$
such that the following happens for all $d < d_0$, and all $\gamma>0$.
For any $\sX \in \St$ such that $H(Y_{\sX}) > 1 - d^\gamma$ and for any
$L^* > 2\gamma \log (1/d)$, there exists
$\sX_{L^*} \in \St_{L^*}$  such that
\begin{align}
I(\sX) &\leq I(\sX_{L^*}) + d^{\gamma-\epsilon}(L^*)^{-1}\log L^*\, ,\label{eq:I_vs_restricted}\\
H(Y_{\sX}) &\geq H(Y_{\sX_{L^*}}) - d^{\gamma-\epsilon}(L^*)^{-1}\log L^*\, . \label{eq:Hy_vs_restricted}
\end{align}
\end{lemma}

Lemmas \ref{lemma:hy_is_large}, \ref{lemma:converse_for_restricted_runs} and \ref{lemma:small_loss_by_restricting_runs}
are proved in Section \ref{subsec:lemma_proofs}.

The proof of Theorem \ref{thm:main_theorem} follows from these lemmas with Lemma
\ref{lemma:small_loss_by_restricting_runs} being used twice.
\begin{proof}[Proof of Theorem \ref{thm:main_theorem}]
Lemma \ref{lemma:achievability} shows achievability. For the converse, we start with a process $\sX  \in \St$
such that $I(\sX) > C - d^3$. By Remark \ref{rem:h_within_1minuseps},
$H(Y_{\sX}) > 1 - d^{1- \delta}$ for any $\delta>0$
and $d<d_0(\delta)$. Use
Lemma \ref{lemma:small_loss_by_restricting_runs}, with
$\gamma = 1 - \delta$, $L^* = \lfloor 1/d \rfloor$ and $\eps = \delta/2$.
It follows that for  $d<d_0(\delta/2)$,
\begin{align*}
I(\sX_{L^*}) &> C - d^{2-2\delta}\, ,\\
H(Y_{\sX}) &\geq H(Y_{\sX_{L^*}}) - d^{2-2\delta}\, .
\end{align*}
We now use Lemma \ref{lemma:hy_is_large} which yields $H(Y_{\sX_{L^*}}) \geq 1 - d^{2-2\delta}$
and hence, by Eq.~(\ref{eq:Hy_vs_restricted}), $H(Y_{\sX}) \geq 1 - 2 d^{2-2\delta} \geq 1 - d^{2-3\delta}$ for small $d$. Now, we can
use Lemma \ref{lemma:small_loss_by_restricting_runs} again with
$\gamma = 2 - 3\delta$,  $L^* = \lfloor 1/d \rfloor$, $\eps =\delta/2$. We obtain
\begin{align*}
I(\sX_{L^*}) \geq C - d^{3-4\delta}\, .
\end{align*}
Finally, using Lemma \ref{lemma:converse_for_restricted_runs},
 we get the required upper bound on $C$.
\end{proof}
%
%
\section{Proofs of the Lemmas}
\label{sec:ProofLemma}

In Section \ref{subsec:run_charac} we show that, for  any stationary
ergodic $\sX$ that achieves a rate close to capacity, the run-length distribution must be
close to the distributions obtained for the i.i.d. Bernoulli$(1/2)$ process.
In Section \ref{subsec:information_rate}, we suitably rewrite the
rate $I(\sX)$ achieved by stationary ergodic process $\sX$ as the sum
of three terms.
In Section \ref{subsec:modified_deletion} we construct a modified
deletion process that allows accurate estimation of $H(Y|X^n)$ in
the small $d$ limit. Section \ref{subsec:self_improving_bd_hy} proves
a key bound on $H(Y_{\sX})$ that leads directly to Lemma \ref{lemma:hy_is_large}.
Finally, in Section \ref{subsec:lemma_proofs}
we present proofs of the Lemmas quoted in Section \ref{sec:outline}
using the tools developed.

We will often write $X_{a}^b$ for the random vector
$(X_a,X_{a+1},\dots, X_b)$ where the $X_i$'s are distributed according
to the process $\sX$.
%
%
\subsection{Characterization in terms of runs}
\label{subsec:run_charac}

Let $m_n$ be the number of runs in $X^n$.
Let $L_1^+,L_2,\dots, L_{m_n}$ be the run lengths ($L_1^+$ being the length of the intersection of that run with $X^n$). It is clear that $H(X^n) \le 1+H(m_n, L_1^+,L_2,\dots,L_{m_n})$
(where one bit is needed to remove the $0,1$ ambiguity).
By ergodicity $m_n/n\to 1/\E[L]$ almost surely as $n\to\infty$.
Also $m_n\le n$  implies $H(m_n)/n \leq \log n /n \to 0$.
Further, 
$\limsup_{n\rightarrow \infty} H(L_1^+,L_2,\dots,L_{m_n})/n \le \lim_{n\rightarrow \infty}
H(L) m_n/n = H(L)/\E[L] $.
If $H(\sX)$ is the entropy rate of the process $\sX$,
by taking the $n\to\infty$ limit, it is easy to deduce that
\begin{align}
H(\sX) \leq \frac{ H(L) }{\E[L]} \, ,
\label{eq:run_hx_upper_bd}
\end{align}
with equality if and only if $\sX$ is a process with i.i.d. runs with common
distribution $p_L$.

We know that given $\E[L]=\mu$, the probability distribution
with largest possible entropy $H(L)$ is geometric with mean $\mu$, i.e.
$p_L(l) = (1-1/\mu)^{l-1}1/\mu $ for all $ l \geq 1$, leading to
\begin{align}
\frac{H(L)}{\E[L]} \leq  -\big(1 - \frac{1}{\mu}\big) \log
\big(1 - \frac{1}{\mu}\big) -
\frac{1}{\mu}\log \frac{1}{\mu} \equiv h(1/\mu) \, .
\label{eq:BoundMu}
\end{align}
Here we introduced the notation $h(p) = -p \log p - (1-p) \log(1-p)$
for the binary entropy function.

Using this, we are able to obtain sharp bounds on $p_L$ and
$\mu(\sX)$.
\begin{lemma}
 There exists $d_0>0$
such that the following occurs. For  any $\beta > 1/2$ and $d < d_0$, if $\sX\in\St$ is such that $H(\sX) > 1 - d^{\beta}$, we have
\begin{align}
|\mu(\sX) - 2 | \leq 7 \, d^{\beta/2} \, .
\end{align}
\label{lemma:mean_closeto2}
\end{lemma}
\begin{proof}
By Eqs. (\ref{eq:run_hx_upper_bd}) and (\ref{eq:BoundMu}),
we have $h(1/\mu)\ge 1- d^\beta$. By  Pinsker's inequality
$h(p)\le 1-(1-2p)^2/(2\ln 2)$, and therefore
$|1-(2/\mu)|^2\le (2\ln 2)d^\beta$.
The claim follows from simple calculus.
\end{proof}

\begin{lemma}\label{lemma:L_TV}
There exists $d_0>0$ and $\const'< \infty$
such that the following occurs for  any $\beta > 1/2$ and $d < d_0$.
For any  $\sX\in\St$  such that $H(\sX) > 1 - d^{\beta}$, we have
\begin{align}
\sum_{l=1}^\infty \left|p_L(l) - \frac{1}{2^l} \right| \leq
\const' d^{\beta/2}\, .
\label{eq:L_TV}
\end{align}
\end{lemma}
\begin{proof}
Let $p_L^*(l) = 1/2^l, \ l \geq 1$ and recall that
$\mu(\sX) =\E[L]=\sum_{l\ge 1}
p_L(l)l$. An explicit calculation yields
\begin{align}
H(L)= \mu(\sX)- D(p_L || p_L^*) \, .
\label{eq:Lentropy}
\end{align}
Now, by Pinsker's inequality,
\begin{align}
D(p_L || p_L^*) \geq \frac{2}{\ln 2}\lVert p_L-p_L^*\rVert_{\rm TV}^2\, .
\label{eq:pinsker}
\end{align}
Combining Lemma \ref{lemma:mean_closeto2},
and Eqs.~(\ref{eq:run_hx_upper_bd}), (\ref{eq:Lentropy})
and (\ref{eq:pinsker}), we
get the desired result.
\end{proof}

For the rest of Section \ref{subsec:run_charac}, we only state our
technical estimates, deferring proofs to Appendix \ref{app:run_charac}.

We now state a tighter bound on probabilities of large run lengths. We will
find this useful, for instance, to control the number of bit flips in going from  general $\sX$ to $\sX_{L^*}$ having bounded run lengths. 

\begin{lemma}
There exists $d_0>0$
such that the following occurs: Consider  any $\beta > 1/2$, and define $\ell \equiv \lfloor 2\beta \log (1/d) \rfloor$. For all $d< d_0$, if $\sX\in\St$ is such that $H(\sX) > 1 - d^{\beta}$, we have
\begin{align}
\sum_{l=\ell}^\infty\  l p_L(l) \ &\leq 20 d^{\beta} \, ,
\end{align}
\label{lemma:L_tail_control}
\end{lemma}

We use
$L(k)$ to denote the vector of lengths $(L_1, L_2, \ldots, L_k)$ of a randomly selected block of $k$ consecutive runs
(a `$k$-block').
Formally, $(L_1,L_2,\dots,L_k)$ is the vector of lengths of the first
$k$ runs starting from bit $X_1$, under the Palm measure  $\prob_1$ introduced
in Section \ref{sec:Preliminaries}.
\begin{coro}
There exists $d_0>0$
such that the following occurs: Consider  any positive integer $k$ and any $\beta > 1/2$, and define $\ell \equiv \lfloor 2\beta \log (1/d) \rfloor$. For all $d<d_0$, if $\sX\in\St$ is such that $H(\sX) > 1 - d^{\beta}$, we have
\begin{align}
\sum_{l_1+ \ldots +l_k \geq k\ell}\,  (l_1 + \ldots +l_k) p_{L(k)}(l_1, \ldots, l_k) \ &\leq 20 k^2 d^{\beta} \, .
\end{align}
\label{coro:Lk_tail_control}
\end{coro}

Clearly, $\E[L_1 + \ldots + L_k] = k \mu(\sX)$.  We have
\begin{align*}
H(\sX) \leq \frac{H(L_1, L_2, \ldots, L_k)}{k \mu(\sX)}\, .
\end{align*}
A stronger form of Lemma \ref{lemma:L_TV} follows.

\begin{lemma}\label{lemma:L_TVstrong}
Let $p_{L(k)}^*(l_1,\dots,l_k) \equiv 2^{-\sum_{i=1}^k l_i}$. For the same $\const'$ and $d_0  >0$ as in Lemma \ref{lemma:L_TV}, the following occurs. Consider any positive integer $k$ and any $\beta > 1/2$. For all $d < d_0$, if $\sX\in\St$ is such that $H(\sX) > 1 - d^{\beta}$, we have
\begin{align*}
\sum_{l_1=1}^\infty \sum_{l_2=1}^\infty \ldots \sum_{l_k=1}^\infty
\left|p_{L(k)}(l_1, \ldots, l_k) - p_{L(k)}^*(l_1, \ldots, l_k) \right| \leq
\const'\sqrt{k}\,d^{\beta/2}\, .
\end{align*}
\end{lemma}

We now relate the run-length distribution in $\sX$ and in $Y(X^n)$ (as $n \rightarrow \infty$). For this, we first need a characterization of
$Y$ in terms of a stationary ergodic process. Let $\sD = ( \ldots, D_{-1}, D_0, D_1, D_2, \ldots )$ be an i.i.d. Bernoulli$(d)$, independent
of $\sX$. Construct $\sY$ as follows. Look at $X_1, X_2, \ldots$. Delete bits corresponding
to $D_1, D_2, \ldots$. The bits remaining are $Y_1, Y_2, \ldots$ in order. Similarly, in
$X_0, X_{-1}, X_{-2}, \ldots$ delete bits corresponding
to $D_0, D_{-1}, D_{-2}, \ldots$. The bits remaining are $Y_0, Y_{-1}, \ldots$ in order.

\begin{propo}
The process $\sY$ is stationary and ergodic for any stationary ergodic $\sX$.
\label{propo:Y_stat_ergodic}
\end{propo}

Notice on the other hand that $(\sX,\sY)$ are \emph{not} jointly stationary.

The channel output $Y(X^n)$ is then $(\sY)_1^M$ where $M \sim \textup{Binomial}(n, 1-d)$. It is easy to check that
\begin{align*}
H(\sY) = H(Y_{\sX})
\end{align*}
(cf. Eq.~(\ref{eq:hyx_defined})). We will henceforth use $H(\sY)$ instead of the more cumbersome  notation $H(Y_{\sX})$.

Let $q_L$ denote the block perspective run-length distribution for $\sY$. Denote by $q_{L(k)}$
the block perspective distribution for $k$-blocks in $\sY$.
Lemmas \ref{lemma:mean_closeto2}, \ref{lemma:L_TV}, \ref{lemma:L_tail_control}, \ref{lemma:L_TVstrong}
 and Corollary \ref{coro:Lk_tail_control}
hold for any stationary ergodic process,
hence they hold true if we replace $(\sX, p)$ with $(\sY, q)$.

In proving the upper bound, it turns out that we are able to establish a bound of $H(\sY) > 1 - d^{2-\eps}$ for $\eps > 0$ and small $d$, but
no corresponding bound for $H(\sX)$.
Next, we establish that if $H(\sY)$ is close to $1$, this leads to tight control over the tail for $p_L(\,\cdot\,)$. This is
a corollary of Lemma \ref{lemma:L_tail_control}.

\begin{lemma}
There exists $d_0>0$
such that the following occurs: Consider any $\gamma>1/2$, and define
$\ell \equiv \lfloor 2\gamma \log (1/d) \rfloor$. For all $d<d_0$,
if $H(\sY) \geq 1 - d^{\gamma}$, we have
\begin{align*}
\sum_{l=2\ell}^\infty\  l p_L(l) \ &\leq 80 d^{\gamma} \, .
\end{align*}
Note that $p_L$ refers to the block
length distribution of $\sX$, not $\sY$.
\label{lemma:pL_tail_control_fromHy}
\end{lemma}

\begin{coro}
There exists $d_0>0$
such that the following occurs: Consider any positive integer $k$ and $\gamma>1/2$, and define
$\ell \equiv \lfloor 2\gamma \log (1/d) \rfloor$. For all $d<d_0$,
if $H(\sY) \geq 1 - d^{\gamma}$, we have
\begin{align*}
\sum_{l=2kl_0}^\infty\  (l_1+ \ldots + l_k) p_{L(k)}(l_1, \ldots, l_k) \ &\leq 80 k^2 d^{\gamma} \, .
\end{align*}
\label{coro:pLk_tail_control_fromHy_STRONG}
\end{coro}

Consider $\sX$ being i.i.d. Bernoulli$(1/2)$. Clearly, this corresponds to $\sY$ also i.i.d. Bernoulli$(1/2)$.
Hence, each has the same run length distribution
$p_L^*(l) = q_L^*(l) = 2^{-l}$. This happens irrespective of the deletion probability $d$.
Now suppose $\sX$ is not i.i.d. Bernoulli$(1/2)$ but
approximately so, in the sense that $H(\sX)$ close to $1$. The next lemma establishes, that in this case also, the run
length distribution of $\sY$ is very close to that of $\sX$, for small run lengths and small $d$.

\begin{lemma} \label{lemma:xy_similar_runs}
There exist a function $(\const,\eps)\mapsto d_0(\const, \eps) > 0$
and constants $\const_1 < \infty$, $\const_2 < \infty$ such that the
following happens, for any  $\beta\in (1/2,2)$, $\eps>0$ and $\const < \infty$.\\
(i) For  all $d < d_0$, for all $\sX$ such that $H(\sX) > 1 - d^{\beta}$,
and all $l < \const \log(1/d)$, we have
\begin{align*}
|p_L(l) - q_L(l)| \leq \const_1 d^{1 +\beta/2-\eps} \, .
\end{align*}
(ii) For all $d< d_0$ and all $\sX$ such that $H(\sX) > 1 - d^{\beta}$, we have
\begin{align}
\label{eq:muX_close2_muY}
|\mu(\sX) - \mu(\sY)| \leq \const_2 d^{1+\beta/2}\, .
\end{align}
\end{lemma}
Let us emphasize that $\const_1, \const_2$ do not depend at all on $\beta, \eps, \const$, where as $d_0$ does not depend on $\beta$ in the above lemma. Analogous comments apply to the remaining lemmas in this section.

As before, we are able to generalize this result to blocks of $k$ consecutive runs.
\begin{lemma}\label{lemma:xy_similar_runs_kSTRONG}
There exist a function $(\const,\eps)\mapsto d_0(\const, \eps) > 0$
and a constant $\const < \infty$ such that the
following happens, for any $\beta\in (1/2,2)$, $\eps>0$ and $\const <
\infty$.

For all $d < d_0$, for all integers $k >0$ and $(l_1, l_2, \ldots,
l_k)$ such that $\sum_{i=1}^k l_i < \const \log(1/d)$,
and all $\sX$ such that $H(\sX) > 1 - d^{\beta}$,
we have
\begin{align*}
|p_{L(k)}(l_1, \ldots, l_k) - q_{L(k)}(l_1, \ldots, l_k)| \leq \const' \, d^{1 +\beta/2 - \eps} \, .
\end{align*}
\end{lemma}

In proving the lower bound, we have $H(\sX^\dagger) = 1 - O(d^2)$, but no corresponding bound for $H(\sY)$. The next lemma allows us to get tight control over the
tail of $q_L^\dagger(\cdot)$.

\begin{lemma}
For any $\eps>0$, there exists $d_0\equiv d_0(\eps)>0$
such that the following occurs: Consider any $\beta\in (1/2,2]$, and define
$\ell \equiv \lfloor 4\log (1/d) \rfloor$. For all $d<d_0$,
if $H(\sX) \geq 1 - d^{\beta}$, we have
\begin{align*}
\sum_{l=\ell}^\infty\  l q_L(l) \ &\leq d^{\beta-\eps} \, .
\end{align*}
\label{lemma:qL_tail_control_fromHx}
\end{lemma}

Define $p_{L(k)}^*(l_1, \ldots, l_k )\equiv 2^{-\sum_{i=1}^k l_i}$. We
show, using Lemma \ref{lemma:xy_similar_runs_kSTRONG}, that if
$H(\sY)$ is close to 1, than one can bound the distance between
$p_{L(k)}(\,\cdot\,)$ and $p_{L(k)}^*(\,\cdot\,)$.
\begin{lemma}\label{lemma:xy_similar_runs_gammaSTRONG}
There exist a function $(\const,\eps)\mapsto d_0(\const, \eps) > 0$
and constants $\const_1 < \infty$, $\const_2 < \infty$ such that the
following happens, for any   $\eps>0$ and $\const <
\infty$.\\
 (i) For all $d<d_0$, all sources $\sX$ such that $H(\sX) > 1 - d^{0.6}$ and $H(\sY) > 1 - d^{\gamma}$,
and all integers $k>0$ and $(l_1, l_2, \ldots, l_k)$ such that $\sum_{i=1}^k l_i < \const \log(1/d)$,
we have
\begin{align}
|p_{L(k)}(l_1, \ldots, l_k) - p_{L(k)}^*(l_1, \ldots, l_k )| &\leq  \, d^{\gamma/2-\eps}
\label{eq:L_Vstrong_fromY} \, ,\\
|p_{L(k)}(l_1, \ldots, l_k) - q_{L(k)}(l_1, \ldots, l_k)| &\leq  \, d^{1 +\gamma/2-\eps} \, .
\label{eq:xy_similar_runs_gammaSTRONG}
\end{align}
(ii) For all $d<d_0$, all sources $\sX$ such that $H(\sX) > 1 -
d^{0.6}$ and $H(\sY) > 1 - d^{\gamma}$,
we have
\begin{align}
|\mu(\sX) -2|  &\leq \const_1 d^{\gamma/2} \label{eq:muX_close2_two_using_Hy}\, ,\\
|\mu(\sX) -\mu(\sY)|  &\leq \const_2 d^{1+ \gamma/2} \label{eq:muX_close2_muY_using_Hy} \, .
\end{align}
\end{lemma}

The next Lemma assures us that if $\sX \in \St_{\lfloor 1/d \rfloor}$, then very few runs in $\sY$ are much longer than ${\lfloor 1/d \rfloor}$. In fact, we show that $q_L(\lambda {\lfloor 1/d \rfloor})$ decays exponentially in $\lambda$.

\begin{lemma}
There exists $d_0>0$ such that, for all $d< d_0$, the following occurs: Consider any $\sX \in \St_{\lfloor 1/d \rfloor}$ such that $H(\sX)> 1- d^{2/3}$. Then, for all $\lambda > 2$ such that $\lambda \lfloor 1/d \rfloor$ is an integer, we have
\begin{align*}
q_L(\lambda \lfloor 1/d \rfloor) \leq d^{\lambda-2} \, .
\end{align*}
\label{lemma:q_exp_decay_in_St1byd}
\end{lemma}

Next, we prove some analogous results for super-runs, cf. Definition \ref{def:super-run}, that we also need.

We denote by $\tLr$ the length of the first run in a random super-run
and by $\tLa$ the total length of the remaining runs of the same
super-run.
More precisely, we repeat here the construction of Section
\ref{sec:Preliminaries},
and define a new Palm measure, $\prob_{s1}$, which is the measure of
$\sX$ conditional on $X_1$ being the first bit of a super-run. Then,
$\tLr$ the length of the first run of this super-run, and $\tLa$ is
the residual length of the same super run, always under the Palm
measure $\prob_{s1}$.
 Here `rep' indicates `repeated' with $\tLr$ being the number of repeated bits and `alt' indicates `alternating' with $\tLa$ being the number of alternating bits. We denote
the type of a random super run by $\tT \equiv (\tLr, \tLa)$ and the length by $\tL \equiv \tLa + \tLr$. We need versions of Lemmas \ref{lemma:L_tail_control} and \ref{lemma:pL_tail_control_fromHy} for super-runs.

Define $\tmu(\sX) \equiv 1/\E[\tL]$. It is easy to see that
\begin{align}
H(\sX) \leq \frac{H(\tT)}{\tmu(\sX)} \, .
\label{eq:superrun_hX_upper_bd}
\end{align}
We denote by $p_{\tT}$ the distribution of $\tT$.
Define $p_{\tT}^*(l_1,l_2) \equiv 2^{-l_1-l_2}$, this being the distribution for the i.i.d. Bernoulli$(1/2)$
process $\sX^*$. We denote by $p_{\tL}$ the distribution of $\tL$ in $\sX$.
Clearly,
\begin{align*}
p_{\tL}(l) = \sum_{\lr=2}^l p_{\tT}(\lr, l-\lr) \, .
\end{align*}

\begin{lemma}
There exists $d_0>0$
such that the following occurs. For  any $\beta > 1/2$ and $d < d_0$, if $\sX\in\St$ is such that $H(\sX) > 1 - d^{\beta}$, we have
\begin{align*}
|\tmu(\sX) - 4 | \leq 4 \, d^{\beta/2} \, .
\end{align*}
\label{lemma:tmu_closeto4}
\end{lemma}

\begin{lemma}
There exists $d_0>0$
such that the following occurs: Consider  any $\beta > 1/2$, and define $\ell \equiv \lfloor 2\beta \log (1/d) \rfloor$. For all $d< d_0$, if $\sX\in\St$ is such that $H(\sX) > 1 - d^{\beta}$, we have
\begin{align*}
\sum_{l=\ell}^\infty l p_{\tL}(l) \leq 40 d^\beta\, .
\end{align*}
\label{lemma:tL_tail_bound}
\end{lemma}
Let $q_{\tL}(\cdot)$ the distribution of super-run lengths in $\sY$, and $\tmu(\sY)$ denote the mean length of a super-run in $\sY$.

\begin{lemma}
There exists $d_0>0$
such that the following occurs: Consider any  $\gamma>1/2$, and define
$\ell \equiv \lfloor 2\gamma \log (1/d) \rfloor$. For all $d<d_0$,
if $H(\sX) \geq 1 - d^{0.6}$ and  $H(\sY) \geq 1 - d^{\gamma}$, we
have
\begin{align*}
\sum_{l=\ell}^\infty\  l p_{\tL}(l) \ &\leq 80 d^{\gamma} \, .
\end{align*}
Note that $p_{\tL}$ refers to the super-run
length distribution of $\sX$, not $\sY$.
\label{lemma:tL_tail_control_fromHy}
\end{lemma}

\begin{coro}
There exists $d_0>0$
such that the following occurs: Consider any positive integer $k$, any $\gamma>1/2$, and define
$\ell \equiv \lfloor 2\gamma \log (1/d) \rfloor$. For all $d<d_0$,
if $H(\sX) \geq 1 - d^{0.6}$ and  $H(\sY) \geq 1 - d^{\gamma}$, we have
\begin{align*}
\sum_{l_1+ \ldots +l_k \geq k\ell}^\infty\  (l_1+ \ldots + l_k) p_{\tL(k)}(l_1, \ldots, l_k) \ &\leq 80 k^2 d^{\gamma}\, .
\end{align*}
\label{coro:tLk_tail_control_fromHy_STRONG}
\end{coro}

Proofs of all results stated in Section \ref{subsec:run_charac} above (except the first two) are available in Appendix \ref{app:run_charac}.

\subsection{Rate achieved by a process}
\label{subsec:information_rate}

We make use of an approach similar to that of Kirsch and Drinea \cite{KirschDrinea} to
evaluate $I(\sX)$ for a stationary ergodic process $\sX$ that may
be used to generate an input for the deletion channel.
A fundamental difference is that \cite{KirschDrinea} only considers
processes with i.i.d. runs.
Our analysis is instead  general. This enables us to obtain tight upper and lower bounds (up to $O(d^{3-\eps})$),
hence leading to an estimate for the channel capacity.

We depart from the notation of Kirsch and Drinea, retaining $X_i$ for the $i$th bit
of $X$, and using $Y(j)$ to denote the $j$th run in $Y(X^n)$.
Denote by $L_1, L_2, \ldots, L_{m}$ the lengths of runs in $X_1^n$ (where $m$ is a non-decreasing function of $n$ for any fixed $X_1^\infty$).  Let the $i$th run consists of $b(i)$'s, where $b(i) \in \{0,1\}$. For instance, if the first run consists of $0$'s, then
$b(i)= i+1 \; (\textup{mod } 2)$.

We use $X(j)$ to denote the
concatenation of runs in $X$ that led to $Y(j)$, with the first run in
$X(j)$ contributing at least one bit (if the run is completely deleted, then it is
part of $X(j-1)\,$). $X(1)$ is an exception. This is made precise in
Table \ref{table:KirschDrinea_runs}, which is essentially the same
as \cite[Figure 1]{KirschDrinea}, barring changes in notation.
We call runs in $X(j)$ the \emph{parent runs} of the run $Y(j)$.

\begin{table}[t]
\begin{center}
{\fontfamily{cmtt}\selectfont
\begin{tabular}{ll}
\hline
1: & Set $X(1)=Y(1)=$the empty string.\\
2: & $j \leftarrow 1$\\
3: & For $i=1$ to $m$ do\\
4: & \hspace{0.4cm} $\sigma \leftarrow b(i)^{L_i}$\\
5: & \hspace{0.4cm} $\omega \leftarrow$ the bits in $Y$ that arise from $i$th run in $X$\\
6: & \hspace{0.9cm} \% $\sigma$ is a (possibly empty) string of all $b(i)$'s.\\
7: & \hspace{0.9cm} \% $Y(j)$ is a (possibly empty) string of all $b(j)$'s.\\
8: & \hspace{0.4cm} If $b(i)=b(j)$ or $|\omega|=0$ then \\
9: &\hspace{0.9cm} \% $\omega$ is contained in the current block $Y(j)$ of $Y$\\
10: &\hspace{0.9cm} $Y(j) \leftarrow Y(j) \omega$\\
11: &\hspace{0.9cm} $X(j) \leftarrow X(j) \sigma$\\
12: & \hspace{0.4cm} Else  \% $\omega$ is a prefix of $Y(j+1)$\\
13: &\hspace{0.9cm} $j \leftarrow j+1$\\
14: &\hspace{0.9cm} $Y(j) \leftarrow Y(j) \omega$\\
15: &\hspace{0.9cm} $X(j) \leftarrow X(j) \sigma$\\
16: & \hspace{0.4cm} End If\\
17: &  End For\\
\hline
\end{tabular}}
\end{center}
\caption{Procedure for generating $Y(1), Y(2), \ldots, Y(M)$ and $X(1), X(2), \ldots, X(M)$ given $X^n$ and $Y(X^n)$ (adapted from \cite[Figure 1]{KirschDrinea}).}\label{table:KirschDrinea_runs}
\end{table}

We define $K(X^n)$ as the vector of $|X(j)|$.
Let the total number of runs in $Y(X^n)$ be $M$.
Thus,
\begin{align*}
Y(X^n)= \,&Y(1)  \ldots Y(M-1) Y(M)\, ,\\
X^n = \,&X(1)  \ldots X(M-1) X(M)\, ,\\
K(X^n) = \,&\left( |X(1)|, \ldots , |X(M-1)|\right)\, .
\end{align*}
Note that $X(j)$ consists of an odd number of runs for $1< j<M$.

We write
\begin{align}
I(X^n; Y(X^n)) = H(Y) - H(Y, K| X^n) + H(K|X^n, Y)\, ,
\label{eq:I_three_terms}
\end{align}
which is analogous to the identity $I(X^n; Y(X^n)) = H(X^n) - H(X^n,
K| Y) + H(K|X^n, Y)$ used in \cite{KirschDrinea},
but more convenient for our proof.

Let $L_{\sY}$ be an integer random variable having the distribution $q_L$, i.e. the distribution
of run length in $\sY$. It is easy to see that
\begin{align*}
\lim_{n \rightarrow \infty} \frac{H(Y(X^n))}{n(1-d)} = H(\sY) \leq \frac{H(L_{\sY})}{\mu(\sY)}
\end{align*}
holds, similar to (\ref{eq:run_hx_upper_bd}).
It turns out that this suffices for our upper bound
(cf. Lemma \ref{lemma:hy_is_large}).

Consider the second term in Eq.~\eqref{eq:I_three_terms}.
Let $D^n$ denote the $n$-bit binary vector that indicates which
bit locations in $X^n$ have suffered deletions. We have
\begin{align}
H(Y, K| X^n) &= H(D^n|X^n) - H(D^n|X^n, Y, K) \nonumber\\
    &= n h(d) - H(D^n|X^n, Y, K)\, .
\label{eq:hYKgivenX_split}
\end{align}
We study $H(D^n|X^n, Y, K)$ by constructing an appropriate modified deletion process in Section \ref{subsec:modified_deletion}

Consider the third term in Eq.\eqref{eq:I_three_terms}. From \cite{KirschDrinea}, we know that
\begin{align*}
\lim_{n \rightarrow \infty} \frac{H(K|X^n, Y)}{n} = \frac{\lim_{n \rightarrow \infty} H(\,|X(2)| \ | X(2) \ldots X(M), Y(2) \ldots Y(M))}{\E [|X(2)|]}\, .
\end{align*}
Here $X(2) \ldots X(M)$ denotes the string obtained by concatenating
$X(2)$, \dots, $X(M)$, without separation marks, and analogously for
$Y(2) \ldots Y(M)$.
Roughly, single deletions do not lead to ambiguity in $|X(2)|$ if $X(2) \ldots $ and $Y(2) \ldots$ are known. Thus,
this term is $O(d^2)$. It turns out we can we can get a good estimate for this term by computing it
for the i.i.d. Bernoulli($1/2$) case.

\begin{lemma}
\label{lemma:K_conditional_entropy}
For any $\eps>0$, there exists $d_0\equiv d_0(\eps)>0$, and $\const < \infty$ such that for all $d<d_0$ the following occurs: Consider any
$\sX \in \St_{\lfloor 1/d \rfloor}$ such that $H(\sX)>1-d^{1-\eps}$ and $\max\{H(\sX), H(\sY)\}>1-d^{\gamma}$ for some $\gamma \in (1/2,2)$.
Then
\begin{align}
\left|\lim_{n \rightarrow \infty} \, \frac{1}{n} \,
  H(K(X^n)|X^n,Y(X^n)) - d^2 c_4  \right|\le  \const
d^{1+\gamma-\eps/2}\, ,
\label{eq:K_conditional_entropy}
\end{align}
where
\begin{align*}
c_4\equiv &\sum_{j=4}^\infty 2^{-(2+j)}  \,  (j-1 ) (j-3)\, h\!\left( \frac{1}{j-1}\right)\nonumber\\
&+\sum_{i=2}^\infty \sum_{j=4}^\infty  2^{-(i+j+1)}  \, (i+j-1 )(j-3)\, h\!\left( \frac{i+ 1}{i+j-1}\right).
\end{align*}
\end{lemma}
Note that with $\gamma=2-\eps/2$, we obtain $|\delta| \leq \const d^{3-\eps}$.

The proof of Lemma \ref{lemma:K_conditional_entropy} is quite technical and uses a modified deletion process (cf. Section \ref{subsec:modified_deletion}). We defer it to Appendix \ref{app:K_conditional_entropy_proof}. 

\begin{lemma}
\label{lemma:hy_upper_bound}
For any $\eps >0$, there exists $d_0 \equiv d_0(\eps)>0$ such that if $H(\sY) \geq 1 - d^{2-\eps/2}$, then
\begin{align*}
H(\sY) \leq 1 - \frac{1}{2} \sum_{l=1}^\infty q_L(l) \big ( \log q_L(l) + l \big) + d^{3-\eps}\, ,
\end{align*}
for all $d < d_0$.
\end{lemma}

The proof of this lemma is fairly straightforward.

\begin{proof}[Proof of Lemma \ref{lemma:hy_upper_bound}]
An explicit calculation yields $H(q_L) = \mu(\sY) - D(q_L || q_L^*)$ where $q_L^*$ is the run length distribution corresponding to the i.i.d. Bernoulli$(1/2)$ half process (cf. proof of Lemma \ref{lemma:L_TV}). We know $H(\sY) \leq H(q_L)/\mu(\sY)$. It follows that
\begin{align}
H(\sY) \leq 1 - D(q_L || q_L^*)/\mu(\sY)\, .
\label{eq:hy_upperbd_1minusDbymu}
\end{align}
Using Lemma \ref{lemma:xy_similar_runs_gammaSTRONG}(ii), we deduce that
\begin{align*}
\left |\frac{1}{\mu(\sY)} - \frac{1}{2} \right |  \leq \frac{1}{3}\,d^{1-\eps/2}\, ,
\end{align*}
 and, in particular, $\mu(\sY) <3$  for small $d$. Hence, substituting
 in Eq.~(\ref{eq:hy_upperbd_1minusDbymu}) and using the lower bound  $H(\sY) \geq 1 - d^{2-\eps/2}$
we  have $ D(q_L || q_L^*) < 3 d^{2-\eps/2}$. Explicit calculation gives $D(q_L || q_L^*)=
 \sum_{l=1}^\infty q_L(l) \big ( \log q_L(l) + l)$. The result follows by plugging into Eq.~\eqref{eq:hy_upperbd_1minusDbymu}.
\end{proof}

%
%
\subsection{A modified deletion process}
\label{subsec:modified_deletion}

We want to get a handle on the term $H(D^n|X^n,Y,K)$. The main
difficulty in achieving this is that
a fixed run in $Y$ can arise in ways from parent runs, via a countable infinity of different deletion `patterns'. For example, consider that a run in $Y$ may have \emph{any} odd number of parent runs. Moreover, a countable infinity of these deletion patterns `contribute' to $H(D^n|X^n,Y,K)$.

However, we expect that deletions are typically well separated at small deletion probabilities, and as a result, there are only a few dominant `types' of deletion patterns that influence the leading order terms $H(D^n|X^n,Y,K)$.
Deletions that `act' in isolation from other deletions should contribute an order $d$ term: for instance a positive fraction of runs in $X^n$ should have a length $4$, and with probability of order $d$, they should shrink to runs of length $3$ in $Y$ due to one deletion. Each time this occurs, there are four (equally likely) candidate positions at which the one deletion occurred, contributing $\log(4)$ to $H(D^n|X^n,Y,K)$. Similarly, pairs of `nearby' deletions (for instance in the same run of $X^n$) should contribute a term of order $d^2$. We should be able to ignore instances of more than two deletions occurring in close proximity, since (intuitively) they should have a contribution  of $O(d^3)$ on $H(D^n|X^n,Y,K)$.

We formalize this
intuition by constructing a suitable \emph{modified deletion process} that allows us to focus on the dominant deletion patterns in our estimate of this term. We bound the error in our estimate due to our modification of the deletion process, leading to an estimate of $H(D^n|X^n,Y,K)$ that is exact up to order $d^2$.

We restrict attention to $\sX \in \St_{\lfloor 1/d \rfloor}$. Denote by $R_j$ the $j$th run in $\sX$ (where the run including bit 1 is labeled $R_1$).
$R_j$ has length $L_j$. Recall that the deletion process $\sD$ is an i.i.d. Bernoulli$(d)$ process, independent
of $\sX$, with $D_1^n$ being the $n$-bit vector that contains a $1$ if
and only if the corresponding bit in $X^n$ is deleted by the channel $W_n$.
We define an auxiliary sequence of channels
$\widehat{W}_n$ whose output --denoted by $\widehat{Y}(X^n)$--
is obtained by modifying the deletion channel output: $\widehat{Y}(X^n)$ contains all bits present in $Y(X^n)$ and some of the deleted bits in addition. Specifically,  whenever
there are \emph{three} or more deletions in a single run $R_i$
under $\sD$, the run $R_i$ suffers no deletions
in $\widehat{Y}(X^n)$.

Formally, we construct this sequence of channels
when the input is a stationary process $\sX$ as follows.
For all integers $i$, define:
\begin{enumerate}[]
\item $\sZ^{i} \equiv$
Binary process that is zero throughout except if $R_i $ contains at
$3$ or more deletions, in which case $Z^{a,i}_l=1$ if and only if $X_l \in R_i$ and $D_l=1$.
\end{enumerate}

Define
\begin{align*}
\mathbb{Z} = \sum_{i=-\infty}^\infty \sZ^{i}\, ,
\end{align*}
where $\sum$ here denotes bitwise OR.
Finally, define $\widehat{\sD}(\sD, \sX) \equiv \sD \oplus \sZ $ (where $\oplus$ is componentwise
sum modulo $2$).
The output of the channel $\hW_n$ is simply defined by
deleting from $X^n$ those bits whose positions correspond to $1$s in
$\hsD$. We define $\widehat{K}(X^n)$ for the modified deletion process
in the same way as $K(X^n)$.
The sequence of channels $W_n$ are defined by $\sD$, and
the coupled sequence of channels $\widehat{W}_n$ are defined
by $\sD$. We emphasize that $\widehat{\sD}$ is a function of
$(\sX,\sD)$.

Note that if $D_l = 0$
then $Z_l=0$ and hence $\widehat{D}_l =0$. Thus $\widehat{\sD}$ is obtained by flipping the
$1$s in $\sD$ that also correspond to $1$s in $\sZ$. If $Z_i =1$, i.e. $D_i=1, \widehat{D}_i=0$, we will say that a deletion is \emph{reversed} at position $i$.
It is not hard to see that the process $\sZ$ is  stationary.
(In fact $(\sX,\sD,\sZ,\hsD)$ are jointly stationary.) Define $z \equiv \prob(Z_i=1)$,
where $i$ is arbitrary.


The expected number of
deletions reversed due to a run with length $\ell$ is
bounded above by
\begin{align}
\ell d - \ell d (1-d)^{l-1} - 2 \binom{l}{2} d^2(1-d)^{\ell -2}
\leq \ell (\ell -1)(\ell -2) d^3
\leq \ell^3 d^3\, ,
\label{eq:reversals_per_run}
\end{align}
using $(1-d)^{l-1} \geq 1- (l-1)d$ and $(1-d)^{l-2} \geq 1- (l-2)d$.

We know that each run has length at least $1$.
Thus, we have the following.
\begin{fact}
\label{fact:z_is_d3_EL3}
For arbitrary stationary process $\sX$, the probability $z$ of a reversed deletion at an arbitrary position $i$ is bounded as
$
z \leq d^3\E[L^3]\,
$.
\end{fact}

Now $\E[L^3] \leq d^{-2} \E[L]$ for $\sX \in \St_{\lfloor 1/d
  \rfloor}$.  Combining with Lemmas \ref{lemma:L_tail_control} and \ref{lemma:pL_tail_control_fromHy}, we obtain:
\begin{fact}
\label{fact:ELcube_is_small}
For any $\eps > 0$, there exists $d_0\equiv d_0(\eps)>0$ and $\const< \infty$ such that for any $d< d_0$  the following occurs: Consider any $\sX \in \St_{\lfloor 1/d \rfloor}$  such that
$\max\{H(\sX),H(\sY)\} \geq 1-d^{\gamma}$. Then we have $\E[L^3] <  \const d^{\gamma -2}$.
\end{fact}

Note that $\max\{H(\sX),H(\sY)\} \geq 1-d^{2-\eps/2}$ holds for relevant processes $\sX$ (see Lemma \ref{lemma:hy_is_large}), justifying our assumption above.

The next proposition follows immediately from Facts \ref{fact:z_is_d3_EL3} and \ref{fact:ELcube_is_small}.
\begin{propo}
\label{propo:z_is_small}
For any $\eps > 0$, there exists $d_0\equiv d_0(\eps)>0$ and $\const< \infty$ such that for any $d< d_0$  the following occurs: Consider any $\sX \in \St_{\lfloor 1/d \rfloor}$  such that
$\max\{H(\sX),H(\sY)\} \geq 1-d^{\gamma}$. Then we have $z < \const d^{1+\gamma}$.
\end{propo}

We now analyze the modified deletion process with the aim of estimating $H(\widehat{D}^n|X^n, \widehat{Y}, \widehat{K})$. Notice that for any run $R_i$, either all deletions in $R_i$ are reversed (in which case we say that $R_i$ suffers deletion reversal), or none of the deletions are reversed (in which case we say that $R_i$ is unaffected by reversal). It follows that
\begin{align}
H(\widehat{D}^n|X^n, \widehat{Y}, \widehat{K}) &= \sum_{j=1}^{M} H(\widehat{D}(j)|\hX(j), \widehat{Y}(j))\, ,
\label{eq:modified_deletion_entropy}
\end{align}
where $\widehat{D}(j)$ consists of the substring of $\widehat{D}^n$ corresponding to $\hX(j)$.
As before, when we study $H(\widehat{D}^n|X^n, \widehat{Y}, \widehat{K})/n$ in the limit $n\rightarrow \infty$, the terms  corresponding to $j=1$ and $j=M$ can be neglected, and we can perform the calculation by considering the stationary processes $\sX$, $\sY$ and $\sD$.

Recall the definition of the parent runs $\hX(j)$ of a run $\widehat{Y}(j)$ for $j>1$ from Section \ref{subsec:information_rate}.  Consider the possibilities for how many runs $\hX(j)$ contains, and the resultant ambiguity (or not) in the position of deletions (under $\widehat{\sD}$) in the parent run(s):
\begin{enumerate}[]
\item {\bf A single parent run.}\\
Let the parent run be $R_\parent$. The parent run should not
disappear\footnote{We emphasize that we are referring here to
  deletions under $\widehat{\sD}$.}; by definition it should contribute at least one bit to $\widehat{Y}(j)$.
The run $R_{\parent+1}$ should not disappear (else it is also a
parent). $R_\parent$ can suffer $0, 1$ or $2$ deletions (else we have
a deletion pattern not allowed under $\widehat{\sD}$). The cases of $1$ or $2$ deletions lead to ambiguity in the location of deletions.

Note that  if $R_{\parent-1}$ disappears  then $R_{\parent-2}$ also disappears (else $R_{\parent-2}, R_{\parent-1}$ are also parents of $\widehat{Y}(j)$), and so on.

\item {\bf A combination of three parent runs.}\\
Let the parent runs be $R_\parent, R_{\parent +1}$ and $R_{\parent+2}$. We know that $R_\parent$ and $R_{\parent+3}$ did not disappear and $R_{\parent+1}$ has disappeared, by definition of $X(j)$ (cf. Table \ref{table:KirschDrinea_runs}).
If $R_\parent$ and $R_{\parent+2}$
suffer no deletions, this leads to no ambiguity in the location of deletions. Ambiguity can arise in case $R_\parent$ and $R_{\parent+2}$ suffer between one and four deletions in total.

Note that  if $R_{\parent-1}$ disappears  then $R_{\parent-2}$ also disappears, and so on.

\item {\bf A combination of $2k+1$ parent runs, for $k=2,3, \ldots$.}\\
Let the parent runs be $R_\parent, R_{\parent+1}, \ldots, R_{\parent+2k}$. The runs $R_{\parent+1}, R_{\parent+3}, \ldots, R_{\parent+2k-1}$ must disappear and $R_\parent$ does not disappear. The runs $R_\parent, R_{\parent+2}, \ldots, R_{\parent+2k}$ must suffer between one and $2(k+1)$ deletions in total for ambiguity to arise in the location of deletions.
\end{enumerate}

Define
\begin{align*}
p_{L(3)}(\gt1, l_2, l_3) &\equiv \sum_{l_1=2}^\infty p_{L(3)}(l_1, l_2, l_3) \, ,\\
p_{L(3)}(\gt1, l_2, \gt1) &\equiv \sum_{l_1=2}^\infty \; \sum_{l_3=2}^\infty p_{L(3)}(l_1, l_2, l_3)\, ,
\end{align*}
and so on.

The following lemma shows the utility of the modified deletion process. We obtain this result by adding the contributions of the cases enumerated above.
\begin{lemma}
There exists $d_0>0$ such that for any $d< d_0$ the following occurs: Consider any $\sX \in \St_{\lfloor 1/d \rfloor}$. Then
\begin{align}
&\lim_{n\to\infty} \frac{1}{n}H(\widehat{D}^n|X^n, \widehat{Y}, \widehat{K}) = \nonumber\\
&\frac{d}{\mu(\sX)}  \sum_{l=2}^\infty p_{L}(l) \; l \log l   \nonumber\\
&+ \frac{d^2}{\mu(\sX)} \sum_{l=2}^\infty p_L(l)\Big \{ \binom{l}{2} \log \binom{l}{2}
-l^2 \log l \Big \}\nonumber\\
&+\frac{d^2}{\mu(\sX)} \sum_{l=2}^\infty  \Big \{p_{L(3)}(\gt1, l, \gt1) \; l \log l
- p_{L(3)}(1, l, 1) \; l \log l \Big \} \nonumber\\
&+ \frac{d^2}{\mu(\sX)} \left (\sum_{l_0 >1, l_2} \Big \{p_{L(3)}(l_0, 1, l_2)\, (l_0+l_2) \log (l_0+l_2) \Big \}
+ \sum_{1,1,l_2} \Big \{ p_{L(3)}(1, 1, l_2) \,l_2 \log l_2 \Big \} \right )  +\delta\, ,
\label{eq:modified_deletion_entropy_estimate}
\end{align}
where
\begin{align}
-11 d^3 \log(1/d) \E [L^3] \leq \delta \leq  140 d^3 \log(1/d) \E
[L^3]\, .
\label{eq:modified_deletion_estimate_error}
\end{align}

\label{lemma:hatD_givenxy}
\end{lemma}

The proof of Lemma \ref{lemma:hatD_givenxy} is quite technical and is deferred to Appendix \ref{app:hatD_lemma_proof}.

Making use of the estimates of $p_{L(k)}(\cdot)$ derived in Section \ref{subsec:run_charac},
we obtain the following corollary of Lemma \ref{lemma:hatD_givenxy}. It is proved in Appendix \ref{app:hatD_lemma_proof}.

\begin{coro}
For any $\eps > 0$, there exists $d_0\equiv d_0(\eps)>0$ and $\const< \infty$ such that for any $d< d_0$ the following occurs: Consider any $\sX \in \St_{\lfloor 1/d \rfloor}$  such that $H(\sX) \geq 1-d^{1-\eps}$ and
$\max\{H(\sX), H(\sY)\} \geq 1-d^{\gamma}$ for some $\gamma \in (0,2)$. Then
\begin{align*}
&\lim_{n\to\infty} \frac{1}{n}H(\widehat{D}^n|X^n, \widehat{Y}, \widehat{K}) \, = \;
\frac{d}{\mu(\sX)}  \Big \{ \sum_{l=2}^\infty p_{L}(l) \; l \log l \Big \}  + d^2 c_3 +
\xi\, ,
\end{align*}
where
$|\xi| \leq \const d^{1+\gamma-\eps/2}$. Recall that
\begin{align*}
c_3\equiv \frac{1}{2} \left ( -1 + \sum_{l=3}^\infty 2^{-l} \left\{\binom{l}{2} \log \binom{l}{2} - l^2 \log l
+(l-1)(l-3) \log (l-1) + (l-2) \log (l-2) \right \} \,\right) \, .
\end{align*}
\label{coro:hatD_givenxy}
\end{coro}
Note that with $\gamma=2-\eps/2$, we obtain $|\xi| \leq \const d^{3-\eps}$.

We need to show that our estimate for the modified deletion process is also a good estimate for original deletion process. The following simple fact helps us do this:
\begin{fact}
\label{fact:U_V_W}
Suppose $U, \widehat{U}$ and $V$ are random variables with the property that $U$ is a deterministic
function of $\widehat{U}$ and $V$, and also $\widehat{U}$ is a deterministic function of $U$ and $V$. (Denote
this property by $U \xleftrightarrow{V}  \widehat{U}$.) Then
\begin{align*}
|H(U)-H(\widehat{U})| \leq H(V)\, .
\end{align*}
\end{fact}
\begin{proof}
We have $H(U) \leq H(\widehat{U},V) \leq H(\widehat{U})+H(V)$. Similarly,  $H(\widehat{U})\leq H(U)+H(V)$.
\end{proof}

It is not hard to see that $(X^n, Y, K, D^n)  \xleftrightarrow{Z^n} (X^n, \widehat{Y}, \widehat{K}, \widehat{D}^n)$ and
$(X^n, Y, K) \xleftrightarrow{Z^n} (X^n, \widehat{Y}, \widehat{K})$. Using Fact \ref{fact:U_V_W}, we obtain
\begin{align}
|H(\widehat{D}^n|X^n, \widehat{Y}, \widehat{K}) - H({D}^n|X^n, Y, K)| \leq 2 H(Z^n) \leq 2n h(z)\, .
\label{eq:D_hatD_condentropy_difference}
\end{align}

Combining Eq.~\eqref{eq:D_hatD_condentropy_difference} with Corollary \ref{coro:hatD_givenxy},
we obtain an estimate for the second term in Eq.~\eqref{eq:I_three_terms}.
For future convenience, we form an estimate in terms of  $q_L(\cdot)$ instead of $p_L(\cdot)$, using Lemma \ref{lemma:xy_similar_runs_gammaSTRONG} to make the switch.

\begin{coro}
For any $\eps > 0$, there exists $d_0\equiv d_0(\eps)>0$ and $\const< \infty$ such that for any $d< d_0$  the following occurs: Define $\ell \equiv \lfloor 4 \log (1/d) \rfloor$. Consider any $\sX \in \St_{\lfloor 1/d \rfloor}$  such that $H(\sX) \geq 1-d^{1-\eps}$ and
$\max\{H(\sX), H(\sY)\} \geq 1-d^{2-\eps/2}$. Then
\begin{align*}
\lim_{n \rightarrow \infty} \frac{1}{n}H(Y(X^n), K(X^n)|X^n) =  - \frac{d}{2}  \sum_{l=2}^\ell q_{L}(l) \; l \log l  +\frac{dc_2}{4\ln 2}  \sum_{l=1}^\ell q_L(l) l \nonumber\\
+d \log(1/d) + \frac{d}{\ln 2}\left( 1 - \frac{c_2}{2}\right ) +  d^2 \bigg ( - c_3 - \frac{1}{2 \ln 2} \bigg ) + \delta\, ,
\end{align*}
where $|\delta| \leq \const d^{3-\eps}$. Recall $c_2 \equiv \sum_{l=1}^\infty 2^{-l} l \ln l$.
\label{coro:hygivenx_q}
\end{coro}

Corollary \ref{coro:hygivenx_q} is also proved in Appendix \ref{app:hatD_lemma_proof}.

\subsection{A self improving bound on $H(\sY)$}
\label{subsec:self_improving_bd_hy}

Our next Lemma constitutes a `self-improving' bound on the closeness of $H(\sY)$ to $1$ and leads directly to Lemma \ref{lemma:hy_is_large}.
\begin{lemma}
\label{lemma:cyclic_hy_bound_SLstar}
There exists a function $(\const,\eps)\mapsto d_0(\const,\epsilon)>0$
such that  the following happens for any $\epsilon > 0$, and constants
$\const>0$ and  $\gamma \in (1/2,2)$.
For any $d< d_0$ and any $\sX \in S_{\lfloor1/d \rfloor}$ such that
\begin{align*}
I(\sX) &\geq 1 - d \log(1/d) -A_1 d - \const d^{2-(\epsilon/4)}
\end{align*}
and $H(\sY) \geq 1 - d^{\gamma}$, we have
\begin{align*}
H(\sY) \geq 1 - d^{1+ \gamma/2 -\epsilon/2}\, .
\end{align*}
\end{lemma}
\begin{proof}
From Eq.~\eqref{eq:I_three_terms} we have
\begin{align}
I(\sX) &= \lim_{n \rightarrow \infty} \frac{1}{n} \left \{ H(Y) - H(D^n) + H(D^n|X^n,Y,K) + H(K|X^n,Y)\right \}\nonumber\\
&= (1-d) H(\sY) - h(d)+
\lim_{n \rightarrow \infty} \frac{1}{n} \left \{H(D^n|X^n,Y,K) + H(K|X^n,Y)\right \}\, .
\label{eq:Ix_expanded_in_hy_recursive}
\end{align}

Using Eq.~\eqref{eq:D_hatD_condentropy_difference} and Proposition \ref{propo:z_is_small}, we have
\begin{align*}
\frac{1}{n}\big | H(D^n|X^n,Y,K) - H(\hD^n|X^n,\hY,\hK) \big | \leq
\const_1 d^{1+\gamma}\log(1/d)\, .
\end{align*}

It follows from $H(\sX) > I(\sX)$ and our assumed lower bound on $I(\sX)$, that $H(\sX)>1-d^{1-\eps} $ for some $\eps>0$.
Using Corollary \ref{coro:hatD_givenxy}, $|\mu(\sX)-2| \leq \const_2 d^{\gamma/2}$ from Lemma \ref{lemma:xy_similar_runs_gammaSTRONG}(ii), and Lemmas  \ref{lemma:xy_similar_runs_gammaSTRONG}(i) and \ref{lemma:pL_tail_control_fromHy} to control $p_L(\cdot)$, we have
\begin{align*}
\lim_{n\to\infty} \frac{1}{n}H(\widehat{D}^n|X^n, \widehat{Y}, \widehat{K}) \, = \;
\frac{d}{2}  \Big \{ \sum_{l=2}^\infty 2^{-l} \; l \log l \Big \}  + \delta_1\, ,
\end{align*}
where $|\delta_1| \leq \const_3 d^{1+\gamma/2-\eps/4}$.

Lemma \ref{lemma:K_conditional_entropy} gives
\begin{align*}
\lim_{n \rightarrow \infty} H(K|X^n,Y) \leq \const_4 d^{1+\gamma/2-\eps/4} \, .
\end{align*}
We used here $\gamma<2$.

Plugging back into Eq.~\eqref{eq:Ix_expanded_in_hy_recursive}, we obtain
\begin{align*}
I(\sX) \leq  H(\sY) - d \log(1/d) -A_1 d + \const_5 d^{1+\gamma/2-\eps/4}\, .
\end{align*}

The result follows from the assumption on $I(\sX)$.
\end{proof}

\subsection{Auxiliary lemmas for our lower bound}
\begin{lemma}
Recall $\sX^\dagger$ is the process consisting of i.i.d.  runs with
distribution $p_L^\dagger(l) = 2^{-l}(1+ d( l \log l -c_2l/2))$ (cf. Lemma
\ref{lemma:achievability}). There exists $d_0>0$ such that, for any $d<d_0$ we have the following: For any integer $i$ and any $x_{-\infty}^{i-1}$, we have
\begin{align*}
\big |\prob\big\{X^\dagger_i=1 \big |
(X^\dagger)_{-\infty}^{i-1}=x_{-\infty}^{i-1}\big\} - 1/2\big | \leq
0.05\, .
\end{align*}
\label{lemma:Xdagger_almost_i.i.d.}
\end{lemma}
\begin{proof}
Without loss of generality, suppose $x_{i-1} = 1$. Also, suppose that
it is the $l$th consecutive 1 to occur. Now, since the runs' starting
points  form a renewal process under $\sX^\dagger$,  we have
\begin{align*}
\frac{\prob\big\{X^\dagger_i=0 \big | (X^\dagger)_{-\infty}^{i-1}=x_{-\infty}^{i-1}\big\}}
{\prob\big\{X^\dagger_i=1 \big | (X^\dagger)_{-\infty}^{i-1}=x_{-\infty}^{i-1}\big\}}
= \frac{p_L^\dagger(l)}{\sum_{l'>l} p_L^\dagger(l')} \, .
\end{align*}

A little calculus yields
\begin{align*}
\sum_{l'>l} p_L^\dagger(l') = 2^{-l} \left ( 1 + d \{l \log l + \eta_{1,l}\} \right) \, ,
\end{align*}
where $|\eta_{1,l}| \leq \const_1 l$ for some $\const < \infty$. In comparison, $p_L^\dagger(l) = 2^{-l} \left ( 1 + d \{l \log l - c_2 l/2 \} \right )$.

\noindent Case (i): $l < 1/\sqrt{d}$.\\
In this case, we have $p_L^\dagger(l) = 2^{-l}(1 + \eta_{2,l})$ with $|\eta_{2,l}| \leq d^{0.4}$ and $\sum_{l'>l} p_L^\dagger(l') = 2^{-l}(1 + \eta_{3,l})$ with $|\eta_{3,l}| \leq d^{0.4}$, for sufficiently small $d$. The result follows.

\noindent Case (ii): $l \ge 1/\sqrt{d}$.\\
In this case, $ \{l \log l + \eta_{1,l}\} = \{l \log l - c_2 l/2 \}
(1+\eta_{4,l})$, where $|\eta_{4,1}| \leq 0.01$
provided $d$ is small enough. It follows that
\begin{align*}
\left | \frac{p_L^\dagger(l)}{\sum_{l'>l} p_L^\dagger(l')} - 1 \right | \leq 0.02 \, .
\end{align*}
The result follows.

\end{proof}

\begin{lemma}
Let $q_L^\dagger(\cdot)$ be the run length distribution of $\sY^\dagger$ corresponding to input $\sX^\dagger$. Then there exists $d_0$ (same as in Lemma \ref{lemma:Xdagger_almost_i.i.d.}) such that, for any $d< d_0$, we have $q_L(l) \leq (3/4)^l$ for all $l$.
\label{lemma:Ydagger_exponential_decay}
\end{lemma}
\begin{proof}
It follows from Lemma \ref{lemma:Xdagger_almost_i.i.d.}, that for any $y_{-\infty}^{i-1}$, we have
\begin{align*}
\big |\prob\big\{Y^\dagger_i=1 \big | (Y^\dagger)_{-\infty}^{i-1}=y_{-\infty}^{i-1}\big\} - 1/2\big | \leq 0.1 \, ,
\end{align*}
for $d< d_0$. This gives $q_L(l) \leq (0.45/0.55)^l$, implying the result.
\end{proof}

%
%
\subsection{Proofs of Lemmas \ref{lemma:achievability}, \ref{lemma:hy_is_large},  \ref{lemma:converse_for_restricted_runs} and \ref{lemma:small_loss_by_restricting_runs}}
\label{subsec:lemma_proofs}

We first prove Lemma \ref{lemma:small_loss_by_restricting_runs}, followed by Lemmas \ref{lemma:achievability}, \ref{lemma:hy_is_large} and \ref{lemma:converse_for_restricted_runs}.

\begin{proof}[Proof of Lemma \ref{lemma:small_loss_by_restricting_runs}]
We construct $\asX \in \St_{L^*}$ from $\sX$ as follows: Suppose a super-run starts at $X_j$ and continues until $X_{j+L^*}$. We flip one or both of $X_{j+L^*+1}$ and $X_{j+L^*+2}$ such that the super-run ends at $X_{j+L^*}$. (It is easy to verify that this can always be done. If multiple different choices work, then pick an arbitrary one.)
The density of flipped bits  in $\sX$ is upper bounded by
$\alpha=2 \E[\tL \ind(\tL \geq L^*)]/L^*$. The expected fraction of bits in the
channel output $\aY=Y(\aX^n)$ that have been flipped relative to
$Y=Y(X^n)$ (output of the same channel realization
with different input) is also
at most $\alpha$. Let $F=F(\sX, \sD)$ be the binary vector  having the same
length as $Y$, with a $1$ wherever the corresponding bit in $\aY$
is flipped relative to $Y$, and $0$s elsewhere. The expected fraction of $1$'s
in $F$ is at most $\alpha$. Therefore
\begin{align}
H(F) \leq n (1-d) h(\alpha) + \log (n+1)\, .
\label{eq:hflips_bound}
\end{align}
%
Recall Fact \ref{fact:U_V_W}. Notice that $Y \xleftrightarrow{F} \aY$, whence
\begin{align}
|H(Y) - H(\aY)| \leq H(F)\, .
\label{eq:hy_bound_flips}
\end{align}
Further, $\sX-\asX-\aX^n-\aY$ form a Markov chain,
and $\asX$, $\aX^n$ are deterministic functions of $\sX$.
Hence, $H(\aY | \aX^n) = H(\aY | \asX)$. Similarly,
$H(Y | X^n) = H(Y| \sX)$. Therefore (the second step is analogous to
Eq.~(\ref{eq:hy_bound_flips}))
\begin{align}
\label{eq:hygivenx_bound_flips}
|H(\aY | \aX^n)  - H(Y | X^n)| =
|H(\aY | \sX) - H(Y | \sX)|
\leq \;  H(F)\, .
\end{align}
It follows from Lemma \ref{lemma:tL_tail_control_fromHy}  and $L^* > 2\gamma\log(1/d)$ that
$\alpha \leq 80 d^{\gamma}/L^*$ for sufficiently small $d$.
Hence, $h(\alpha) \leq d^{\gamma-\epsilon}  \log L^* /L^*$ for $d < d_0( \epsilon)$,
for some $d_0( \epsilon) > 0$.
 Now Eqs.~\eqref{eq:hflips_bound} and \eqref{eq:hy_bound_flips} gives Eq.~\eqref{eq:I_vs_restricted}, where as Eq.~\eqref{eq:Hy_vs_restricted} follows by combining Eqs.~(\ref{eq:hflips_bound}),
(\ref{eq:hy_bound_flips}) and (\ref{eq:hygivenx_bound_flips}) to bound
$|I(\sX) - I(\asX)|$.

\end{proof}

\begin{proof}[Proof of Lemma \ref{lemma:achievability}]

 We first make some preliminary observations. Direct calculation leads
 to $H(\sX^\dagger)=H(p_L^\dagger)/\mu(\sX^{\dagger}) = 1 - O(d^2)$, and $|\mu(\sX^\dagger) -2|=O(d)$. From Lemma \ref{lemma:xy_similar_runs}(ii), we deduce $|\mu(\sY^\dagger) -2|=O(d)$.

Since $\sX^\dagger$ consists of independent runs, the same is true for
$\sY^\dagger$. Hence, recalling the notation $q_L^*(l) = 2^{-l}$, we have
\begin{align*}
H(\sY^\dagger) &= H(q_L^\dagger)/\mu(\sY^\dagger) = 1 - D(q_L^\dagger|| \{2^{-l}\})/\mu(\sY^\dagger)\nonumber \\
&= 1- \frac{1}{\mu(\sY^\dagger)} \sum_{l=1}^\infty q_L^\dagger(l) \big ( \log q_L^\dagger(l) + l\big ) \, .
\end{align*}

Define $\ell \equiv \lfloor 4 \log (1/d) \rfloor$. It follows from Lemma \ref{lemma:Ydagger_exponential_decay} that $\sum_{l=\ell+1}^\infty q_L^\dagger(l) l = O(d^{3})$, leading to
\begin{align}
H(\sY^\dagger) &\geq 1- \frac{1}{\mu(\sY^\dagger)} \sum_{l=1}^\ell q_L^\dagger(l) \big ( \log q_L^\dagger(l) + l\big )  + O(d^{3})\, .
\label{eq:Hydagger_lb_initial1}
\end{align}

Now, from Lemma \ref{lemma:xy_similar_runs}(i), we know that
\begin{align}
|q_L^\dagger(l) - p_L^\dagger(l)| \leq \const_2 d^{2-\eps/2}
\label{eq:qLdagger_near_pLdagger}
\end{align}
for $l<\ell$.

A Taylor approximation yields
\begin{align}
\sum_{l=1}^{\ell} q_L^\dagger(l) \big ( \log q_L^\dagger(l) + l \big)
&= \frac{1}{\ln2}\sum_{l=1}^{\ell} \bigg ( \left(q_L^\dagger(l) -
  2^{-l}\right ) + 2^{l-1} \left (q_L^\dagger(l) - 2^{-l}\right )^2
\bigg ) + O(d^{3-\eps}) \nonumber\\
&=\frac{1}{\ln2} \sum_{l=\ell+1}^{\infty}  \left(q_L^\dagger(l) -
  2^{-l}\right ) + \frac{2^{l-1}}{\ln 2}\sum_{l=1}^{\ell}  \left (q_L^\dagger(l) -
  2^{-l}\right )^2
 + O(d^{3-\eps}) \nonumber\\
&= \frac{d^2}{2\ln 2} \sum_{l=1}^{\ell}
    2^{-l} \left (-c_2l/2+ l \ln l \right)^2 + O(d^{3-\eps}) \nonumber\\
&= \frac{d^2}{2\ln 2} \sum_{l=1}^{\infty}
    2^{-l} \left (-c_2l/2+ l \ln l \right)^2 + O(d^{3-\eps}) \nonumber \\
&= \frac{d^2}{2\ln 2}\left( \frac{3}{2} c_2^2 + \sum_{l=1}^{\infty}
    2^{-l} \left ( \left (l \ln l \right)^2-c_2l^2 \ln l \right)\right) + O(d^{3-\eps}) \, . \label{eq:first term_Idagger}
\end{align}

Plugging back into Eq.~\eqref{eq:Hydagger_lb_initial1} and using $|\mu(\sY^\dagger)-2|= O(d)$, we obtain
\begin{align}
H(\sY^\dagger) &\geq 1- \frac{d^2}{4\ln 2}\left( \frac{3}{2} c_2^2 + \sum_{l=1}^{\infty}
    2^{-l} \left ( \left (l \ln l \right)^2-c_2l^2 \ln l \right)\right) + O(d^{3-\eps})  \, .
\label{eq:Hydagger_lb_initial}
\end{align}

We construct $\asX^\dagger \in \St_{\lfloor 1/d \rfloor}$ from
$\sX^\dagger$ by flipping a few bits as in the proof of Lemma
\ref{lemma:small_loss_by_restricting_runs}. The fraction of flipped
bits, both in $\sX^\dagger$ and in $\sY^\dagger$, is at most $\alpha=2
\E[\tL \ind(\tL \geq \lfloor 1/d \rfloor)]/\lfloor 1/d \rfloor \leq
O(2 ^{-d/2}) = O(d^4)$.
Proceeding as in the proof of Lemma \ref{lemma:small_loss_by_restricting_runs}, cf. Eqs.~\eqref{eq:hflips_bound} and \eqref{eq:hygivenx_bound_flips}, we have
\begin{align}
\big |   H(\aY^\dagger | (\aX^\dagger)^n)  - H(Y^\dagger | (X^\dagger)^n)   \big | \leq n h(\alpha) = n O(d^3)\, .
\label{eq:aYdagger_hyx_close}
\end{align}
For each bit that is flipped, the number of runs in $Y$ can change by at most $2$, and the number of runs of a particular length can change by at most $3$. It follows that
\begin{align*}
\left | \frac{1}{\mu(\sY^\dagger)} - \frac{1}{\mu(\asY^\dagger)} \right | \leq 2\alpha =  O(d^4)\, ,
\end{align*}
and, for any positive integer $l$,
\begin{align*}
\left | \frac{q_L^\dagger(l)}{\mu(\sY^\dagger)} - \frac{\aq_L^\dagger(l)}{\mu(\asY^\dagger)} \right | \leq 3\alpha =  O(d^4)\, .
\end{align*}
 We then deduce from the above that
\begin{align*}
\left |\mu(\sY^\dagger) - \mu(\asY^\dagger)\right | = O(d^4)\, ,
\end{align*}
and for any $l>0$,
\begin{align*}
\left |q_L^\dagger(l) - \aq_L^\dagger(l)\right | \leq \const_1 d^4\, ,
\end{align*}
where $\aq_L^\dagger(\cdot)$ is the distribution of runs under $\aY$.
From Eq.~\eqref{eq:qLdagger_near_pLdagger}, it follows that for $l< \ell$,
\begin{align}
\left |q_L^\dagger(l) - p_L^\dagger(l)\right | \leq 2 \const_2 d^{2-\eps/2}\, .
\label{eq:aqdagger_close}
\end{align}

We have $H(\aY^\dagger | (\aX^\dagger)^n)=H(\aY^\dagger, K^\dagger |
(\aX^\dagger)^n) - H( K^\dagger | (\aX^\dagger)^n,\aY^\dagger)$
where $K^\dagger\equiv K( (\aX^\dagger)^n)$.
We use  Corollary \ref{coro:hygivenx_q} and Lemma \ref{lemma:K_conditional_entropy}
to arrive at
\begin{align}
 \lim_{n \rightarrow \infty} \frac{1}{n} H(\aY^\dagger | (\aX^\dagger)^n) &=   d \log (1/d) - \frac{d}{2}  \sum_{l=2}^\ell \aq_L^\dagger(l) \; l \log l  +\frac{dc_2}{4 \ln 2} \sum_{l=1}^\ell \aq_L^\dagger(l) l  \nonumber\\
& + \left(1 - \frac{c_2}{2} \right) \frac{d}{\ln2} - \left (c_3 + c_4 + \frac{1}{2 \ln 2}\right) d^2 + O(d^{3-\eps})\, .
\label{eq:ahyx_estimate}
\end{align}

Combining Eqs.~\eqref{eq:aYdagger_hyx_close}, \eqref{eq:ahyx_estimate} and \eqref{eq:aqdagger_close}, we obtain,
\begin{align*}
\lim_{n \rightarrow \infty} \frac{1}{n} H(Y^\dagger | (X^\dagger)^n) &=   d \log (1/d) - \frac{d}{2}  \sum_{l=2}^\ell p_L^\dagger(l) \; l \log l  +\frac{dc_2}{4 \ln 2} \sum_{l=1}^\ell p_L^\dagger(l) l  \nonumber\\
& + \left(1 - \frac{c_2}{2} \right) \frac{d}{\ln2} - \left (c_3 + c_4 + \frac{1}{2 \ln 2}\right) d^2 + O(d^{3-\eps})\, .
\end{align*}

A calculation yields
\begin{align}
\lim_{n \rightarrow \infty} & \frac{1}{n} H(Y^\dagger | (X^\dagger)^n) = \nonumber \\
&d \log (1/d)
+\left(1 - \frac{c_2}{2} \right) \frac{d}{\ln2} \nonumber\\
&-  d^2 \left (c_3 + c_4 + \frac{1}{4 \ln 2}
    \left [2  +  3 c_2^2+ 2 \sum_{l=1}^\infty
        2^{-l} \big ( (l \ln l)^2 - c_2 l^2 \ln l \big )
    \right ]
\right) + O(d^{3-\eps})\, .
\label{eq:hyxdagger_estimate}
\end{align}

Finally,
\begin{align*}
I(\sX^\dagger) = (1-d) H(\sY^\dagger) + \lim_{n \rightarrow \infty} \frac{1}{n} H(Y^\dagger | (X^\dagger)^n)\, .
\end{align*}


The result now follows by using the estimates in Eqs.~\eqref{eq:first term_Idagger}
and Eq.~\eqref{eq:hyxdagger_estimate}.

We obtain
\begin{align*}
I(\sX^\dagger) \geq 1 - d\log(1/d) -A_1 d +A_2 d^2 + O(d^{3-\eps})\, ,
\end{align*}
where
\begin{align*}
A_1 &= \log (2e) - \frac{c_2}{2 \ln 2}\, ,\\
A_2 &=  - \frac{1}{4 \ln 2}  \left(
\frac{3}{2} c_2^2 + \sum_{l=1}^{\infty}
    2^{-l} \left ( \left (l \ln l \right)^2-c_2l^2 \ln l \right)\right) \\
&\phantom{= xx }+c_3 +c_4  +  \frac{1}{4 \ln 2}\left(2+3c_2^2+ 2 \sum_{l=1}^{\infty}
    2^{-l} \left ( \left (l \ln l \right)^2-c_2l^2 \ln l \right)\right)\\
&= c_3 +c_4+ \frac{1}{4 \ln 2}  \left( 2+ \frac{3}{2} c_2^2 +
\sum_{l=1}^{\infty}
    2^{-l} \left (l \ln l \right)^2- c_2\sum_{l=1}^{\infty}
    2^{-l}l^2 \ln l \right)\, .
\end{align*}
\end{proof}

\begin{proof}[Proof of Lemma \ref{lemma:hy_is_large}]
Let $\gamma_* = \sup\{\gamma: H(\sY) \geq 1-d^{\gamma} \}$. 
Then $\gamma_* \geq 1+ \gamma_*/2 - \eps/2$ must hold, else Lemma \ref{lemma:cyclic_hy_bound_SLstar} leads to a contradiction. It follows that $\gamma_* \geq 2 - \eps$, hence the result.

  We use here the fact that $d_0$ in Lemma \ref{lemma:cyclic_hy_bound_SLstar} does not depend on $\gamma$.
\end{proof}

\begin{proof}[Proof of Lemma \ref{lemma:converse_for_restricted_runs}]
Fix $\eps > 0$. Consider any $\sX \in \St_{\lfloor 1/d \rfloor}$.
Assume
\begin{align*}
I(\sX) &\geq 1 - d \log(1/d) -A_1 d -  d^{2-(\eps/8)}\, .
\end{align*}
(If not, we are done, for small enough $d$.)

By Lemma \ref{lemma:hy_is_large}, we know that $H(\sY)> 1-d^{2-(\eps/2)}$.  Now, we use Lemma \ref{lemma:hy_upper_bound}, Corollary \ref{coro:hygivenx_q} and Lemma \ref{lemma:K_conditional_entropy} for the three terms in Eq.~\eqref{eq:I_three_terms}, to arrive at
\begin{align}
I(\sX) \leq 1 - &d \log (1/d) - \frac{1}{2} \sum_{l=1}^\infty q_L(l) \big ( \log q_L(l) + l)\nonumber \\
&  + \frac{d}{2}  \sum_{l=2}^{4 \log(1/d)} q_{L}(l) \; l \log l  -\frac{dc_2}{4\ln 2} \sum_{l=1}^{4 \log(1/d)} q_L(l) l + \tc_1 d + \tc_2 d^2 + \const_1 d^{3-\eps}\, ,
\label{eq:Ix_upper_bd1}
\end{align}
where $\tc_1, \tc_2$ can be explicitly computed in terms of constants
above, and $\const_1<\infty$ is independent of $q_L$.
 The precise value
of these constants is irrelevant for the argument below.

Since we know that $\sX \in \St_{\lfloor 1/d \rfloor}$, Lemma \ref{lemma:q_exp_decay_in_St1byd} tells us that the tail of $q_L$ is small. Define
$\ell \equiv \lfloor 8/d \rfloor$. We deduce that
\begin{align*}
\sum_{l=\ell +1}^\infty q_{L}(l) &\leq d^{4}\, , \qquad \sum_{l=\ell +1}^\infty l q_{L}(l) \leq d^{4} \,  ,
\end{align*}
for small enough $d$.
From elementary calculus, we obtain
\begin{align}
\sum_{l=\ell+1}^\infty q_L(l) \big ( \log q_L(l) + l) \geq \sum_{l=\ell+1}^\infty q_L(l) \log \left (\frac{\sum_{l=\ell+1}^\infty q_L(l)}{2^{-\ell}} \right )\nonumber\\
\geq \ell d^4 + d^4 \log d^4 \geq d^{3-\eps/2} \, .
\label{eq:Dbound}
\end{align}
From Lemma \ref{lemma:L_tail_control}, we deduce 
\begin{align}
\sum_{l=4 \log(1/d)}^\ell q_L(l) l \leq d^{2 - \eps}\, .
\label{eq:ELY_bd}
\end{align}

Plugging the bounds in Eqs.~\eqref{eq:Dbound}, \eqref{eq:ELY_bd} into Eq.~\eqref{eq:Ix_upper_bd1}, we obtain
\begin{align*}
I(\sX) \leq 1 - &d \log (1/d) - \frac{1}{2} \sum_{l=1}^\ell q_L(l) \big ( \log q_L(l) + l)\nonumber \\
& + \frac{d}{2}  \sum_{l=2}^\ell q_{L}(l) \; l \log l  -\frac{dc_2}{4\ln 2} \sum_{l=1}^\ell q_L(l) l  + \tc_1 d + \tc_2 d^2 + \const_2 d^{3-\eps}\, ,
\end{align*}
where $\const_2<\infty$ is independent of $q_L$.

Now we simply maximize the
bound over `distributions'  $\{q_L(l)\}_{l=1}^{\ell}$ satisfying $\sum_{l \leq \ell}q_L(l) \leq 1$, to arrive at an
optimal distribution
\begin{align*}
q_L^*(l) = B(d) 2^{-l} 2^{d(-S l/2 + l \log l)} 
\end{align*}
for $l \leq \ell$, where $B(d)$ is such that $\sum_{l \leq \ell} q_L^*(l) = 1$, and $S = c_2/\ln 2$. Note that
$q_L^*(l)$ has no dependence on the process $\sX$ we started with.

It is easy to verify that
\begin{align*}
B(d)= 1  + O(d^{2-\eps/2})\, .
\end{align*}
This leads to
\begin{align*}
q_L^*(l) = \left \{
\begin{array}{ll}
2^{-l} \left ( 1+d(-c_2 l/2+ l \ln l) + O(d^{2-\eps/2})\right) \qquad &\mbox{for } l \leq \ell \\
2^{-l/2}O(1) \qquad &\mbox{otherwise.}
\end{array} \right .
\end{align*}
We now have
\begin{align}
I(\sX) \leq 1 - &d \log (1/d) - \frac{1}{2} \sum_{l=1}^\ell q_L^*(l) \big ( \log q_L^*(l) + l)\nonumber \\
& + \frac{d}{2}  \sum_{l=2}^\ell q_{L}^*(l) \; l \log l  -\frac{dc_2}{4\ln 2} \sum_{l=1}^\ell q_L^*(l) l + \tc_1 d + \tc_2 d^2 + \const_4 d^{3-\eps}\, ,
\label{eq:Ix_bound_optimal_qL}
\end{align}
for some $\const_4 < \infty$.
Again, calculus yields
\begin{align*}
\sum_{l=1}^{\lfloor 6\log(1/d) \rfloor} q_L^*(l) \big ( \log
q_L^*(l)  + l \big) &= \frac{d^2}{2\ln 2}\left( \frac{3}{2} c_2^2 +
\sum_{l=1}^{\infty}
    2^{-l} \left ( \left (l \ln l \right)^2-c_2l^2 \ln l \right)\right) +
    O(d^{3-\eps})\, .
\end{align*}

We substitute in Eq.~\eqref{eq:Ix_bound_optimal_qL} to get the result.
\end{proof}

\section{Discussion}
\label{sec:discussion}

The previous best lower bounds on the capacity of the deletion channel
were derived using first order Markov sources. In contrast, we found that the optimal coding scheme for small $d$ consists of independent runs with run length distribution $p_L^\dagger(l) = 2^{-l}(1+ d( l \log l -c_2l/2))$
This leads to the natural question \emph{How much `loss' do we incur if we are only allowed to use an input distribution that is a first order Markov source?}

%
%

The following theorem is fairly straightforward to prove using the results we have derived. It provides an upper bound on the rate achievable with a Markov source, and also a precise analytical characterization of the optimal Markov source for small $d$.
\begin{thm}
\label{thm:markov_loss}
Fix any $\eps>0$. Consider the class of first order Markov sources. There exists $\const< \infty$ and $d_0 \equiv d_0(\eps)>0$, such that for and any $\sX$ in this class,
\begin{align*}
I(\sX) \leq 1 - d \log(1/d) -A_1 d + A_2' d^2 + \const d^{3-\eps} \,
\end{align*}
holds for any $d< d_0$, where
\begin{align*}
A_2' &\equiv 2c_5^2/\ln 2 + c_3 + c_4 + 1/(2 \ln 2)\, ,\\
c_5 &\equiv \frac{\ln 2}{4}\sum_{l=1}^{\infty}\left \{ l(l-3)2^{-l}\log l \right \} \, .
\end{align*}
 Denote the symmetric first order Markov source with $\mathring{p}(d)\equiv \prob(X_i = b | X_{i-1}=b) = 1/2 + c_5 d$ for $b \in \{0,1 \}$, by $\mathring{\sX}$. We have
 \begin{align*}
 I(\mathring{\sX}) \geq 1 - d \log(1/d) -A_1 d + A_2' d^2 + \const d^{3-\eps} \, .
 \end{align*}
\end{thm}

Numerical evaluation yields $A_2' \approx  1.57796256$ and $c_5
\approx 0.60409609$. We have $A_2 - A_2'  \approx 0.10018339$,
implying that \emph{the restriction to Markov sources leads to a rate loss
  of $0.10018339\, d^2$ bits per channel use, with respect to the
  optimal coding scheme}.

\begin{remark}
\label{rem:lb_off_pt9dsq}
Lower bounds are derived in \cite{Drinea07lb} using Markov sources and `jigsaw' decoding. In this case we can show that the best achievable rate is
\begin{align*}
1 - d \log(1/d) -A_1 d + (A_2'-c_4) d^2 + O(d^{3-\eps}) \, ,
\end{align*}
and that $\mathring{\sX}$ achieves this rate to within $O(d^{3-\eps})$. Thus, the lower bounds in \cite{Drinea07lb} are off by $A_2 - A_2' - c_4 \approx 0.904 d^2$, to leading order.
\end{remark}

\begin{remark}
The utility of our asymptotic analysis is confirmed by considering the
prescription for the optimal
 optimal Markov source $\mathring{\sX}$ provided by Theorem
 \ref{thm:markov_loss}.
Drinea and Mitzenmacher  \cite{Drinea07lb} optimized numerically over Markov sources
obtaining, for instance, $p=0.53$  for $d=0.05$.
Our analytical prediction yields  $\mathring{p}(0.05)\approx 0.530204804$.
\end{remark}

In comparison, we have shown that $I(\sX^\dagger) = C -
O(d^{3-\eps})$. In fact, we conjecture that an
even stronger bound holds.
\begin{conj}
$I(\sX^\dagger) = C - \Theta(d^4)$
\end{conj}
The reasoning behind this conjecture is as follows: We expect the next
order correction to the optimal input distribution to be quadratic in
$d$. If $I(\sX)$ is a `smooth' function of the input distribution, a
change of order $d^2$ in the input distribution should
imply that $I(\sX)$ decreases by an amount $\Theta((d^2)^2)= \Theta(d^4)$ below capacity.

Our work leaves several open questions:
\begin{itemize}
\item Can the capacity be expanded as
\begin{align*}
C= 1- d \log(1/d) -A_1 d + A_2 d^2 + A_3 d^3 + A_4 d^4 + \ldots \,
\end{align*}
for small $d$? If yes, is this series convergent? In other words, is there a $d_0>0$ such that for all $d< d_0$, the infinite sum on the right has terms that decay exponentially in magnitude? We expect that the answer to both these questions is in the affirmative. We provide a very coarse reasoning for this below.

The analysis carried out in the present paper suggests that  the
optimal input distribution for $d < d_0$ does not have
`long range dependence'. In particular, we expect correlations to
decay exponentially in the distance between bits. Suppose we are
computing contribution to capacity due to `clusters' of $k$ nearby deletions. These `clusters' should correspond to $k$ deletions occurring within $2k+1$ consecutive runs. This should give us a term $A_k d^k$ with the error being bounded by the probability of seeing $(k+1)$ deletions in $2k+1$ consecutive runs. This error should decay exponentially in $k$ for $d< d_0$, assuming our hypothesis on correlation decay.

\item What is the next order correction to the optimal input distribution? It appears that this correction should be of order $d^2$ and should involve non-trivial dependence between the run length distribution of consecutive runs. It would be illuminating to shed light on the type of dependence that would be most beneficial in terms of maximizing rate $I(\sX)$ achieved. Moreover, it appears that computing this correction heuristically may, in fact, be tractable, using some of the estimates derived in this work.

\item Can the results here be generalized to other channel models of insertions/deletions?

\item What about the deletion channel in the large deletion probability regime, i.e., $d \rightarrow 1$? What is the best coding scheme in this limit? It seems this limit may be harder to analyze than the $d \rightarrow 0$ limit studied in the present work: For $d=1$ the channel capacity is 0 and there is no particular coding scheme that we can hope to modify slightly in order to achieve good performance for $d$ close to 1. This is in contrast to the case $d=0$, where we know that the i.i.d. Bernoulli$(1/2)$ input achieves capacity.

\item Can a similar series expansion approach be used to `solve' other hard channels in particular asymptotic regimes of interest?

\item We did not compute explicitly the constants in the error terms of our upper and lower bounds, thus preventing us from numerically evaluating our upper and lower bounds on capacity (cf. Remark \ref{rem:constants_not_computed}). It would be interesting to compute constants for the error terms leading to improved numerical bounds on capacity.
\end{itemize}

\vskip4pt

{\bf Acknowledgments.}
Yashodhan Kanoria is supported by
a 3Com Corporation Stanford Graduate Fellowship. This research was
supported by NSF, grants
CCF-0743978 and CCF-0915145, and a Terman fellowship.

\appendix

\section{Proofs of Preliminary results}
\label{app:prelim_results}

\begin{proof}[Proof of Theorem \ref{thm:cap_limit}]
This is just a reformulation of Theorem 1 in \cite{Dobrushin},
to which we  add the remark $C = \inf_{n\ge 1} C_n$, which is of independent
interest.
In order to prove this fact, consider the channel $W_{m+n}$,
and let $X^{m+n}= (X_1^m,X_{m+1}^{m+n})$ be its input.
The channel $W_{m+n}$ can be realized as follows.
First the input is passed through a channel
$\tW_{m+n}$ that introduces deletions independently in the two strings
$X_{1}^m$ and $X_{m+1}^{m+n}$ and outputs
$\tY(X_1^{m+n})\equiv (Y(X_{1}^m),|,Y(X_{m+1}^{m+n}))$
where $|$ is a marker. Then the marker is removed.

This construction proves that $W_{m+n}$ is  physically degraded
with respect to $\tW_{m+n}$, whence
\begin{eqnarray*}
(m+n)C_{m+n}&\le &\max_{p_{X^{m+n}}} I(X^{m+n};\tY(X_{1}^{m+n}))\\
&\le & mC_m+nC_n\, .
\end{eqnarray*}
Here the last inequality follows from the fact that $\tW_{m+n}$ is the product
of two independent channels, and hence the mutual information is maximized
by a product input distribution.

Therefore the sequence $\{nC_n\}_{n\ge 1}$ is superadditive, and
the claim follows
from Fekete's lemma.
\end{proof}

\begin{proof}[Proof of Lemma \ref{lemma:stationary_suffices}]
Take any stationary $\sX$, and let $I_n= I(X^n; Y(X^n))$. Notice that
 $Y(X_1^n)-X_1^n-X_{n+1}^{n+m}-Y(X_{n+1}^{n+m})$ form a Markov chain.
Define $\tY(X^{n+m})$ as in the proof of Theorem \ref{thm:cap_limit}.
We therefore have $I_{n+m} \leq I(X^{n+m};\tY(X^{n+m})) \leq I(X_1^m;\tY(X_1^m)) +I(X_{m+1}^{m+n};Y(X_{m+1}^{m+n})) = I_m+I_n$.
(the last identity follows by
stationarity of $\sX$). Thus $I_{m+n}\leq I_n+I_m$ and the limit
$\lim_{n\to\infty}I_n/n$
exists by Fekete's lemma, and is equal to $\inf_{n\geq 1}I_n/n$.

Clearly, $I_{n} \leq C_n$ for all $n$.
Fix any $\ve>0$. We will construct a process $\sX$ such that
\begin{align}
 I_{N}/N\geq  C - \ve \qquad \forall \; N>N_0(\ve)\, ,
\label{eq:liminf_closeto_limsup}
\end{align}
thus proving our claim.

Fix $n$ such that
$C_n \geq  C - \ve/2$. Construct $\sX$ with
i.i.d. blocks of length $n$ with common distribution $p^*(n)$ that
achieves the supremum in the definition of $C_n$.
In order to make this process stationary,
we make the first complete block to the
right of the position $0$ start at position
$s$ uniformly random in $\{1,2,\dots,n\}$.
We call the position $s$ the offset.
The resulting process is clearly stationary and ergodic.

Now consider $N=kn+r$ for some $k \in \naturals$
and $r \in \{0, 1, \ldots, n-1\}$. The vector $X_1^N$ contains
at least $k-1$ complete blocks of size $n$, call them $x(1), x(2),
\ldots, x(k-1)$ with $x(i) \sim p^*(n)$. The block $x(1)$
starts at position $s$. There will be  further $r+n - s+1$ bits
at the end, so that $X_1^N=(X_1^{s-1}, x(1), x(2), \ldots, x(k-1),
X_{s+kn}^N)$.
We write $y(i)$ for $Y(x(i))$.
Given the output $Y$, we define
$\tY= ( Y(X_1^{s-1})\bb  y(1) \bb y(2) \bb \ldots \bb y(k-1)\bb Y(X_{s+(k-1)n}^N))$,
by introducing $k$ synchronization symbols $\bb$.
There are at most $(n+1)^k$ possibilities for $\tY$ given $Y$
(corresponding to potential placements of synchronization symbols).
Therefore we have
\begin{align*}
H(Y) &= H(\tY) - H(\tY|Y)\\
&\geq H(\tY) - \log((n+1)^k)\\
&\geq (k-1)H(y(1)) - k\log(n+1)\, ,
\end{align*}
where we used the fact that the $(x(i),y(i))$'s are i.i.d..
Further
\begin{align*}
H(Y|X^N) \leq H(\tY|X^N) \leq (k-1)H(y(1)|x(1)) + 2n\, ,
\end{align*}
where the last term accounts for bits outside the blocks.
We conclude that
\begin{align*}
I(X^N;Y(X^N)) &= H(Y) - H(Y|X^N)\\
&\geq (k-1)nC_n - k \log(n+1) - 2n\\
&\geq N(C_n - \ve/2)
\end{align*}
provided $\log(n+1) /n < \ve / 10 $ and $N > N_0\equiv 10n/\ve$.
Since $C_n\ge C-\ve/2$, this in turn implies
Eq.~(\ref{eq:liminf_closeto_limsup}).
\end{proof}

\section{Proofs of Lemmas in Section \ref{subsec:run_charac}}
\label{app:run_charac}

\begin{proof}[Proof of Lemma \ref{lemma:L_tail_control}]
Combining (\ref{eq:run_hx_upper_bd}), Lemma \ref{lemma:mean_closeto2} and (\ref{eq:Lentropy})
it follows that for small enough $d$, we must have
\begin{align}
D(p_L || p_L^*) \leq 3 d^{\beta}
\label{eq:D_is_small}
\end{align}
to achieve $H(\sX) \geq 1- d^{\beta}$.
Now define $\Delta \equiv \sum_{l=l_0}^\infty\  l p_L(l)$.
Take $\alpha = e^{3/5}$.
We have
\begin{align*}
\sum_{l=l_0}^\infty \frac{l}{\alpha^l} = \frac{l_0 \alpha^{-l_0}}{(1-\alpha)^2}
 < d^{\beta}
\end{align*}
for sufficiently small $d$, since $\alpha^{-l_0} \approx \exp \left \{ \frac{6}{5} \beta \log d \right \}$.
Thus,
\begin{align*}
&\sum_{l=l_0}^\infty l (p_L(l) -\alpha^{-l}) &\geq  \Delta - d^{\beta}\nonumber \\
\Rightarrow &\sum_{l \in \I}l p_L(l) &\geq   \Delta - d^{\beta}
\end{align*}
where $\I = \{l: l \geq l_0, p_L(l) \geq \alpha^{-l} \}$.

This yields,
\begin{align}
\sum_{l \in \I} p_L(l) \log \frac{p_L(l)}{p_L^*(l)} \geq \sum_{l \in \I} lp_L(l) \log \frac{2}{\alpha} \geq \log(2/\alpha)( \Delta - d^{\beta})
\label{eq:I_contri_to_D}
\end{align}

It remains to show that the sum of terms from outside $\I$ is not too small.
By Markov inequality, we have
\begin{align}
&\sum_{l \in \I} p_L(l) &\leq \Delta / l_0 \nonumber\\
\Rightarrow &\sum_{l \notin \I} p_L(l) &\geq 1 -\Delta / l_0
\label{eq:notI_sum_is_large}
\end{align}
 With a fixed sum constraint on $(p_L(l), l \notin \I)$, the smallest value of
$\sum_{l\notin \I} p_L(l) \log \frac{p_L(l)}{p_L^*(l)}$ is achieved when
\begin{align}
\frac{p_L(l)}{p_L^*(l)} = \const= \frac{\sum_{l \notin \I} p_L(l)}{\sum_{l \notin \I} 2^{-l}} \qquad \forall l \notin \I
\label{eq:notI_fixed_ratio}
\end{align}
Note that this ratio is smaller than $1$.
It follows from (\ref{eq:notI_fixed_ratio}) and (\ref{eq:notI_sum_is_large}) that for small $d$,
\begin{align}
\sum_{l\notin \I} p_L(l) \log \frac{p_L(l)}{p_L^*(l)} \geq \log(\sum_{l \notin \I} p_L(l)) \geq - 2 \Delta/l_0
\label{eq:notI_contri_to_D}
\end{align}
since we know that $\Delta \leq \mu(\sX)= 3$, and hence $\Delta/l_0 \leq 1/10$.
The lemma follows by combining (\ref{eq:I_contri_to_D}), (\ref{eq:notI_contri_to_D}) and
$D(p_L || p_L^*) \leq 3 d^{\beta}$.

\end{proof}

\begin{proof}[Proof of Corollary \ref{coro:Lk_tail_control}]
Clearly $L_1 + \ldots + L_k \geq k l_*$ occurs only if at least one of the $L_i$'s is at least
$l_*$. Also, the distribution $p_{L(k)}$ has a marginal $p_L$ for each individual $L_i$.
We have
\begin{align*}
&\sum_{l_1+ \ldots +l_k \geq kl_*}\  (l_1 + \ldots +l_k) p_{L(k)}(l_1, \ldots, l_k)\\
&\leq \sum_{i=1}^k \sum_{l_1+ \ldots +l_k \geq kl_* }  \mathbb{I}[l_i \mbox{ is the largest} ] \,
k l_i \, p_{L(k)}(l_1, \ldots, l_k)\\
&\leq \sum_{i=1}^k  \sum_{l_i=l_*}^\infty k l_i \, p_{L} (l_i)\\
&= k^2 \sum_{l=l_*}^\infty  l p_{L} (l)
\end{align*}
The result now follows from the first inequality in Lemma \ref{lemma:L_tail_control}.
\end{proof}

\begin{proof}[Proof of Lemma \ref{lemma:L_TVstrong}]
Repeat proof of Lemma \ref{lemma:L_TV}.
\end{proof}

\begin{proof}[Proof of Proposition \ref{propo:Y_stat_ergodic}]
A time shift by a constant in $\sY$ corresponds to a time shift by a random amount in
$\sX$. The random shift in $\sX$ depends only on the $\sD$ and is hence independent of $\sX$.
Also, $\sD$ is independent identically distributed.
Thus, stationarity of $\sX$ implies stationarity of $\sY$.
\end{proof}

\begin{proof}[Proof of Lemma \ref{lemma:pL_tail_control_fromHy}]
Consider a run $R$ of length $l \geq 2 l_0$ in $\sX$. With probability at least $(1-d)^2$, the
runs bordering $R$ do not disappear due to deletions. Independently, with probability
$\prob [ \textup{Binomial}(l, 1-d) \geq l/2 ]$ at least half the bits of $R$ survive deletion.
Thus, for small $d$, with probability at least $1/2$, $R$ leads to a run of length at least $l/2$ in
$\sY$. Moreover, runs can only disappear in going from $\sX$ to $\sY$. It follows that
\begin{align*}
\sum_{l=l_0}^\infty\  l q_L(l) \ &\geq \sum_{l=2l_0}^\infty\  \left (\frac{l}{2} \right )\left ( \frac{p_L(l)}{2}\right )\, .
\end{align*}
From Lemma \ref{lemma:L_tail_control} applied to $\sY$, we know that
\begin{align*}
\sum_{l=l_0}^\infty\  l q_L(l) \ &\leq 20 d^{\beta} \, .
\end{align*}
The result follows.
\end{proof}

\begin{proof}[Proof of Corollary \ref{coro:pLk_tail_control_fromHy_STRONG}]
Analogous to proof of Corollary \ref{coro:Lk_tail_control}.
\end{proof}

\begin{proof}[Proof of Lemma \ref{lemma:xy_similar_runs}]
We adopt two conventions. First, when we use the $O(\cdot)$
or the $\Omega(\cdot)$ notation, the constant involved does not depend on
the particular $\sX, \sY$ under consideration. Second, we use `typical' in this proof to refer to events
having a probability $\Omega(d^{2 - \delta})$, for some $\delta>0$. Thus, an event with probability $2d^2$ is not
typical, but an event with probability $d^{1.5}$ is typical.

We ignore boundary effects due to runs at the beginning and end.

First, we estimate the factor due to disappearance of runs in moving from $\sX$ in $\sY$.
Define
\begin{align*}
r(\sX) \equiv \lim_{n \rightarrow \infty} \, \frac{\mbox{Number of runs in }Y(X^n)}{\mbox{Number of runs in }X^n}
\end{align*}
We have almost sure convergence of this ratio to a constant value  due to ergodicity.

Runs disappear typically due
to runs of length $1$ being deleted, and the runs at each end being fused with each other (i.e. neither of them is deleted).
Such an event reduces the number of runs by $2$.
Non-typical run deletions lead to a correction factor that is $O(d^2)$.
Hence, the expected number of runs in $Y$ per run
in $X^n$ is $ 1 - 2 p_L(1) d + O(d^2)$. It follows from a limiting argument that
\begin{align}
\label{eq:Z_intermsof_pL1}
r = 1 - 2 p_L(1) d + O(d^2)
\end{align}
In this proof, we make use of the following implication of Lemma \ref{lemma:L_TVstrong}.
\begin{align}
\left | p_{L(k)}(l_1, \ldots, l_k) - 2^{-\sum_{i=1}^k l_i} \right | \leq \const' \sqrt{k} d^{\beta/2}
\label{eq:L_Vstrong}
\end{align}
We immediately have $p_L(1)=1/2+O(d^{\beta/2})$ and hence
$r = 1- d + O(d^{1+\beta/2})$.

Consider $q_L(1)$. Blocks of length $1$ in $Y$ typically arise due to blocks in $\sX$ of length $1$ or $2$.
In case of a block of length $1$, we require that it isn't deleted, and also that bordering blocks
 are not deleted. Consider a randomly selected run in $\sX$ (Formally,
we pick a run uniformly at random in $X^n$ and then take the limit $n \rightarrow \infty$).
The run has length $L=1$ with probability $p_L(1)$. Define
\begin{itemize}
\item $\Ev_1 \equiv$ No bordering block of length $1$.  We have $\prob[\Ev_1,L=1]=(1/8)+ O(d^{\beta/2})$.
\item $\Ev_2 \equiv$ One bordering block of length $1$. We have $\prob[\Ev_2,L=1]=(1/4)+ O(d^{\beta/2})$.
\item $\Ev_3 \equiv$ Two bordering blocks of length $1$. We have $\prob[\Ev_3,L=1]=(1/8)+ O(d^{\beta/2})$.
\end{itemize}
 Probabilities were estimated using $p_L(1) = 1/2 + O(d^{\beta/2})$, $p_{L(2)}(1,1) = 1/4 + O(d^{\beta/2})$ and $p_{L(3)}(1,1,1) = 1/8 + O(d^{\beta/2})$, and their immediate consequences $p_{L(3)}(1,1,>1) = 1/8 + O(d^{\beta/2})$, $p_{L(3)}(>1,1,1) = 1/8 + O(d^{\beta/2})$ and $p_{L(3)}(>1,1,>1) = 1/8 + O(d^{\beta/2})$. We made of Eq.~\eqref{eq:L_Vstrong}.

Probability of arising from block of length $1$ is
\begin{align*}
&\phantom{=} (1-d) \, \big \{
\prob[\Ev_1,L=1](1-O(d^2)) +
\prob[\Ev_2,L=1](1-d)(1-O(d^2)) +
\prob[\Ev_3,L=1](1-d)^2
  \big \} \\
&= p_L(1) (1 - 2d) + O(d^{1+\beta/2})
\end{align*}
Probability
of arising from a block of length $2$ is $p_L(2) 2d + O(d^2) = d/2 + O(d^{1+\beta/2})$, using
Eq.~\eqref{eq:L_Vstrong}. It follows that
\begin{align*}
q_L(1) = \frac{p_L(1)(1 - 2d) + d/2 + O(d^{1+\beta/2})}{r} = p_L(1) + O(d^{1+\beta/2})
\end{align*}
as required.

Now consider $q_L(l)$ for $1 < l < \const \log(1/d)$. Typical modes of creation of such a run in $\sY$ are:
\begin{enumerate}
\item Run of length $l$ in $\sX$ that goes through unchanged.
\item Two runs in $\sX$ being fused due to the length $1$ run between them being deleted. Fused
runs have no deletions. They have $l$ bits in total.
\item Run of length $l+1$ in $\sX$ that suffers exactly one deletion. Bordering runs do not
disappear.
\end{enumerate}
For mode 1, we define events $\Ev_1, \Ev_2, \Ev_3$ as above. Probability estimates are:
\begin{itemize}
\item $\prob[\Ev_1,L=l] = 2^{-l-2}+ O(d^{\beta/2})$.
\item $\prob[\Ev_2,L=l] = 2^{-l-1} + O(d^{\beta/2})$.
\item $\prob[\Ev_3,L=l] = 2^{-l-2} + O(d^{\beta/2})$.
\end{itemize}
using Eq.~\eqref{eq:L_Vstrong} as we did for $L=1$. Thus, probability of creation from randomly selected run
via mode 1 is
\begin{align*}
&\phantom{=} (1 - d)^l \, \big \{
\prob[\Ev_1,L=l](1-O(d^2)) +
\prob[\Ev_2,L=l](1-d)(1-O(d^2)) +
\prob[\Ev_3,L=l](1-d)^2
  \big \} \\
&= p_L(l) - 2^{-l}(l+1)d + O(d^{1+\beta/2-\eps})
\end{align*}
for any $\eps>0$, since $l < \const \log(1/d)$.

The probability of a random set of three consecutive runs being such that the middle run has length $1$
and bordering runs have total length $l$ is $(l-1) 2^{-l-1} + O(d^{\beta/2-\eps})$ using Eq.~\eqref{eq:L_Vstrong} and $l < \const \log(1/d) < d^{-\eps}$ for small enough $d$.
Probability of the middle run being deleted and the other two runs being left intact, along with
bordering runs of this set of three runs not being deleted, is $d + O(l d^2)$. Thus, probability
of creation via mode 2 is $(l-1) 2^{-l-1}d + O(d^{1+\beta/2-\eps})$.

It is easy to check that the probability of mode 3 working on a randomly selected run is $(l+1) \,2^{-l-1} d+
O(d^{1+\beta/2})$.

Combining, we have
\begin{align*}
q_L(l) &= r^{-1}\left \{ p_L(l) -  2^{-l}(l+1)d + (l-1) 2^{-l-1}d + (l+1) \,2^{-l-1} d + O(d^{1+\beta/2-\eps})  \right \} \nonumber\\
&= p_L(l) +O(d^{1+\beta/2-\eps})
\end{align*}
This completes the proof of (i).

For (ii), simply note that
\begin{align*}
\frac{\mu(\sX)}{\mu(\sY)} = r(\sX) \times \lim_{n \rightarrow \infty} \frac{n}{\textup{Length of } Y(X^n)} = \frac{r(\sX)}{1-d}
\end{align*}
It follows from Eq.~\eqref{eq:Z_intermsof_pL1} that
\begin{align}
\label{eq:muXminusmuY_intermsof_pL1}
|\mu(\sX) - \mu(\sY)| \leq 4 \big |p_L(1) - 1/2\big | d \, + \, \const_3 d^2
\end{align}
for some $\const_3 < \infty$.
Eq.~\eqref{eq:muX_close2_muY} follows using Lemma \ref{lemma:L_TV} to bound $p_L(1)$.
\end{proof}

\begin{proof}[Proof of Lemma \ref{lemma:xy_similar_runs_kSTRONG}]
Similar to proof of Lemma \ref{lemma:xy_similar_runs}(i). We use Eq.~\eqref{eq:L_Vstrong} again, and make use of $k \leq \sum_{i=1}^k l_i \leq \const \log (1/d)$ to deduce that $\sqrt{k+2} \leq d^{-\eps/2}$ for small enough $d$.
\end{proof}

\begin{proof}[Proof of Lemma \ref{lemma:qL_tail_control_fromHx}]
From Lemma \ref{lemma:xy_similar_runs}(ii), we know that
\begin{align}
\left | \sum_{l=1}^\infty  lp_L(l) - \sum_{l=1}^\infty lq_L(l) \right |\leq \const_1 d^{1+\beta/2}
\label{eq:qLHx_1}
\end{align}

Recall $l \equiv \lfloor 4 \log(1/d)\rfloor$. Using Lemma \ref{lemma:xy_similar_runs}(i), we deduce
\begin{align}
\left | \sum_{l=1}^{\ell-1}  lp_L(l) - \sum_{l=1}^{\ell-1} lq_L(l) \right |\leq \const_2 d^{1+\beta/2 -\eps/2}
\label{eq:qLHx_2}
\end{align}

From Lemma \ref{lemma:L_tail_control}, we know that
\begin{align}
 \sum_{l=\ell}^\infty  lp_L(l) \leq \const_3 d^{\beta}
 \label{eq:qLHx_3}
\end{align}
Note that $\const_1, \const_2, \const_3$ do not depend on $\beta$.

Combining Eqs.~\eqref{eq:qLHx_1}, \eqref{eq:qLHx_2} and \eqref{eq:qLHx_3}, and using $\beta \leq 2$, we arrive at the desired result.

\end{proof}

\begin{proof}[Proof of Lemma \ref{lemma:xy_similar_runs_gammaSTRONG}]
By Lemma \ref{lemma:L_TVstrong} applied to $\sY$, we know that
\begin{align*}
\sum_{l_1=1}^\infty \sum_{l_2=1}^\infty \ldots \sum_{l_k=1}^\infty
\left|q_{L(k)}(l_1, l_2, \ldots, l_k) - p_{L(k)}^*(l_1, \ldots, l_k ) \right| \leq
\const_5 \sqrt{k}\,d^{\gamma/2}\, .
\end{align*}
Using Lemma \ref{lemma:xy_similar_runs_kSTRONG}, we have for $d<d_0(\const,\gamma)$, for any integer $k$ and $(l_1, \ldots, l_k)$ such that $\sum_{i=1}^k l_i < \const \log(1/d)$.
\begin{align*}
 \left|p_{L(k)}(l_1, l_2, \ldots, l_k) - q_{L(k)}(l_1, l_2, \ldots, l_k) \right| \leq
\const_6 \,d\, .
\end{align*}
Thus, we obtain Eq.~\eqref{eq:L_Vstrong_fromY},
using $k<\const \log(1/d)< d^{-\eps}$ for small $d$.
Eq.~\eqref{eq:L_Vstrong_fromY} follows. Also, note that we can deduce
\begin{align}
|p_L(1) - p_L^*(1)| \leq 2 \const_5 d^{\gamma/2}
\label{eq:pL1_bound_using_Hy}
\end{align}
for small enough $d$.
We repeat the proof of Lemma \ref{lemma:xy_similar_runs}(i) (or Lemma \ref{lemma:xy_similar_runs_kSTRONG}), using Eq.~\eqref{eq:L_Vstrong_fromY} instead of Eq.~\eqref{eq:L_Vstrong} to obtain Eq.~\eqref{eq:xy_similar_runs_gammaSTRONG}. This completes
the proof of (i).

For (ii), we proceed as follows to prove Eqs.~\eqref{eq:muX_close2_two_using_Hy} and \eqref{eq:muX_close2_muY_using_Hy}. In the proof of Lemma \ref{lemma:xy_similar_runs}(ii), we deduced that $|\mu(\sX) - \mu(\sY)| \leq 4 \big |p_L(1) - 1/2\big | d \, + \, \const_7 d^2$
(this is Eq.~\eqref{eq:muXminusmuY_intermsof_pL1} with the constant renamed). Using Eq.~\eqref{eq:pL1_bound_using_Hy} to bound $p_L(1)$, we obtain Eq.~\eqref{eq:muX_close2_muY_using_Hy}. From Lemma \ref{lemma:mean_closeto2} applied to $H(\sY)$, we know that $|\mu(\sY) - 2|\leq 7 d^{\gamma/2}$. Eq.~\eqref{eq:muX_close2_two_using_Hy} follows.
\end{proof}

\begin{proof}[Proof of Lemma \ref{lemma:q_exp_decay_in_St1byd}]
Associate each run in $\sY$ with the run in $\sX$ from which its first bit came. Consider any run $R_P$ in $\sX$. If it gives rise to a run in $\sY$ of length $\lambda \lfloor 1/d \rfloor$, then we know that the runs $R_{P+1}, R_{P+3}, \ldots, R_{P+2\lfloor \lambda -0.1 \rfloor-1}$ were all deleted (since $\sX \in \St_{\lfloor 1/d \rfloor}$). This occurs with probability at most $d^{\lfloor \lambda -0.1 \rfloor}$. Further, for each run in $\sX$, there are $\mu(\sX)(1-d)/\mu(\sY)$. This implies
\begin{align*}
q_L(\lambda \lfloor 1/d \rfloor) \leq \frac{\mu(\sY)}{\mu(\sX)(1-d)} d^{\lfloor \lambda -0.1 \rfloor}
\end{align*}
From Lemmas \ref{lemma:mean_closeto2} and \ref{lemma:xy_similar_runs}(ii), we know that
$|\mu(\sX)-2|<0.1$ and $|\mu(\sY)-2|<0.1$ for small enough $d$. Plugging into the above equation yields the desired result.
\end{proof}

\begin{proof}[Proof of Lemma \ref{lemma:tmu_closeto4}]
We make use of Eq.~\eqref{eq:superrun_hX_upper_bd}. Maximizing $H(\tT)$ for fixed $\tmu$, it is not hard to deduce that
\begin{align}
\frac{H(\tT)}{\tmu} &\leq f(\tmu) \label{eq:bound_tmu}\\
\textup{where } f(x) &\equiv  - \frac{2}{x} - \left(1- \frac{2}{x}\right) \log (x -2) + \log x
\nonumber
\end{align}
with equality iff $\sX$ consists of i.i.d. super-runs with $p_{\tT}(\lr, l-\lr) = (\lambda -1)^2 \lambda^{-l}$ where $\lambda = \tmu/(\tmu-2)$. Now, using Eq.~\eqref{eq:superrun_hX_upper_bd}, $H(\sX) \leq H(\tT)/\tmu$, and Eq.~\eqref{eq:bound_tmu}, we know that we must have $f(\tmu) \geq 1- d^{-\beta}$. Now, we have $f(4)=1$. Further, it is easy to check that $f(\cdot)$ achieves its unique global and local maximum at $4$, increasing monotonically before that and decreasing monotonically after that. It follows that for any fixed $\eps>0$, for small enough $d$, we must have $|\tmu - 4|\leq \eps$. It then follows from Taylor's theorem that
$f(\tmu)\leq 1 - (\tmu -4)^2/15$, so that we must have $|\tmu -4| \leq 4 d^{\beta/2}$ for $d \leq d_0$, where $d_0>0$.
\end{proof}

\begin{proof}[Proof of Lemma \ref{lemma:tL_tail_bound}]
An explicit calculation yields
\begin{align*}
H(\tT) = \tmu(\sX) - D(p_{\tT}||p_{\tT}^*)
\end{align*}
The proof now mirrors the proof of Lemma \ref{lemma:L_tail_control}, making use of Lemma \ref{lemma:tmu_closeto4} in place of Lemma \ref{lemma:mean_closeto2}.
\end{proof}

\begin{proof}[Proof of Lemma \ref{lemma:tL_tail_control_fromHy}]
It is easy to see that $f_{\sX}=\sum_{l=\ell}^\infty\ l p_{\tL}(l) /\tmu(\sX)$ is the asymptotic
fraction of bits in $\sX$ that are part of super-runs of length at least $\ell$.
Similarly, $f_{\sY} \sum_{l=\ell}^\infty\ l q_{\tL}(l) /\tmu(\sX)$ is the asymptotic
fraction of bits in $\sX$ that are part of super-runs of length at least $\ell$.

We argue that $f_{\sY}\geq 0.9 f_{\sX}$. Consider any bit $b_P$ at position $P$ in $\sX$ that is part of a super-run $S_i$ with length $\tL_i \geq \ell$. Consider
a contiguous substring of $S_i$ that includes $b_P$ of length exactly $\ell$. Clearly such a substring exists. The probability that it does not undergo any deletion is at least $1-\ell d \leq 0.9$ for small enough $d$. Further, if this substring does not undergo any deletion, then all bits in this substring are part of the same super-run in $\sY$, which must therefore have length at least $\ell$. It follows that bit $b_P$ is part of a super-run of length at least
$\ell$ in $\sY$ with probability at least $0.9$. Thus, we have proved $f_{\sY}\geq 0.9 f_{\sX}$. From Lemma \ref{lemma:tmu_closeto4}, it follows that $\tmu(\sX) \leq 5$ and
$\tmu(\sY) \geq 3$ for small enough $d$. Putting these facts together leads to the result.
\begin{align*}
\sum_{l=\ell}^\infty\  l p_{\tL}(l) \leq 5 f_{\sX} \leq 5 f_{\sY} / 0.9 \leq \frac{5}{0.9 \cdot 3 } \sum_{l=\ell}^\infty l q_{\tL}(l) \leq 80 d^{\gamma} \, ,
\end{align*}
where we have made use of Lemma \ref{lemma:tL_tail_bound} applied to $\sY$.
\end{proof}

\begin{proof}[Proof of Corollary \ref{coro:tLk_tail_control_fromHy_STRONG}]
Analogous to proof of Corollary \ref{coro:Lk_tail_control}.
\end{proof}

\section{Proof of Lemma \ref{lemma:K_conditional_entropy}}
\label{app:K_conditional_entropy_proof}

The proof of Lemma \ref{lemma:K_conditional_entropy} is quite intricate and requires us to 
define a new modified deletion process in terms of super-runs.

Now we define a new modification to the deletion process, we call it the perturbed deletion process to avoid confusion with the modified deletion process $\hsD$.

The input process $\sX$ is divided into super-runs as $\ldots, S_{-1}, S_0, S_1, \ldots$\,(cf. Definition \ref{def:super-run}).
For all integers $i$, define:
\begin{enumerate}[]
\item $\bsZ^{i} \equiv$
Binary process that is zero throughout except if $(S_i, S_{i+1}, S_{i+2}))$ have three or more deletions in total, in which case $\bZ^{i}_l=1$ if and only if $X_l \in S_i$ and $D_l=1$.
\end{enumerate}

Define
\begin{align*}
\bsZ = \sum_{i=-\infty}^\infty \bsZ^{i}
\end{align*}
where $\sum$ here denotes bitwise OR.
Finally, define $\bsD(\sD, \sX) \equiv \sD \oplus \bsZ $ (where $\oplus$ is componentwise
sum modulo $2$).
The output of the channel is simply defined by
deleting from $X^n$ those bits whose positions correspond to $1$s in
$\bsD$. We define $\bK$ for the modified deletion process similarly to $K$.

We make use of the following fact:
\begin{propo}
\label{propo:product_of_increasing}
Consider any integer $m>0$. Let $U_1, U_2, \ldots, U_m$ be random variables, taking values in $\naturals$, that have the same marginal distribution, i.e., $U_i \sim U$ for $i=1, 2, \ldots, m$, and arbitrary joint distribution. Let $f_1, f_2, \ldots, f_m:\naturals \rightarrow \reals_+$ be non-decreasing functions. Then we have
\begin{align*}
\E\bigg[\prod_{i=1}^m f_i(U_i)\bigg] \leq \E\bigg[\prod_{i=1}^m f_i(U)\bigg]
\end{align*}
\end{propo}
\begin{proof}[Proof of Proposition \ref{propo:product_of_increasing}]
We prove the result for $m=2$. The proof can easily be extended to arbitrary $m\in \naturals$.

We want to show that for random variables $U$ and $V$, with $U \sim V$, and non-decreasing, non-negative valued functions $f,g$, we have
\begin{align*}
\E[f(U) g(V)] = \E[f(U) g(U)]
\end{align*}

\noindent {\bf Part I:}

Define $\cH= \{f: \E[f(U)\ind(V\geq b)] \leq \E[f(U)\ind(U\geq b)], \ \forall b \in \reals\}$.

\noindent{\bf Claim:} The class $\cH$ contains all non-negative, non-decreasing functions $f$.

\noindent{Proof of Claim:}\\
\noindent(i) We have $\ind_{[a, \infty)} \in \cH, \forall a \in \reals$.
\begin{align*}
\E[\ind(U\geq a)\ind(V\geq b)] \leq \min \big \{ \prob(U\geq b), \prob(U\geq a) \big \}
= \prob(U \geq \max(a,b)) = \E[\ind(U\geq a)\ind(U\geq b)]
\end{align*}

\noindent(ii) If $f_1, f_2 \in \cH$ then $c_1 f_1 + c_2 f_2 \in \cH$ for any $c_1 >0, c_2>0$.\\
This follows from linearity of expectation.

Define the class of `simple increasing functions'
 \begin{align*}
 \cI \equiv \{f: \exists k \in \naturals \textup{ s.t. }f = \sum_{i=1}^k c_i \ind_{[a_i, \infty)} \textup{ for some } c_i>0, a_i \in \reals \textup{ for }i=1,2, \ldots, k\}
 \end{align*}

\noindent(iii) It follows from (i) and (ii) that $\cI \subseteq \cH$.

Now, it is not hard to see that for any non-negative non-decreasing $f$, we can find a monotone non-decreasing sequence of functions $(f_n)_{n=1}^\infty \in \cI$ such that $f_n \uparrow f$. By the monotone convergence theorem, we have
\begin{align*}
\lim_{n \rightarrow \infty}\E[f_n(U)\ind(V\geq b)] &= \E[f(U)\ind(V\geq b)] \, ,\\
\lim_{n \rightarrow \infty}\E[f_n(U)\ind(U\geq b)] &= \E[f(U)\ind(U\geq b)]\, .
\end{align*}
Combining with (iii), we infer that $f \in \cH$, proving our claim.\\

\noindent {\bf Part II:}

Define $\hcH_f= \{g: \E[f(U)g(V)] \leq \E[f(U)g(U)]\}$.

From Part I, we infer that $\ind(V\geq b) \in \hcH_f$ for all $b \in \reals$. We now repeat the steps in the proof of the Claim in Part I, to obtain the result ``The class $\hcH_f$ contains all non-negative, non-decreasing functions $g$." This completes our proof of the proposition.

\end{proof}

\begin{lemma}
\label{lemma:Kmod_conditional_entropy_intermsof_pLk}
There exists $d_0>0$ such that for any $d<d_0$ the following occurs: Consider any
$\sX \in \St_{\lfloor 1/d \rfloor}$.
Then
\begin{align}
&\lim_{n \rightarrow \infty} \, \frac{1}{n} \, H(\bK(X^n)|X^n,\bY(X^n)) = \frac{d^2}{\mu(\sX)} \bigg \{\nonumber\\
&\phantom{+}\sum_{k=2}^\infty \;\sum_{l_{k+1}=2}^\infty p_{L(k+2)}\big(1,1, \ldots (k+1 \textup{ ones}), l_{k+1}\big )  \, \big (k-1 +l_{k+1} \big )\, h\!\left( \frac{1}{k-1+l_{k+1}}\right)\nonumber\\
&+\sum_{l_0=2}^\infty \, \sum_{k=2}^\infty \;\sum_{l_{k+1}=2}^\infty p_{L(k+2)}\big(l_0,1, 1, \ldots (k \textup{ ones}), l_{k+1}\big )  \, \big (l_0 + k-1 +l_{k+1} \big )\, h\!\left( \frac{l_0+ 1}{l_0+k-1+l_{k+1}}\right)\nonumber\\
& \;\bigg\} \; + \;
\delta
\label{eq:Kmod_conditional_entropy_intermsof_pLk}
\end{align}
for some $\delta$ such that $|\delta| \leq 18 d^{3} \E [\tL^{2}]$.
\end{lemma}

\begin{proof}[Proof of Lemma \ref{lemma:Kmod_conditional_entropy_intermsof_pLk}]
Using the chain rule, we obtain
\begin{align*}
H(\bK(X^n)|X^n,\bY(X^n)) =  \sum_{j=1}^{M} H(|\bX(j)|\, | \bX(j) ... \bX(M), \bY(j) ... \bY(M)) \end{align*}

Consider the term $t_j\equiv H(|\bX(j)| | \bX(j) ... \bX(M), \bY(j) ... \bY(M))$. Suppose the first bit in $\bX(j) \ldots$ is part of super-run $S_i$. Call the first run in $\bX(j)$ be $R_P$. By the construction of the perturbed deletion process, we know that $S_i, S_{i+1}$ and $S_{i+2}$ cannot have more than two deletions in total.

Different cases may arise:
\begin{itemize}
\item $L_P > |\bY(j)|$\\
    If $L_P + L_{P+2} \geq |\bY(j)|$ then we know that $\bX(j)=(R_P, R_{P+1}, R_{P+2})$. If not, then we know that $\bX(j)=(R_P, R_{P+1}, R_{P+2}, R_{P+3}, R_{P+4})$. In either case, $t_j=0$.
\item $L_P > |\bY(j)|$\\
    It must be that $\bX(j)=R_P$. Again, $t_j=0$
\item $L_P = |\bY(j)|$\\
    In this case, if $L_{P+1} > 1$ or $L_{P+2}>1$, then we know that $\bX(j)=R_P$ and $t_j=0$. Suppose $L_{P+1}=L_{P+2}=1$. Now consider the possibility that $\bX(j)=(R_P, R_{P+1}, R_{P+2})$ (this is the only alternative to $\bX(j)=R_P$). For this possibility to exist, the following condition must hold
     \begin{align*}
       \cC \equiv \;&\left \{\textup{$\bY(j)\bY(j+1) \bY(j+2) \ldots$ must match exactly $R_{P}R_{P+3} R_{P+4} \ldots$}\right .\\
       &\left .\textup{until the end of $S_{i+2}$}\right \} \; \cap \, \{L_{P+1}=L_{P+2}=1\}
       \end{align*}
(Else, we would need more than two deletions in $(S_i, S_{i+1}, S_{i+2})$, a contradiction.)
\end{itemize}

Note that in any case, there are at most two possibilities for $\bX(j)$, so we have $t_j \leq 1$.

Let us understand $\cC$ better. Let $S_i$ include $k$ runs to the right of $R_P$, i.e.,
$L_{P+1}=L_{P+2}= \ldots = L_{P+k}=1$ and $L_{P+k+1}>1$. Condition $\cC$ can arise, along with $\bX(j)$ starting at $R_P$ iff:
\begin{itemize}
\item Runs $R_{P-1}$ does not disappear under $\bsD$.
\item Super-runs $(S_i, S_{i+1}, S_{i+2})$ undergo no more than two deletions in total. Event $\Ev$.
\item One of the following deletion patterns occur:
\begin{itemize}
\item (Only if $L_{P}>1$) The bit $R_{P+1}$ is deleted and one deletion in $R_P$. Event $\Ev_1$.
\item The bits $R_{P+1}$ and $R_{P+2}$ are deleted.     Event $\Ev_2$.
\item The bits $R_{P+2}$ and $R_{P+3}$ are deleted.     Event $\Ev_3$.\\
$\vdots$
\item The bits $R_{P+k-1}$ and $R_{P+k}$ are deleted.     Event $\Ev_k$.
\item The bit $R_{P+k}$ is deleted and one deletion in $R_{P+k+1}$.     Event $\Ev_{k+1}$.
\end{itemize}
\end{itemize}

Define $p_0 \equiv (1-d)^{\tL_i + \tL_{i+1}+\tL_{i+2}-2}$. It is easy to see that
$\prob(\Ev_1\cap \Ev) = p_0 d^2 L_P$, $\prob(\Ev_l\cap \Ev) = p_0 d^2$ for $l=2,3, \ldots, k-1$,
and $\prob(\Ev_k\cap \Ev) = p_0 d^2 L_{P+k+1}$. We know that exactly one of these has occurred.
$(\Ev_1\cap \Ev) \cup (\Ev_2\cap \Ev)$ leads to $\bX(j)=(R_P, R_{P+1}, R_{P+2})$, whereas all
other possibilities lead to $\bX(j)=R_P$. It follows that if $\cC$ holds, $L_P=l_P$ and
$L_{P+k+1}=l_{P+k+1}$,
\begin{align*}
t_j = h\left(\frac{l_P \ind (l_P>1) + 1 }{l_P \ind (l_P>1) + k -1 + l_{P+k+1} }\right) \, .
\end{align*}

Let $R_P$ be a uniformly random run (cf. Section \ref{sec:Preliminaries}). The probability of seeing $L_P=l_P$, $k$, $L_{P+k+1}=l_{P+k+1}$ and $(\Ev_1 \cup \Ev_2 \cup \ldots \cup \Ev_{k})\cap \Ev$ is
\begin{align*}
p_{L(k+2)}(l_P, 1, 1, \ldots (k \textup{ ones}), l_{P+k+1}) \; p_0 d^2\;  (l_P \ind (l_P>1) + k -1 + l_{P+k+1})
\end{align*}
where $p_0= (1-d)^{\tL_i + \tL_{i+1}+\tL_{i+2}-2}$
It is easy to see that $p_0 \in (1-d(\tL_i + \tL_{i+1}+\tL_{i+2}), 1)$. Also,  the conditional probability of $R_{P-1}$ not disappearing is in $(1-d,1)$. Thus the expected contribution of $R_P$ to the sum is
\begin{align*}
d^2 \bigg \{ \sum_{l_P=2}^\infty \sum_{k=2}^\infty \sum_{l_{P+k+1}=2}^\infty  & p_{L(k+2)}(l_P, 1, 1, \ldots (k \textup{ ones}), l_{P+k+1}) \;  \big(l_P \ind (l_P>1) + k -1 + l_{P+k+1}\big) \nonumber\\
 &\cdot h\left(\frac{l_P \ind (l_P>1) + 1 }{l_P \ind (l_P>1) + k -1 + l_{P+k+1} }\right)\bigg \} + \delta
\end{align*}
where $|\delta| \leq 2 d^3  E[(\tL_i + \tL_{i+1}+\tL_{i+2})^2]\leq 18 d^3 E[\tL^2]$, using Fact \ref{propo:product_of_increasing} in the final inequality.
The result follows.
\end{proof}

\begin{coro}
\label{coro:Kmod_conditional_entropy_estimate}
For any $\eps>0$, there exists $d_0\equiv d_0(\eps)>0$, and $\const < \infty$ such that for any $d<d_0$ the following occurs: Consider any
$\sX \in \St_{\lfloor 1/d \rfloor}$ such that $H(\sX)>1-d^{1-\eps}$ and $\max\{H(\sX),H(\sY)\}>1-d^{\gamma}$ for some $\gamma \in (1/2,2)$.
Then
\begin{align}
&\lim_{n \rightarrow \infty} \, \frac{1}{n} \, H(\bK(X^n)|X^n,\bY(X^n)) = \frac{d^2}{2} \bigg \{\nonumber\\
&\phantom{+}\sum_{k=2}^\infty \;\sum_{l_{k+1}=2}^\infty 2^{-(1+k + l_{k+1})}  \, \big (k-1 +l_{k+1} \big )\, h\!\left( \frac{1}{k-1+l_{k+1}}\right)\nonumber\\
&+\sum_{l_0=2}^\infty \, \sum_{k=2}^\infty \;\sum_{l_{k+1}=2}^\infty 2^{-(l_0+k + l_{k+1})}  \, \big (l_0 + k-1 +l_{k+1} \big )\, h\!\left( \frac{l_0+ 1}{l_0+k-1+l_{k+1}}\right)\nonumber\\
& \;\bigg\} \; + \;
\eta
\label{eq:Kmod_conditional_entropy_estimate}
\end{align}
for some $\eta$ such that $|\eta| \leq \const d^{2+\gamma/2-\eps/2}$.
\end{coro}
\begin{proof}[Proof of Corollary \ref{coro:Kmod_conditional_entropy_estimate}]
We prove the corollary assuming $H(\sY) > 1 -d^{\gamma}$. The proof assuming $H(\sX) > 1 -d^{\gamma}$ is analogous.

Consider the second summation in Eq.~\eqref{eq:Kmod_conditional_entropy_intermsof_pLk}.
Define $\ell \equiv \lfloor 4 \log(1/d) \rfloor$. Consider any term with $l_0 \leq \ell  $, $k \leq \ell $, $l_{k+1}\leq \ell$. Using Lemma \ref{lemma:xy_similar_runs_gammaSTRONG} (i) (Eq.~\eqref{eq:L_Vstrong_fromY}), we have
\begin{align*}
\big| p_{L(k+2)}\big(l_0,1, 1, \ldots (k \textup{ ones}), l_{k+1}\big ) - 2^{-(l_0+k + l_{k+1})} \big | \leq d^{\gamma/2-\eps/4}
\end{align*}
for $d<d_0(\eps)$. Note that $d_0$ does not depend on $l_0, k, l_{k+1}$. It follows that
\begin{align*}
\sum_{l_0=2}^\ell \, \sum_{k=2}^\ell \;\sum_{l_{k+1}=2}^\ell p_{L(k+2)}\big(l_0,1, 1, \ldots (k \textup{ ones}), l_{k+1}\big )  \, \big (l_0 + k-1 +l_{k+1} \big )\, h\!\left( \frac{l_0+ 1}{l_0+k-1+l_{k+1}}\right)\\
= \sum_{l_0=2}^\infty \, \sum_{k=2}^\infty \;\sum_{l_{k+1}=2}^\infty 2^{-(l_0+k + l_{k+1})}  \, \big (l_0 + k-1 +l_{k+1} \big )\, h\!\left( \frac{l_0+ 1}{l_0+k-1+l_{k+1}}\right) + \delta_{21}
\end{align*}
where $|\delta_{21}| \leq d^{\gamma/2-\eps/2}$.

We make use of Lemma \ref{lemma:tL_tail_control_fromHy} to bound the error due to the missed terms. Let $\tl_0$ be the length of the super-run containing the initial run of length $l_0$. Clearly, $\tl_0 \geq l_0 + k$. Let $\tl_1$ be the length of the next super-run to the right. Clearly, $\tl_1 \geq l_{k+1}$. Now
\begin{align*}
&\{l_0 > \ell\} \textup{ OR } \{k > \ell\} \textup{ OR } \{l_{k+1} > \ell\} \\
\Rightarrow &\{l_0 + k + l_{k+1} > \ell\}\\
\Rightarrow &\{\tl_0+\tl_1 > \ell\}
\end{align*}
Also, $\big (l_0 + k-1 +l_{k+1} \big ) \leq  \tl_0+\tl_1$ and $h(p) \leq 1$ for any $p$. It follows that the missed terms contribute
\begin{align*}
\delta_{22} \leq  \sum_{\tl_0 + \tl_1 \geq 4 \ell}
p_{\tL(2)}(\tl_0, \tl_1) \, \big(\tl_0+\tl_1\big)
\leq d^{\gamma/2-\eps/2}
\end{align*}
to the sum, where we have used Lemma \ref{lemma:tL_tail_control_fromHy} in the second inequality.

Thus, we have established
\begin{align*}
\sum_{l_0=2}^\infty \, \sum_{k=2}^\infty \;\sum_{l_{k+1}=2}^\infty p_{L(k+2)}\big(l_0,1, 1, \ldots (k \textup{ ones}), l_{k+1}\big )  \, \big (l_0 + k-1 +l_{k+1} \big )\, h\!\left( \frac{l_0+ 1}{l_0+k-1+l_{k+1}}\right)\\
= \sum_{l_0=2}^\infty \, \sum_{k=2}^\infty \;\sum_{l_{k+1}=2}^\infty 2^{-(l_0+k + l_{k+1})}  \, \big (l_0 + k-1 +l_{k+1} \big )\, h\!\left( \frac{l_0+ 1}{l_0+k-1+l_{k+1}}\right) + \delta_2
\end{align*}
with $|\delta_2| \leq 2 d^{\gamma/2-\eps/2}$ for $d< d_0(\eps)$. The first summation in Eq.~\eqref{eq:Kmod_conditional_entropy_intermsof_pLk} can be similarly handled. Finally, Lemma \ref{lemma:xy_similar_runs_gammaSTRONG}(ii) tells us that $|\mu(\sX)-2| \leq d^{\gamma/2}$ for small enough $d$. Putting the estimates together yields the result.
\end{proof}


\begin{proof}[Proof of Lemma \ref{lemma:K_conditional_entropy}]
We prove the lemma assuming $H(\sY) > 1 -d^{\gamma}$. The proof assuming $H(\sX) > 1 -d^{\gamma}$ is analogous.

It is easy to verify that the right hand side of Eq.~\eqref{eq:Kmod_conditional_entropy_estimate} is, in fact, $d^2 c_4  + \eta$. We show that
\begin{align}
\lim_{n \rightarrow \infty} \frac{1}{n}| H(\bK(X^n)|X^n,\bY(X^n)) - H({K}(X^n)|X^n,{Y}(X^n))| \leq d^{1+\gamma-\eps/2}
\label{eq:K_Khat_condH_close}
\end{align}
whence Eq.~\eqref{eq:K_conditional_entropy} follows using Corollary \ref{coro:Kmod_conditional_entropy_estimate}.

Consider $\bZ^n$
defined in our construction of the perturbed deletion process.  We define $U(X^n,D^n, Z^n) \in \{\t, 0,1\} ^{|\bY|}$ constructed as follows: Start from the first bit in $\bY$ and consider bits sequentially
\begin{itemize}
\item For each bit also present in $Y$, $U$ has a $\t$.
\item For each bit not present in $Y$, $U$ has $0$ if that bit $0$ and a $1$ if that bit is $1$.
\end{itemize}
Clearly, the corresponding stationary process $\sU$ can also be defined.

Recall Fact \ref{fact:U_V_W}. It is not hard to see that $(X^n, Y) \xleftrightarrow{U} (X^n, \widehat{Y})$ and
$(X^n, Y, K) \xleftrightarrow{(U,Z)} (X^n, \widehat{Y}, \widehat{K})$.
It follows that
\begin{align*}
| H(\widehat{K}(X^n)|X^n,\widehat{Y}(X^n)) - H({K}(X^n)|X^n,{Y}(X^n))| \leq 2H(U) + H(Z)
\end{align*}

Let $\bz\equiv \prob[\bZ_j = 1]$ for arbitrary $j$. The number of deletions reversed in a random super-run is at most $d^3 \sum_{l_0,l_1,l_2} p_{\tL(3)}(l_0,l_1,l_2) (l_0+l_1+l_2)^3$ in expectation (similar to Eq.~\eqref{eq:reversals_per_run}). Using Proposition \ref{propo:product_of_increasing}, this is bounded above by $27d^3 \E[\tL^3]$. Since each super-run  has length at least one, it follows that $\bz \leq 27 d^3 \E[\tL^3]$.
Using Lemma \ref{lemma:tL_tail_control_fromHy} and $\tL \leq 1/d$ w.p. 1, we find that $\E[\tL^3] \leq d^{\gamma-2}$ for small enough $d$. Hence, $\bz \leq 27 d^{1+\gamma}$. It follows that $H(\bsZ) \leq h(\bz) \leq d^{1+\gamma-\eps/2}$ for small enough $d$.

Let $u \equiv \prob(U_j \neq \t)$ for arbitrary $j$. Then $u = \bz/(1-d)$. It follows that $H(\sU) \leq u + h(u) \leq d^{1+\gamma-\eps/2}$ for small enough $d$. Finally, we have
\begin{align*}
\lim_{n \rightarrow \infty} \frac{2H(U) + H(Z)}{n} = 2(1-d)H(\sU) + H(\sZ) \leq 3 d^{1+\gamma-\eps/2}
\end{align*}
leading to the desired bound Eq.~\eqref{eq:K_Khat_condH_close}.
\end{proof}



\section{Proof of Lemma \ref{lemma:hatD_givenxy} and its corollaries}
\label{app:hatD_lemma_proof}

\begin{proof}[Proof of Lemma \ref{lemma:hatD_givenxy}]
We make use of \eqref{eq:modified_deletion_entropy} and the fact that $\sX$ is stationary and ergodic. Consider a randomly chosen run $R_\parent$ in $\sX$. We associate $H(\widehat{D}(j)|\hX (j), \widehat{Y}(j))$ with $R_\parent$ if $R_\parent$ is the first run in $\hX(j)$. Denote by $L_{P+i}$ the length of $R_{P+i}$ for any integer $i$. We add contributions from the three possibilities of how $\widehat{Y}(j)$ arose under $\widehat{D}(j)$:

\begin{enumerate}
\item {\bf From a single parent run}

Define
\begin{align*}
B_1 \equiv R_P \textup{ suffers one or two deletions under $\sD$ and } \exists j \textup{ s.t. } \hX(j)= R_P
\end{align*}
Clearly, $B_1$ is exactly the event we are interested in here.
We will restrict attention to a subset of $B_1$ and the prove that we are missing a very small contribution.
Define
\begin{align*}
E_1 \equiv B_1 \cap \{\textup{$R_{P-1}$ and $R_{P+1}$ do not disappear under $\sD$.}\}
\end{align*}
Consider $B_1 \backslash E_1$. For this event, one of the following must occur:
\begin{itemize}
\item Run $R_{P-1}$ disappears under $\sD$ but not under $\hsD$. For this, we need at least $3$ deletions in run $R_{P-1}$. A simple calculation shows that this occurs with probability less than $d^3 L_{-1}^3$.
\item Run $R_{P-1}$ disappears under $\hsD$ as well. In this case $R_{P-2}$ also disappears under $\hsD$. Thus, we need $R_{P-1}$ and $R_{P-2}$ both to disappear under $\sD$ which occurs with probability at most $d^2$. Moreover, we require at least one deletion in $R_P$ (probability less than $L_P d$). Thus, the overall probability is bounded above by $d^3 L_P$.
\item Run $R_{P+1}$ disappears under $\sD$ but not under $\hsD$. For this, we need at least $3$ deletions in run $R_{P+1}$. This occurs with probability less than $d^3 L_{1}^3$.
\end{itemize}
Thus, $0\leq \prob(B_1 \backslash E_1) < d^3 (L_{P-1}^3+L_P+L_{P+1}^3)$. The largest possible value of $H(\widehat{D}(j)|\hX (j), \widehat{Y}(j))$ for a particular occurrence of $B_1 \backslash E_1$ is $\max_{i=1,2} \log \binom{L_P}{i} \leq 2 \log L_P$. Thus, the additive error introduced by restricting to $E_1$ in our estimate of
$\lim_{n\to\infty} \frac{1}{n}H(\widehat{D}^n|X^n, \widehat{Y}, \widehat{K})$ is
\begin{align}
\label{eq:delta1E_bound}
0 \leq \delta_{1E}(d, \sX) \leq d^3 \E[2(L_{P-1}^3+L_P+L_{P+1}^3) \log L_P] \leq 6d^3 \E[L^3 \log L]
\end{align}
where we have made use of Proposition \ref{propo:product_of_increasing}.

Partition $E_1$ into two events:
\begin{align}
B_{11} &\equiv E_1 \cap \{R_P \textup{ undergoes one deletion under $\sD$} \}\\
B_{12} &\equiv E_1 \cap \{R_P \textup{ undergoes two deletions under $\sD$} \}
\end{align}
Let $T_1$ be the contribution of $B_{1}$, $T_{11}$ be the contribution of $B_{11}$
and $T_{12}$ be the contribution of $B_{12}$. Then we have
\begin{align}
T_1 = T_{11}+T_{12}+ \delta_{1E}
\label{eq:T1_split}
\end{align}

\begin{itemize}
\item {\bf One deletion in $R_\parent$:}\\
Consider $B_{11}$. The contribution of a particular occurrence is $\log L_P$.
Now
\begin{align}
\prob (B_{11}&, L_P=l, L_{P-1}=l_{P-1}, L_{P+1}=l_{P+1})  \nonumber\\
&= p_{L(3)}(l_{-1},l,l_{+1})\, (1-d^{l_{P-1}})\,(1-d^{l_{P+1}})\, l_Pd(1-d)^{l-1}
\label{eq:B11_particularl}
\end{align}
We have, for $l>1$,
\begin{align*}
p_{L(3)}(\gt1, l, \gt1) \, l d (1-d)^{l-1} \, (1-2d^2) \, \leq &\; \prob (B_{11}, L_P=l, L_{P-1}>1, L_{P+1}>1) \\
&\phantom{xxxxxxxx}\leq \, p_{L(3)}(\gt1, l, \gt1)\, l d (1-d)^{l-1}
\end{align*}
since probability that $R_{\parent-1}$ of length greater than $1$ disappears is bounded above by $d^2$ and similarly for $R_{\parent+1}$.
It follows that
\begin{align*}
\prob(B_{11}, L_P=l, L_{\parent-1}>1, L_{\parent+1}>1) \, = \,  p_{L(3)}(\gt1, l, \gt1)\, l d (1-(l-1)d ) + \eta_{1,1}(l)\\
-2d^3  p_{L(3)}(\gt1, l, \gt1)\, l  \, \leq \, \eta_{1,1}(l) \, \leq \, d^3 p_{L(3)}(\gt1, l, \gt1)\, l \binom{l-1}{2}
\end{align*}
Similarly we get
\begin{align*}
\prob(B_{11}, L_P=l, L_{\parent-1}=1, L_{\parent+1}=1) \, = \,  p_{L(3)}( 1, l, 1)\, l d (1-(l+1) d ) + \eta_{1,4}(l)\\
0  \, \leq \, \eta_{1,4}(l) \, \leq \, d^3 p_{L(3)}(1, l, 1)\, l \binom{l+1}{2}
\end{align*}
and
\begin{align*}
\prob(B_{11}, L_P=l, L_{\parent-1}>1, L_{\parent+1}=1) \, = \,  p_{L(3)}(\gt 1, l, 1)\, l d (1-ld ) + \eta_{1,3}(l)\\
-d^3  p_{L(3)}(\gt1, l, 1)\, l  \, \leq \, \eta_{1,3}(l) \, \leq \, d^3 p_{L(3)}(\gt1, l, 1)\, l \binom{l}{2}
\end{align*}
and

\begin{align*}
\prob(B_{11}, L_P=l, L_{\parent-1}=1, L_{\parent+1}>1) \, = \,  p_{L(3)}(1, l, \gt1)\, l d (1-ld ) + \eta_{1,2}(l)\\
-d^3  p_{L(3)}(1, l, \gt1)\, l  \, \leq \, \eta_{1,2}(l) \, \leq \, d^3 p_{L(3)}(1, l, \gt1)\, l \binom{l}{2}
\end{align*}

Combining, we arrive at the following contribution of $B_{11}$ to $\lim_{n\rightarrow \infty} H(\widehat{D}^n|X^n, \widehat{Y}, \widehat{K})/n$:
\begin{align}
T_{11} &= \frac{1}{\mu(\sX)} \sum_{l=2}^\infty \prob(B_{11}, L_P=l) \log l \nonumber\\
& = \frac{d}{\mu(\sX)} \; \sum_{l=2}^\infty \Big \{
p_{L(3)}(\gt1, l, \gt1) \; l \log l \big( 1 - (l-1)d\,\big)
+ p_{L(3)}(1, l, 1) \; l \log l \big( 1 - (l+1)d\,\big) +    \nonumber\\
&\phantom{= \frac{d}{\mu(\sX)} \Big ( \sum_{l=2}^\infty \ \,} \big( p_{L(3)}(1, l, \gt1)+p_{L(3)}(\gt1, l, 1)\big) \; l \log l \big( 1 - ld\,\big)  \Big \} + \delta_{11}
\label{eq:T11}
\end{align}
with
\begin{align}
  -\frac{2d^3}{\mu(\sX)} \sum_{l=2}^\infty p_L(l)\, l \log l &\leq
\delta_{11} = \delta_{11} (d, \sX) \leq
\frac{d^3}{\mu(\sX)} \sum_{l=2}^\infty p_L(l)\, l \binom{l+1}{2}\log l
\label{eq:delta11_detail}
\end{align}
We have normalized by $\mu(\sX)$ to move from a per run contribution to a per bit contribution.

It is easy to infer
\begin{align}
-d^3\E[L^3 \log L]\leq \delta_{11} \leq d^3\E[L^3 \log L]
\label{eq:delta11_crude}
\end{align}
from Eq.~\eqref{eq:delta11_detail}.

\item{\bf Two deletions in $R_{P}$:}\\
Consider $B_{12}$. If $L_P=l>2$ then entropy contribution is $\log \binom{l}{2}$. We have, for $l>2$,
\begin{align*}
\prob(B_2, L_P=l) = p_L(l) \binom{l}{2} d^2  (1-d)^{l-2}\,\cdot \,\prob(R_{P-1} \textup{ and } R_{P+1} \textup{ do not disappear under $\sD$})
\end{align*}

It follows that
\begin{align*}
p_L(l) \binom{l}{2} d^2 (1-d)^l \leq \prob(B_2, L_P=l) \leq p_L(l) \binom{l}{2} d^2 (1-d)^{l-2}
\end{align*}
leading to
\begin{align*}
\prob(B_2, L_P=l) = p_L(l) \binom{l}{2} d^2 + \eta_2\\
-d^3 p_L(l)l \binom{l}{2} \leq \eta_2 \leq 0
\end{align*}

Combining, we arrive at the following contribution to $\lim_{n\rightarrow \infty} H(\widehat{D}^n|X^n, \widehat{Y}, \widehat{K})/n$:
\begin{align}
\label{eq:T12}
T_{12} &= \frac{1}{\mu(\sX)} \sum_{l=3}^\infty \prob(B_2, L_P=l) \log \binom{l}{2} \nonumber\\
& = \frac{d^2}{\mu(\sX)} \sum_{l=3}^\infty p_L(l) \binom{l}{2}  \log \binom{l}{2} +\delta_{12}
\end{align}
with
\begin{align}
-d^3 \E[L^3 \log L]
\leq -\frac{d^3}{\mu(\sX)} \sum_{l=3}^{\infty} p_L(l)l \binom{l}{2} \log \binom{l}{2}  \leq
\delta_{12} = \delta_{12} (d, \sX) \leq
0
\label{eq:delta12_bound}
\end{align}
\end{itemize}

Plugging Eqs.~\eqref{eq:T11} and \eqref{eq:T12} into Eq.~\eqref{eq:T1_split}, we obtain
our desired estimate on the contribution $T_1$ of the event $B_1$,
\begin{align*}
T_1 =& \;
\frac{d}{\mu(\sX)}  \; \sum_{l=2}^\infty \Big \{
p_{L(3)}(\gt1, l, \gt1) \; l \log l \big( 1 - (l-1)d\,\big)
+ p_{L(3)}(1, l, 1) \; l \log l \big( 1 - (l+1)d\,\big) +    \nonumber\\
&\phantom{\ \frac{d}{\mu(\sX)} \Big ( \sum_{l=2}^\infty } \big( p_{L(3)}(1, l, \gt1)+p_{L(3)}(\gt1, l, 1)\big) \; l \log l \big( 1 - ld\,\big)  \Big \} \; \\
& \; + \; \frac{d^2}{\mu(\sX)} \; \sum_{l=3}^\infty p_L(l) \binom{l}{2}  \log \binom{l}{2}
+ \delta_1 \, ,
\end{align*}
where $\delta_1= \delta_{1E}+\delta_{11}+\delta_{12}$ is bounded using Eqs.~\eqref{eq:delta1E_bound}, \eqref{eq:delta11_crude} and \eqref{eq:delta12_bound} as
\begin{align}
\label{eq:delta1_bounds}
- 2 d^3 \E[L^3 \log L]
\leq
\delta_1
\leq
 7 d^3 \E[L^3 \log L] \, .
\end{align}

\item {\bf From a combination of three parent runs}\\
Define
\begin{align*}
B_3 \equiv\; &R_\parent \textup{ and } R_{\parent+2} \textup{ suffer at least one deletion in total under $\hsD$ and } \nonumber \\
&\exists j \textup{ s.t. } \hX(j)=(R_P\, R_{P+1}\, R_{P+2})
\end{align*}
We are interested in the contribution due to occurrence of event $B_3$.

Again, we will restrict attention to a subset of $B_3$ and the prove that we are missing a very small contribution.
Define
\begin{align*}
E_3 \equiv B_3 \cap \{\textup{$R_{P-1}$ and $R_{P+3}$ do not disappear under $\sD$.}\}
\end{align*}
Similar to our analysis for Case 1, we can show that
\begin{align*}
0\leq \prob(B_3 \backslash E_3) < d^3 (L_{P-1}^3+L_P+L_{P+2}+L_{P+1}^3) \, .
\end{align*}
The largest possible value of $H(\widehat{D}(j)|\hX (j), \widehat{Y}(j))$ for a particular occurrence of $B_3 \backslash E_3$ is
\begin{align*}
\max_{i=1,2,3,4} \; \log \binom{L_P+L_{P+2}}{i} \leq 4 \log (L_P+L_{P+2})
 \end{align*}
 since $R_P$ and $R_{P+2}$ can suffer at most 4 deletions in total under $\hsD$. Thus, the additive error introduced by restricting to $E_3$ in our estimate of
$\lim_{n\to\infty} \frac{1}{n}H(\widehat{D}^n|X^n, \widehat{Y}, \widehat{K})$ is
\begin{align}
\label{eq:delta1E_bound_B}
0 \leq \delta_{3E}(d, \sX) \leq d^3 \E[4 (L_{P-1}^3+L_P+L_{P+2}+L_{P+1}^3) \log (L_P+L_{P+2})]
\end{align}
Now, $\log (L_P+L_{P+2}) \leq \log (2 L_PL_{P+2})= 1 + \log L_P + \log L_{P+2}$.
From Proposition \ref{propo:product_of_increasing}, $\E[L_{P-1}^3 \log L_P] \leq E[L^3 \log L]$, also
 $\E[L_P \log L_{P+2}] \leq \E[L \log L]$,  and so on. Plugging into Eq.~\eqref{eq:delta1E_bound_B}, we arrive at
\begin{align}
0 \leq \delta_{3E}(d, \sX) \leq d^3 \E[ 16 L^3 + 32 L^3 \log L]
\label{eq:delta3E_bound}
\end{align}

Now, we further restrict to a subset of $E_3$. Define
\begin{align*}
B_{31} = E_3 \cap \{\textup{One deletion in total in } R_P, R_{P+2}\} \cap
\{L_{P+1}=1\}
\end{align*}

Consider the event $E_3 \backslash B_{31}$. This can occur due to one of the following:
\begin{itemize}
\item More than one deletion in $R_P, R_{P+2}$: This occurs with probability at most $\binom{L_P+L_{P+2}}{2}d^3$ (since we also need $R_{P+1}$ to disappear).
\item $L_{P+1}>1$: Now the probability that $R_{P+1}$ disappears is at most $d^2$.  Thus, the probability of $\prob(E_3 \cap \{L_{P+1}>1\}) \leq (L_P+L_{P+2})d^3$.
\end{itemize}
It follows from union bound that $\prob(E_3 \backslash B_{31}) \leq d^3 (L_P+L_{P+2})^2$.
As before, the largest possible value of
$H(\widehat{D}(j)|\hX (j), \widehat{Y}(j))$ for a particular occurrence of $E_3 \backslash B_{31}$ is $4 \log (L_P+L_{P+2})$. Thus, the additive error introduced by restricting to $B_{31}$ in estimating the contribution of $E_3$ is
\begin{align*}
0 \leq \delta_{32} \leq 4 d^3 (L_P+L_{P+2})^2 \log (L_P+L_{P+2})
\end{align*}
Now, we use $\log (L_P+L_{P+2}) \leq 1 + \log L_P + \log L_{P+2}$ and Proposition \ref{propo:product_of_increasing} to obtain
\begin{align}
0 \leq \delta_{32} \leq  d^3 \E [16 L^2+ 32 L^2 \log L]
\label{eq:delta32_bound}
\end{align}
Denoting by $T_{31}$ the contribution of $B_{31}$, and $T_3$ the contribution of $B_3$, we have
\begin{align}
T_3 = T_{31} + \delta_{3E} + \delta_{32}
\label{eq:T3_pieces}
\end{align}

We consider two cases in estimating $T_{31}$:
\begin{itemize}
\item $L_P>1$\\
The value of
$H(\widehat{D}(j)|\hX (j), \widehat{Y}(j))$ for a particular occurrence is $\log(L_P + L_{\parent+2})$.
We have
\begin{align*}
\prob(B_{31}, L_P=l_0, |R_{\parent+2}|=l_2) = d^2 p_{L(3)}(l_0,1,l_2) (l_0+l_2)  + \eta_{3,1} \\
-d^3 p_{L(3)}(l_0,1,l_2) (l_0+l_2)^2  \leq  \eta_{3,1} \leq 0
\end{align*}
\item $L_{\parent}=1$\\
The value of
$H(\widehat{D}(j)|\hX (j), \widehat{Y}(j))$ for a particular occurrence is $\log L_{\parent+2}$ since $R_{\parent}$ should not disappear. We have
\begin{align*}
\prob(B_3, L_P=1, L_{\parent+2}=l_2) = d^2 p_{L(3)}(1,1,l_2) l_2  + \eta_{3,2} \\
-d^3 p_{L(3)}(1,1,l_2) l_2^2  \leq  \eta_{3,2} \leq 0
\end{align*}
\end{itemize}
Combining the two cases, we arrive at the following estimate:
\begin{align}
T_{31} &= \frac{1}{\mu(\sX)} \sum_{l=3}^\infty \prob(B_3, L_P=l_0, |R_{\parent+2}|=l_2) \,\log \big(\,l_2 + l_0 \mathbb{I}(l_0>1)\, \big) \nonumber\\
&= \frac{d^2}{\mu(\sX)} \left (\sum_{l_0 >1, l_2} p_{L(3)}(l_0, 1, l_2)\, (l_0+l_2) \log (l_0+l_2)
+ \sum_{l_2} p_{L(3)}(1, 1, l_2) \,l_2 \log l_2 \right ) +\delta_{31}
\label{eq:T31_estimate}
\end{align}
where
\begin{align*}
-\frac{d^3}{\mu(\sX)} \sum_{l_0, l_2} p_{L(3)} (l_0, 1, l_2)\, (l_0+l_2)^2 \log (l_0+l_2) &\leq
\delta_{31} = \delta_{31} (d, \sX) \leq 0
\end{align*}
Again, we use $\log (L_P+L_{P+2}) \leq 1 + \log L_P + \log L_{P+2}$ and Proposition \ref{propo:product_of_increasing} to obtain
\begin{align}
-{d^3} \E[4L^2 + 8 L^2 \log L] &\leq
\delta_{31} = \delta_{31} (d, \sX) \leq 0
\label{eq:delta31_bound}
\end{align}

Finally, we plug Eq.~\eqref{eq:T31_estimate} into Eq.~\eqref{eq:T3_pieces} to obtain
\begin{align*}
T_3 = \frac{d^2}{\mu(\sX)} \left (\sum_{l_0 >1, l_2} p_{L(3)}(l_0, 1, l_2)\, (l_0+l_2) \log (l_0+l_2)
+ \sum_{l_2} p_{L(3)}(1, 1, l_2) \,l_2 \log l_2 \right ) +\delta_{3}
\end{align*}
where $\delta_3 = \delta_{3E}+ \delta_{32}+\delta_{31}$. Using Eqs.~\eqref{eq:delta3E_bound}, \eqref{eq:delta32_bound} and \eqref{eq:delta31_bound}, we obtain
\begin{align}
\label{eq:delta3_bounds}
- d^3 \E[4L^2 + 8 L^2 \log L] \leq \delta_3 \leq d^3 \E[32 L^3 + 64 L^3 \log L]
\end{align}

\item {\bf From a combination of five parent runs}\\
Define
\begin{align*}
B_5 \equiv\;& R_P, R_{P+2}, R_{P+4} \textup{ suffer at least one deletion in total under $\hsD$ and }\nonumber \\ &\exists j \textup{ s.t. } \hX(j)= (R_P R_{P+1} R_{P+2} R_{P+3} R_{P+4})
\end{align*}
We have $\prob(B_5) \leq d^3 (L_{P}+L_{P+2}+L_{P+4})$ since $R_{P+1}$ and $R_{P+3}$ must disappear. Also, the
largest possible value of
$H(\widehat{D}(j)|\hX (j), \widehat{Y}(j))$ for a particular occurrence is
\begin{align*}
\max_{i=1,2,\ldots,6} \; \log \binom{L_P+L_{P+2}+L_{P+4}}{i} \leq
6 \log (L_P+L_{P+2}+L_{P+4})
\end{align*}
since each run can suffer at most two deletions under $\hsD$.
Thus, the contribution of $B_5$ is $\delta_5$, where
\begin{align}
\label{eq:delta5_bounds}
0 \leq \delta_5 \leq 6 d^3 \E[(L_P+L_{P+2}+L_{P+4}) \log (L_P+L_{P+2}+L_{P+4})] \leq
d^3 \E[36 L + 54 L \log L]
\end{align}
where we have used $\log (L_P+L_{P+2}+L_{P+4}) \leq 2 + \log L_{P} + \log L_{P+2}+ \log L_{P+4}$ and Proposition \ref{propo:product_of_increasing}.

\item{\bf From a combination of $2k+1$ parent runs for $k \geq 3$}\\
Define
\begin{align*}
B_{2k+1} \equiv\;& \exists j \textup{ s.t. } \hX(j)= (R_P R_{P+1} \ldots  R_{P+2k})
\end{align*}
We need $k$ runs to disappear, and this occurs with probability at most $d^k$. The
largest possible value of
$H(\widehat{D}(j)|\hX (j), \widehat{Y}(j))$ for a particular occurrence is
$2(k+1) \log (L_P+L_{P+2}+\ldots+L_{P+2k}) \leq 2(k+1) \log ((k+1)/d)$ since no run has
length exceeding $1/d$. Thus, the contribution of $B_{2k+1}$ is bounded above by $d^k 2(k+1) \log ((k+1)/d)$. Summing we find that the overall contribution $T_{{\rm gt}5}$ of $B_7, B_9, \ldots$ is bounded as
\begin{align}
\label{eq:Tgt5_bounds}
0\leq T_{{\rm gt}5} \leq \sum_{k=3}^\infty d^k 2(k+1) \log ((k+1)/d) \leq 10 d^3 \log(1/d)
\end{align}
for small enough $d$.
\end{enumerate}
Finally, we obtain
\begin{align*}
\lim_{n\to\infty} & \frac{1}{n}H(\widehat{D}^n|X^n, \widehat{Y}, \widehat{K}) =
T_1+ T_3 + T_5 + T_{{\rm gt}5} \nonumber\\
&=\frac{d}{\mu(\sX)}  \sum_{l=2}^\infty \Big \{ p_{L(3)}(\gt1, l, \gt1) \; l \log l \big( 1 - (l-1)d\,\big)
+  \big( p_{L(3)}(1, l, \gt1)+p_{L(3)}(\gt1, l, 1)\big) \; l \log l \big( 1 - ld\,\big) \nonumber\\
&\ \ \ \ \ \ \ \ \ +p_{L(3)}(1, l, 1) \; l \log l \big( 1 - (l+1)d\,\big)\Big \} \nonumber\\
&+ \frac{d^2}{\mu(\sX)} \sum_{l=3}^\infty p_L(l)\binom{l}{2} \log \binom{l}{2}\nonumber\\
&+ \frac{d^2}{\mu(\sX)} \left (\sum_{l_0 >1, l_2} p_{L(3)}(l_0, 1, l_2)\, (l_0+l_2) \log (l_0+l_2)
+ \sum_{1,1,l_2} p_{L(3)}(1, 1, l_2) \,l_2 \log l_2 \right ) +\delta
\end{align*}
where $\delta= \delta_1+\delta_3+\delta_5+T_{{\rm gt}5}$. Rearranging gives Eq.~\eqref{eq:modified_deletion_entropy_estimate}, whereas Eq.~\eqref{eq:modified_deletion_estimate_error} follows for small enough $d$ from Eqs.~\eqref{eq:delta1_bounds},
\eqref{eq:delta3_bounds}, \eqref{eq:delta5_bounds} and \eqref{eq:Tgt5_bounds} and the fact that no run has length exceeding $1/d$.
\end{proof}

\begin{proof}[Proof of Corollary \ref{coro:hatD_givenxy}]
We prove the corollary assuming $H(\sY) > 1 -d^{\gamma}$. The proof assuming $H(\sX) > 1 -d^{\gamma}$ is analogous.

It follows from Fact \ref{fact:ELcube_is_small} that if $H(\sY) \geq 1-d^{\gamma}$,
then $\delta$ (cf. Eq.~\eqref{eq:modified_deletion_estimate_error}) is bounded as $|\delta|<\const_1 d^{1+\gamma} \log(1/d) \leq d^{1+\gamma-\eps/2}$ for small enough $d$, for some $\const_1< \infty$.

Consider $\sum_{l=2}^\infty p_L(l) l^2 \log l$. We separately analyze the first $l_0=\lfloor 4 \log(1/d) \rfloor$ terms of the sum. We use Lemma \ref{lemma:xy_similar_runs_gammaSTRONG}(i) (Eq.~\eqref{eq:L_Vstrong_fromY}) to deduce that
\begin{align}
\label{eq:sample_sum_first_l0}
\sum_{l=2}^{l_0} p_L(l) l^2 \log l=
\sum_{l=2}^{\infty} p_L^*(l) l^2 \log l + \xi_1 \, ,\\
 \textup{with} \quad |\xi_1| \leq \const_4 d^{\gamma/2-\eps/4} (l_0)^3  \leq \const_5 d^{\gamma/2-\eps/2} \, ,\nonumber
\end{align}
for small enough $d$. Next, we use Lemma
\ref{lemma:pL_tail_control_fromHy} to deduce that
\begin{align}
\label{eq:sample_sum_remaining}
\sum_{l=l_0+1}^{\infty} p_L(l) l^2 \log l = \sum_{l=l_0+1}^{\lfloor 1/d \rfloor} p_L(l) l^2 \log l \leq
\const_6 d^{\gamma} (1/d) \log (1/d) \leq \const_7 d^{\gamma-\eps/2}
\end{align}
for small enough $d$. Finally, Lemma \ref{lemma:xy_similar_runs_gammaSTRONG}(ii) tells us that \begin{align*}
|\mu(\sX) - 2| \leq \const_3 d^{\gamma/2}
\end{align*}
Combining with Eqs.~\eqref{eq:sample_sum_first_l0} and \eqref{eq:sample_sum_remaining}, it
follows that
\begin{align*}
\frac{d^2}{\mu(\sX)}\sum_{l=2}^\infty p_L(l) l^2 \log l = \frac{d^2}{2} \left \{\sum_{l=2}^{\infty} p_L^*(l) l^2 \log l \right \} + \eta_2
\end{align*}
where $|\eta_2|\leq \const_8 d^{2+\gamma/2-\eps/2} \leq \const_8 d^{1+\gamma-\eps/2}$, for small enough $d$.

Other terms in Eq.~\eqref{eq:modified_deletion_entropy_estimate} can be similarly analyzed. The result follows.
\end{proof}

\begin{proof}[Proof of Corollary \ref{coro:hygivenx_q}]
We prove the corollary assuming $H(\sY) > 1 -d^{\gamma}$. The proof assuming $H(\sX) > 1 -d^{\gamma}$ is analogous.

By definition, $D^n$ is independent of $X^n$, so $H(D^n)=H(D^n|X^n)=n h(d)$, where $h(\cdot)$ is the binary entropy function.
We have, for $Y= Y(X^n)$,
\begin{align*}
H(Y, K|X^n)&= H(D^n|X^n) - H(D^n|X^n, Y, K)\\
&= nh(d) -  H(\widehat{D}^n|X^n, \widehat{Y}, \widehat{K}) + n \delta_1
\end{align*}
with $|\delta_1(d, \sX) |\leq 2 H(Z^n)/n \rightarrow 2 h(z)$. It follows from Corollary \ref{coro:hatD_givenxy}, with $\gamma=2-\eps/2$, that
\begin{align}
\lim_{n \rightarrow \infty} \frac{1}{n}H(Y(X^n), K(X^n)|X^n) = h(d) - \frac{d}{\mu(\sX)}  \sum_{l=2}^\infty p_{L}(l) \; l \log l
 -  d^2 c_3  + \delta_2
 \label{eq:YKgivenX_prelim}
\end{align}
with $|\delta_2|\leq 2h(z) + \const_1 d^{3-\eps}$. From Proposition \ref{propo:z_is_small}, we know that
$z< \const_1 d^{3-\eps/2}$. It follows that $h(z) \leq \const_2 d^{3-\eps}$ and hence $|\delta_2|\leq \const_3 d^{3-\eps}$.
Simple calculus gives
\begin{align}
h(d) = d \log (1/d) + (d - d^2/2) /\ln 2 + \delta_3
\label{eq:h_estimate}
\end{align}
$|\delta_3| \leq \const_4 d^3$. Using Lemma \ref{lemma:xy_similar_runs_gammaSTRONG}(i) (Eq.~\eqref{eq:xy_similar_runs_gammaSTRONG})
and Lemma \ref{lemma:pL_tail_control_fromHy},
we obtain
\begin{align}
\sum_{l=2}^\infty p_{L}(l) \; l \log l = \sum_{l=2}^\ell q_{L}(l) \; l \log l + \delta_4
\label{eq:llogl_p_q}
\end{align}
where $|\delta_4| \leq \const_5 d^{2-\eps}$ for small enough $d$.
Using Lemma \ref{lemma:xy_similar_runs_gammaSTRONG}(ii)(Eq.~\eqref{eq:muX_close2_muY_using_Hy})
and $\mu(\sX) > 1$ (from Lemma \ref{lemma:mean_closeto2}), we obtain
\begin{align}
\left |\frac{1}{\mu(\sX)} - \frac{1}{\mu(\sY)} \right | \leq \const_6 d^{2-\eps}
\label{eq:muXinv_muYinv_difference}
\end{align}
Also, it follows from $|\mu(\sY) -2| \leq 7 d^{1-\eps/4}$ (Lemma \ref{lemma:mean_closeto2} applied to $\sY$) and elementary calculus that
\begin{align}
\{\mu(\sY)\}^{-1} &= 1- \frac{1}{4} \mu(\sY) + \delta_5\nonumber\\
&= 1- \frac{1}{4} \sum_{l=1}^\ell q_L(l) l + \delta_6
\label{eq:muY_inverse_estimate}
\end{align}
where $|\delta_6|\leq \const_7 d^{2-\eps}$. Here we have used Lemma \ref{lemma:L_tail_control} (applied to $\sY$) to bound $\sum_{l=\ell+1}^\infty q_L(l) l$.

Plugging Eqs.~\eqref{eq:h_estimate}, \eqref{eq:llogl_p_q}, \eqref{eq:muXinv_muYinv_difference} and \eqref{eq:muY_inverse_estimate} into Eq.~\eqref{eq:YKgivenX_prelim}, we obtain the result.

\end{proof}

\bibliographystyle{IEEEtran}

\begin{thebibliography}{99}


\bibitem{MitzenmacherReview} M.~Mitzenmacher, ``A survey of results for
deletion channels and related synchronization channels,''
Probab. Surveys, 6 (2009), 1-33.


\bibitem{Drinea07lb} E.~Drinea and M.~Mitzenmacher,
``Improved lower bounds for the capacity of i.i.d. deletion and
duplication channels,'' IEEE Trans. Inform. Theory, 53 (2007) 2693-2714.

\bibitem{Dobrushin} R.~L.~Dobrushin,
``Shannon's Theorems for Channels with Synchronization Errors,''
Problemy Peredachi Informatsii, 3 (1967), 18-36.

\bibitem{KirschDrinea} A.~Kirsch and E.~Drinea,
``Directly Lower Bounding the Information Capacity for Channels with I.I.D.
Deletions and Duplications,''
Proc. of IEEE  Intl. Symp. on Inform. Theory (ISIT) 2007.


\bibitem{Drinea06dnear1} E.~Drinea and M.~Mitzenmacher,
``A Simple Lower Bound for the Capacity of the
Deletion Channel,'' IEEE Trans. Inform. Theory, 52:10 (2006), 4657–4660.

\bibitem{Fertonani09} D.~Fertonani and T.M.~Duman,
``Novel bounds on the capacity of binary channels with deletions
and substitutions,''
Proc. of IEEE  Intl. Symp. on Inform. Theory (ISIT) 2009.

\bibitem{Dalai} M.~Dalai, ``A new bound for the capacity of the deletion channel with
               high deletion probabilities", arXiv:1004.0400, 2010.

\bibitem{DMP07upperbounds} S.~Diggavi, M.~Mitzenmacher, and H.~Pfister, ``Capacity Upper Bounds for Deletion Channels,"
Proc. of IEEE  Intl. Symp. on Inform. Theory (ISIT) 2007.

\bibitem{KanoriaMontanari} Y.~Kanoria and A.~Montanari, ``On the deletion
channel with small deletion probability,''
Proc. of IEEE  Intl. Symp. on Inform. Theory (ISIT) 2010.

\bibitem{KMS} A.~Kalai, M.~Mitzenmacher and M.~Sudan,
``Tight Asymptotic Bounds for the Deletion Channel
with Small Deletion Probabilities'', Proc. of IEEE  Intl. Symp. on Inform. Theory (ISIT) 2010.


\bibitem{PP1} D.~J.~Daley and D.~Vere-Jones, \emph{An Introduction to
    the Theoryy of Point Processes}, Springer, New York, 2008.

\bibitem{PP2} F.~Baccelli and P.~Br\'emaud, \emph{Elements of Queuing Theory}, Springer, New York, 2003.





\end{thebibliography}

\end{document}